\documentclass[11pt,twoside]{article}
\usepackage{microtype}
\usepackage{algpseudocode,algorithm,epstopdf,multirow}
\usepackage{amssymb,amsmath,latexsym,graphicx,hyperref,epstopdf,color}
\usepackage{epstopdf,caption,subcaption}
\usepackage{chngcntr}
\usepackage[capitalise,noabbrev]{cleveref}
\usepackage{tikz}
\usepackage{pgfgantt}
\usetikzlibrary{shapes,arrows,fit,calc,positioning}
\usetikzlibrary{decorations.pathreplacing}

\normalsize
\newtheorem{theorem}{Theorem}
\newtheorem{lemma}{Lemma}
\newtheorem{corollary}{Corollary}
\newtheorem{proposition}{Proposition}[subsubsection]

\newtheorem{conjecture}{Conjecture}

\newenvironment{proof}{{\sc Proof. }}{\hfill$\Box$\vspace{0.2in}}
\algnewcommand{\LineComment}[1]{\Statex \(\triangleright\) #1}

\title{Acyclic edge coloring conjecture is true on planar graphs without intersecting triangles%
\footnote{An extended abstract appears in {\em Proceedings of TAMC 2020}.}}
\author{Qiaojun~Shu\thanks{Department of Mathematics, Hangzhou Dianzi University.  Hangzhou, China.
	Email: {\tt qjshu@hdu.edu.cn}}
	\thanks{Department of Computing Science, University of Alberta.  Edmonton, Alberta T6G 2E8, Canada.
	Emails: {\tt \{qiaojun,guohui\}@ualberta.ca}}
\and
	Guohui~Lin$^{\ddagger}$\thanks{Correspondence author. ORCID: 0000-0003-4283-3396.}
\and
	Eiji~Miyano\thanks{Department of Artificial Intelligence, Kyushu Institute of Technology. Iizuka, Japan.
	Email: {\tt miyano@ces.kyutech.ac.jp}}}
\date{\today}
\begin{document}
\maketitle
\begin{abstract}
An acyclic edge coloring of a graph $G$ is a proper edge coloring such that no bichromatic cycles are produced.
The acyclic edge coloring conjecture by Fiam{\v{c}}ik (1978) and Alon, Sudakov and Zaks (2001) states that
every simple graph with maximum degree $\Delta$ is acyclically edge $(\Delta + 2)$-colorable.
Despite many milestones, the conjecture remains open even for planar graphs.
In this paper, we confirm affirmatively the conjecture on planar graphs without intersecting triangles.
We do so by first showing, by discharging methods,
that every planar graph without intersecting triangles must have at least one of the six specified groups of local structures,
and then proving the conjecture by recoloring certain edges in each such local structure and by induction on the number of edges in the graph.

\end{abstract}

\section{Introduction}
Let $G$ be a simple graph with vertex set $V(G)$ and edge set $E(G)$.
For an integer $k \ge 2$, a {\em (proper) edge $k$-coloring} is a mapping $c: E(G) \to \{1, 2, \ldots, k\}$ such that any two adjacent edges receive different colors.
(We drop ``proper'' in the sequel.)
$G$ is {\em edge $k$-colorable} if $G$ has an edge $k$-coloring.
The {\em chromatic index} $\chi'(G)$ of $G$ is the smallest integer $k$ such that $G$ is edge $k$-colorable.
An edge $k$-coloring $c$ of $G$ is called {\em acyclic} if there are no bichromatic cycles in $G$, i.e.,
the subgraph of $G$ induced by any two colors is a forest.
The {\em acyclic chromatic index} of $G$, denoted by $a'(G)$, is the smallest integer $k$ such that $G$ is acyclically edge $k$-colorable.

Recall that a connected graph has a block-cut tree representation
in which each block is a $2$-connected component and two blocks overlap at at most one cut-vertex.
Since we may swap any two colors inside a block, if necessary,
in the context of edge coloring or acyclic edge coloring we assume w.l.o.g. that $G$ is $2$-connected.
In the sequel, only simple and $2$-connected graphs are considered in this paper.

Let $\Delta(G)$ ($\Delta$ for short and reserved) denote the maximum degree of the graph $G$.
One sees that $\Delta \le \chi'(G) \le a'(G)$.
Note that $\chi'(G) \le \Delta + 1$ by Vizing's theorem \cite{Vizing64} and that $a'(K_4) = \Delta + 2$.
Fiam{\v{c}}ik \cite{F78} and Alon, Sudakov and Zaks \cite{ABZ01} independently made the following {\em acyclic edge coloring conjecture} (AECC):

\begin{conjecture}\label{conj01} {\rm (AECC)}
For any graph $G$, $a'(G) \le \Delta + 2$.
\end{conjecture}

For an arbitrary graph $G$, the following milestones have been achieved:
Alon, McDiarmid and Reed \cite{AMR91} proved that $a'(G) \le 64 \Delta$ by a probabilistic argument.
The upper bound was improved to $16 \Delta$ \cite{MR98}, to $\lceil 9.62 (\Delta - 1) \rceil$ \cite{NPS12}, to $4 \Delta - 4$ \cite{EP12},
and most recently in 2017 to $\lceil 3.74 (\Delta - 1) \rceil + 1$ by Giotis et al.~\cite{GKP17} using the Lov\'{a}sz local lemma.
On the other hand, the AECC has been confirmed true for graphs with $\Delta \in \{3, 4\}$ \cite{Sku04,AMM12,BC09,SWMW19,WMSW19}.

When $G$ is planar, i.e., $G$ can be drawn in the two-dimensional plane so that its edges intersect only at their ending vertices,
Basavaraju et al.~\cite{BLSNHT11} showed that $a'(G) \le \Delta + 12$.
The upper bound was improved to $\Delta + 7$ by Wang, Shu and Wang~\cite{WSW127} and to $\Delta + 6$ by Wang and Zhang~\cite{WZ16}.
The AECC has been confirmed true for planar graphs without $i$-cycles for each $i \in \{3, 4, 5, 6\}$ in \cite{SW11,WSW4,SWW12,WSWZ14}, respectively.

A {\em triangle} is synonymous with a $3$-cycle.
We say that two triangles are {\em adjacent} if they share a common edge, and are {\em intersecting} if they share at least a common vertex.
Recall that the truth of the AECC for planar graphs without triangles has been verified in \cite{SW11}.
When a planar graph $G$ contains triangles but no intersecting triangles, Hou, Roussel and Wu \cite{HRW11} proved the upper bound $a'(G) \le \Delta + 5$,
and Wang and Zhang \cite{WZ14} improved it to $a'(G) \le \Delta + 3$.
This paper focus on planar graphs without intersecting triangles too,
and we completely resolve the AECC by showing the following main theorem.

\begin{theorem}\label{thm01}
The AECC is true for planar graphs without intersecting triangles.
\end{theorem}

The rest of the paper is organized as follows.
In Section 2, we characterize six groups of local structures (also called {\em configurations}),
and by discharging methods we prove that any planar graph without intersecting triangles must contain at least one of these local structures.
Incorporating a known property of edge colorings and bichromatic cycles (Lemma~\ref{lemma09}),
in Section 3 we prove by induction on the number of edges that the graph admits an acyclic edge $(\Delta + 2)$-coloring.
We conclude the paper in Section 4.

\section{The six groups of local structures}
Recall that we consider only simple $2$-connected planar graphs.

Given a graph $G$, let $d(v)$ denote the degree of the vertex $v$ in $G$.
A vertex of degree $k$ (at least $k$, at most $k$, respectively) is called a $k$-{\em vertex} ({$k^+$-{\em vertex}, $k^-$-{\em vertex}, respectively).
Let $n_k(v)$ ($n_{k^{+}}(v)$, $n_{k^{-}}(v)$, respectively) denote
the number of $k$-vertices ($k^{+}$-vertices, $k^{-}$-vertices, respectively) adjacent to $v$ in $G$.

\begin{figure}[]
\begin{center}
\includegraphics[width=6.0in]{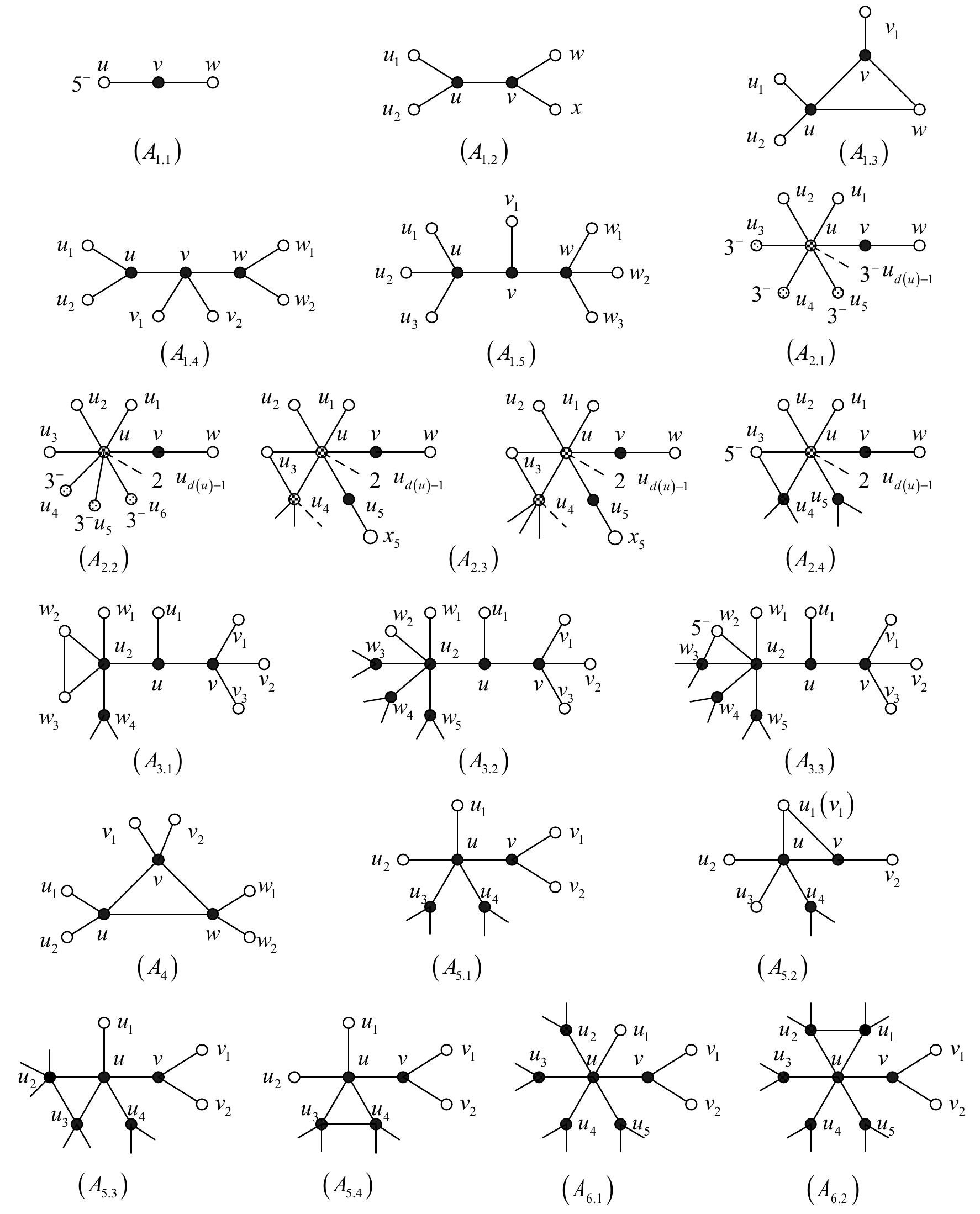}
\medskip
\caption{The six groups of local structures ($A_1$)--($A_6$) specified in Theorem~\ref{thm02}.
	In each local structure, filled vertices are distinct from each other, each has its degree determined, and some or all its neighbors determined;
	a non-filled vertex could collide into another vertex, has an undetermined degree, but an upper-bound on its degree, if known, would be shown.
	Vertex labels start with a character.
	A dashed edge indicates existence of zero, one, or more edges incident at the vertex.\label{fig01}}
\end{center}
\end{figure}

\begin{theorem}\label{thm02}
Let $G$ be a simple $2$-connected planar graph with $\Delta \ge 5$ and without intersecting triangles.
Then $G$ contains at least one of the following local structures (or {\em configurations}, used interchangeably) {\rm ($A_1$)--($A_6$)},
as depicted in Figure~\ref{fig01}:
\begin{description}
\parskip=0pt
\item[{\rm ($A_1$)}]
	A path $u v w$ such that at least one of the following holds:
	\begin{description}
	\parskip=0pt
	\item[{\rm ($A_{1.1}$)}]
		$d(v)= 2$ and $d(u)\le 5$;
	\item[{\rm ($A_{1.2}$)}]
		$d(u)= d(v)= 3$;
	\item[{\rm ($A_{1.3}$)}]
		$d(v)= 3$, $d(u)= 4$, and $uw\in E(G)$;
	\item[{\rm ($A_{1.4}$)}]
		$d(v)= 4$ and $d(u)= d(w)= 3$;
	\item[{\rm ($A_{1.5}$)}]
		$d(v)= 3$ and $d(u)= d(w)= 4$.
	\end{description}
\item[{\rm ($A_2$)}]
	A $6^+$-vertex $u$ adjacent to a $2$-vertex $v$.
	Let $u_1, u_2, \ldots, u_{d(u)-1}$ be the other neighbors of $u$.
	Then at least one of the following cases holds:
	\begin{description}
	\parskip=0pt
	\item[{\rm ($A_{2.1}$)}]
		$n_2(u) + n_3(u) \ge d(u)-2$;
	\item[{\rm ($A_{2.2}$)}]
		$n_2(u) + n_3(u) = d(u)-3$ and $n_3(u) \le 3$;
	\item[{\rm ($A_{2.3}$)}]
		$n_2(u) = d(u)- 4$, $d(u_j) = 2$ for $5 \le j \le d(u)-1$, $n_2(u_4) \in \{d(u_4)- 4, d(u_4)- 5\}$, and $u_3u_4 \in E(G)$;
	\item[{\rm ($A_{2.4}$)}]
		$n_2(u) = d(u)- 5$, $d(u_j) = 2$ for $6 \le j \le d(u)-1$, $d(u_5)= 3$, $d(u_4)= 4$, $d(u_3) \le 5$, and $u_3u_4\in E(G)$.
	\end{description}
\item[{\rm ($A_3$)}]
	A path $u_2 u v$ with $d(u) = 3$, $d(v) = 4$.
	Let $w_1, w_2, \ldots, w_{d(u_2)-1}$ be the neighbors of $u_2$ other than $u$, sorted by their degrees.
	Then at least one of the following cases holds:
	\begin{description}
	\parskip=0pt
	\item[{\rm ($A_{3.1}$)}]
		$d(u_2) = 5$, $d(w_4)= 3$, and $w_2 w_3\in E(G)$;
	\item[{\rm ($A_{3.2}$)}]
		$d(u_2) = 6$, $d(w_3)= d(w_4)= d(w_5)= 3$;
	\item[{\rm ($A_{3.3}$)}]
		$d(u_2) = 6$, $d(w_2) \le 5$, $d(w_3)= 4$, $d(w_4) = d(w_5) = 3$, and $w_2 w_3\in E(G)$.
	\end{description}
\item[{\rm ($A_4$)}]
	A $3$-cycle $uvwu$ with $d(u)= d(v)= d(w)= 4$.
\item[{\rm ($A_5$)}]
	A $5$-vertex $u$ adjacent to $u_1, u_2, u_3, u_4$ and a $3$-vertex $v$, sorted by their degrees.
	Then at least one of the following cases holds:
	\begin{description}
	\parskip=0pt
	\item[{\rm ($A_{5.1}$)}]
		$d(u_3)= d(u_4)= 3$;
	\item[{\rm ($A_{5.2}$)}]
		$d(u_4)= 3$ and $u_1 v\in E(G)$;
	\item[{\rm ($A_{5.3}$)}]
		$d(u_4)= 3$, $d(u_3)= 4$, $d(u_2)= 5$, and $u_2 u_3\in E(G)$.
	\item[{\rm ($A_{5.4}$)}]
		$d(u_3)= d(u_4)= 4$ and $u_3 u_4\in E(G)$.
	\end{description}
\item[{\rm ($A_6$)}]
	A $6$-vertex $u$ adjacent to $u_1$, $u_2$ and four $3$-vertices $u_3$, $u_4$, $u_5$, $v$, sorted by their degrees.
    Then at least one of the following cases holds:
    \begin{description}
	\parskip=0pt
    \item[{\rm ($A_{6.1}$)}]
		$d(u_2)= 3$;
	\item[{\rm ($A_{6.2}$)}]
		$d(u_1)= d(u_2)= 4$ and $u_1 u_2\in E(G)$.
	\end{description}
\end{description}
\end{theorem}

The rest of the section is devoted to the proof of Theorem~\ref{thm02}, by contradiction built on discharging methods.
To this purpose, we assume to the contrary that $G$ contains none of the local structures ($A_1$)--($A_6$).
In the proof we employ several known local structural properties for planar graphs from \cite{SWW12}, which are summarized in Lemma~\ref{lemma02}.

\subsection{Definitions and notations}
Since $G$ is $2$-connected, $d(v)\ge 2$ for any vertex $v \in V(G)$.
Let $G'$ be the graph obtained by deleting all the $2$-vertices of $G$;
let $H$ be a connected component of $G'$.
Clearly, every vertex $v \in V(H)$ has its (called {\em original}) degree $d(v) \ge 3$ in $G$, and $H$ is also a planar graph without intersecting triangles.
In what follows, we assume that $H$ is embedded in the two dimensional plane such that its edges intersect only at their ending vertices,
and we refer $H$ to as a {\em plane} graph.

For objects in the plane graph $H$, we use the following notations:
\begin{itemize}
\parskip=0pt
\item
	$N_H(v) = \{u \mid uv\in E(H)\}$ and $d_H(v) = |N_H(v)|$ --- the degree of the vertex $v\in V(H)$;
\item
	similarly define what a $k$-{\em vertex}, a $k^+$-{\em vertex} and a $k^-$-{\em vertex} are;
\item
	$n'_k(v)$ ($n'_{k^+}(v)$, $n'_{k^-}(v)$, respectively) --- the number of $k$-vertices ($k^+$-vertices, $k^{-}$-vertices, respectively) adjacent to $v$.
\end{itemize}

Note that we do not discuss the embedding of the graph $G$.
Every face in the proof of Theorem~\ref{thm02} is in the embedding of the graph $H$, or simplified as in the graph $H$.
We reserve the character $f$ to denote a face.

\begin{itemize}
\parskip=0pt
\item
	$F(H)$ --- the face set of $H$;
\item
	$V(f)$ --- the set of vertices on (the boundary of) the face $f$;
\item
	a vertex $v$ and a face $f$ are {\em incident} if $v\in V(f)$;
\item
	$n_k(f)$ ($n_{k^+}(f)$, $n_{k^-}(f)$, respectively) --- the number of $k$-vertices ($k^+$-vertices, $k^{-}$-vertices, respectively) in $V(f)$;
\item
	$\delta(f)$ --- the minimum degree of the vertices in $V(f)$;
\item
	$d(f)$ --- the degree of the face $f$, which is the number of edges on the face $f$ with cut edges counted twice,
	and similarly define what a $k$-{\em face}, a $k^+$-{\em face} and a $k^-$-{\em face} are;
\item
	$F(v) = \{f \in F(H) \mid v\in V(f)\}$ --- the set of faces incident at $v$;
\item
	$m_k(v)$ ($m_{k^+}(v)$, $m_{k^-}(v)$, respectively) --- the number of $k$-faces ($k^+$-faces, $k^-$-faces, respectively) in $F(v)$;
\item
	for a vertex $v\in V(H)$ with $d_H(v)= k$, let $u_1, u_2, \ldots, u_{k} $ be all the neighbors of $v$ in clockwise order,
	and denote the face containing $u_1 v u_2$ ($u_2 v u_3$, $\ldots$, $u_{k-1} v u_k$, $u_k v u_1$, respectively) as $f_1$ ($f_2, \ldots, f_{k-1}, f_k$, respectively).%
\footnote{Note that $|F(v)| \le d_H(v)$, and some of these faces $f_1, f_2, \ldots, f_{d_H(v)}$ would refer to the same face when $|F(v)| < d_H(v)$.\label{fn01}}
\end{itemize}

\subsection{Structural properties}
Note from the construction of $H$ that $d_H(u) = d(u)- n_2(u)$ and thus $d_H(u) = d(u)$ if $n_2(u) = 0$.
In this subsection, we characterize structural properties for $6^-$-vertices in $H$, which will be used in the analysis of the discharging.

\begin{lemma}
\label{lemma01}
$\delta(H)\ge 3$.
More specifically, for any vertex $u \in V(H)$,
if $3 \le d(u) \le 5$, then $d_H(u) = d(u)$;
if $d(u) \ge 6$, then $d_H(u) \ge 4$.
\end{lemma}
\begin{proof}
The inexistence of ($A_{1.1}$) in $G$ states that no $5^-$-vertex is adjacent to a $2$-vertex in $G$;
thus $3 \le d(u) \le 5$ implies $d_H(u) = d(u)$.

The inexistences of $(A_{2.1})$ and $(A_{2.2})$ in $G$ together imply that every $6^+$-vertex $u$ is adjacent to at most ($d(u) - 4$) $2$-vertices in $G$;
thus if $d(u) \ge 6$, then $d_H(u) = d(u)- n_2(u) \ge d(u) - (d(u) - 4) = 4$.

The lemma follows from $\delta(G) \ge 2$ and that all $2$-vertices are deleted from $G$.
\end{proof}

\begin{corollary}
\label{coro01}
For any vertex $u \in V(H)$, $n'_2(u) = 0$, $d_H(u) = 3$ if and only if $d(u) = 3$, and $n'_3(u) = n_3(u)$.
\end{corollary}
\begin{proof}
Since $\delta(H) \ge 3$, there is no $2$-neighbor for $u$, that is, $n'_2(u) = 0$.

We have already showed that if $d(u) = 3$ then $d_H(u) = 3$, and if $d(u) \ge 4$ then $d_H(u) \ge 4$.
Thus, $d_H(u) = 3$ if and only if $d(u) = 3$.

The last item can also be interpreted as a $3$-neighbor of $u$ in $G$ remains as a $3$-neighbor of $u$ in $H$;
therefore, $n'_3(u) = n_3(u)$.
\end{proof}

By Corollary~\ref{coro01}, for any $3$-vertex we do not distinguish which graph, $G$ or $H$, it is in;
we use $n_3(u)$, instead of $n'_3(u)$, to denote the number of $3$-neighbors for the vertex $u$ in the graph $H$.
The two items in the next lemma are proven for planar graphs without $5$-cycles in \cite{SWW12}, by assuming the inexistences of a subset of $(A_1)$--$(A_6)$.
They in fact hold for all planar graphs.

\begin{lemma}{\rm \cite{SWW12}}
\label{lemma02}
%
{\rm (1)} 
Let $u$ be a vertex with $d_H(u) = 4$ and $n_3(u) \ge 1$, then $d(u) = 4$.

{\rm (2)} 
Let $f = [uvw]$ be a $3$-face with $\delta(f) = 3$, then $n_{5^+}(f) = 2$.
\end{lemma}

\begin{lemma}
\label{lemma03}
Let $f = [uvw]$ be a $3$-face with $\delta(f)= 4$.

{\rm (1)}
If $d(v) > d_H(v) = 4$, then $\min\{d_H(u)$, $d_H(w)\}\ge 6$;

{\rm (2)}
if $d_H(v) = 4$ and $4 \le d_H(u) \le 5$, then $d(v) = 4$;

{\rm (3)}
$n_{5^+}(f)\ge 1$.
\end{lemma}
\begin{proof}
(1) and (2) hold due to the inexistence of ($A_{2.3}$).

For (3), if $n_{5^+}(f) = 0$, then $d(u) = d(v) = d(w) = 4$ by (2), implying a configuration ($A_4$) in $G$, a contradiction.
\end{proof}

\begin{lemma}
\label{lemma04}
Let $u$ be a vertex with $d_H(u) = 5$.
Then $n_3(u) \le 2$, and if $n_3(u) = 2$ then $d(u) = 5$.
\end{lemma}
\begin{proof}
If $n_3(u) \ge 3$ then $n_2(u) + n_3(u) \ge d(u) - 5 + 3 \ge d(u)- 2$;
since ($A_{2.1}$) does not exist, we have $n_2(u) = 0$ and subsequently $d(u) = 5$, which contradicts to the inexistence of ($A_{5.1}$).
We thus prove that $n_3(u) \le 2$.

From $n_3(u) = 2$, we have $n_2(u) + n_3(u) = d(u) - 5 + 2 = d(u) - 3$;
since ($A_{2.2}$) does not exist, we have $n_2(u) = 0$ and subsequently $d(u) = 5$.
\end{proof}

\begin{lemma}
\label{lemma05}
Let $u$ be a vertex with $d_H(u) = 6$.
Then $n_3(u)\le 4$, and

{\rm (1)}
if $n_3(u)\ge 3$, then $d(u)= 6$;

{\rm (2)}
if $n_3(u) = 4$, then $n'_{5^+}(v)\ge 2$ for any $3$-vertex $v \in N_H(u)$.
\end{lemma}
\begin{proof}
Similar to the proof of Lemma~\ref{lemma04}, since $G$ contains no ($A_{6.1}$), we have $n_3(u) \le 4$, and if additionally $n_3(u) \ge 3$ then $d(u) = 6$,
that is, (1) holds.

If $n_3(u) = 4$, then by (1) $d(u) = 6$.
Since $G$ contains no ($A_{3.2}$), $n'_{5^+}(v)\ge 2$ for any $3$-vertex $v \in N_H(u)$, that is, (2) holds.
\end{proof}

\subsection{Discharging to show contradictions}
To derive a contradiction, we make use of the discharging methods, which are very similar to an amortized analysis.
First, by Euler's formula $|V(H)|- |E(H)|+ |F(H)|= 2$ and the relation $\sum\limits_{v\in V(H)}d_H(v)= \sum\limits_{f\in F(H)}d(f)= 2|E(H)|$,
we have the following equality:

\begin{equation}\label{eq01}
\sum_{u\in V(H)}(2d_H(u) - 6) + \sum_{f\in F(H)}(d(f) - 6) = -12.
\end{equation}

Next, we define an initial weight function $\omega(\cdot)$ on $V(H)$ by setting $\omega(u) = 2d_H(u) -6$ for each vertex $u\in V(H)$ and
setting $\omega(f) = d(f) - 6$ for each face $f\in F(H)$.
It follows from Eq.~(\ref{eq01}) that the total weight is equal to $-12$.
In what follows, we will define a set of discharging rules (R1)--(R4) to move portions of weights from vertices to faces (the function $\tau(u \rightarrow f)$).
At the end of the discharging, a new weight function $\omega'(\cdot)$ is achieved and we are able to show that $\omega'(x) \ge 0$ for every $x \in V(H)\cup F(H)$
(see Lemmas~\ref{lemma06}--\ref{lemma07}).
This contradicts the negative total weight. 

Lemma~\ref{lemma01} states that $\delta(H) \ge 3$ and thus $\omega(u) \ge 0$ for every vertex $u\in V(H)$.
During the discharging process, a $3$-vertex in $H$ is untouched;
likewise, a $6^+$-face does not receive anything;
for a $k$-vertex $u \in V(H)$ where $k \ge 4$ and a $5^-$-face $f \in F(u)$, let $\tau(u \to f)$ denote the amount of weight transferred from $u$ to $f$.
Recall that there are $k$ faces in $F(u)$, though they might not all be distinct from each other (see Footnote~\ref{fn01}).

\begin{description}
\parskip=0pt
\item[{\rm (R1)}]
	When $k = 4$, $\tau(u\to f_i) = \frac 12$ for each $i = 1, 2, 3, 4$.
\item[{\rm (R2)}]
	When $k \ge 5$ and $f = [vuw]$ is a $3$-face,
	\begin{description}
	\parskip=0pt
	\item[{\rm (R2.1)}]
		if $\delta(f) = 3$, then $\tau(u\to f)= \frac 32$;
	\item[{\rm (R2.2)}]
		if $\delta(f) = 4$, $d_H(v)= 4$ and
		\begin{description}
		\parskip=0pt
		\item[{\rm (R2.2.1)}]
			$k\ge 7$, then $\tau(u\to f)= 2$;
		\item[{\rm (R2.2.2)}]
			$k= 6$, then

			$\tau(u\to f)= \begin{cases}
				\frac 32, \mbox{ if } n_3(u)\ge 4;\\
				\frac 54, \mbox{ if } n_3(u)= 3 \mbox{ and } n_{6^+}(f)= 2;\\
				\frac 54, \mbox{ if } n_3(u)= 3, d_H(w)= 5 \mbox{ and } n_3(w)\le 1;\\
				2, \mbox{ otherwise};
				\end{cases}$
		\item[{\rm (R2.2.3)}]
			$k= 5$, then $\tau(u\to f)= \begin{cases}
				2, 		  \mbox{ if } n_3(u)= 0;\\
				\frac 54, \mbox{ if } n_3(u)= 1;\\
				1, 	 	  \mbox{ if } n_3(u)\ge 2;
				\end{cases}$
		\end{description}
	\item[{\rm (R2.3)}]
		if $\delta(f)\ge5 $, then $\tau(u\to f)= 1$.
	\end{description}
\item[{\rm (R3)}]
	When $k\ge 5$ and $f= [vuyx]$ is a $4$-face,

	$\tau(u\to f)=\begin{cases}
		1, 		  \mbox{ if } d(v) = 3 \mbox{ and, either } d(y) = 3 \mbox{ or } d_H(x)= 4;\\
		\frac 34, \mbox{ if } d(v)= 3, d_H(y)\ge 4 \mbox{ and } d_H(x)\ge 5;\\
		\frac 12, \mbox{ if } d_H(v)\ge 4 \mbox{ and } d_H(y)\ge 4.
		\end{cases}$
\item[{\rm (R4)}]
	When $k\ge 5$ and $f= [\cdot xuy\cdot]$ is a $5$-face,

	$\tau(u\to f)=\begin{cases}
		\frac 12, \mbox{ if } d(x) = d(y) = 3;\\
		\frac 14, \mbox{ if } \max\{d_H(x), d_H(y)\}\ge 4.
		\end{cases}$
\end{description}

It remains to validate that $\omega'(x)\ge 0$ for every $x\in V(H)\cup F(H)$.

\begin{lemma}
\label{lemma06}
For every face $f\in F(H)$, $\omega'(f)\ge 0$.
\end{lemma}
\begin{proof}
We distinguish the following four cases for $d(f)$.

{\bf Case 1.} $d(f)\ge 6$.
In this case $f$ does not receive anything and thus $\omega'(f) = \omega(f) = d(f) - 6 \ge 0$.

{\bf Case 2.} $d(f)= 3$ with $f= [uvw]$.
In this case, $\omega(f)= d(f)- 6 = -3$.

If $\delta(f)= 3$, then $n_{5^+}(f) = 2$ by Lemma~\ref{lemma02}(2);
	by (R2.1) each $5^+$-vertex in $V(f)$ gives $\frac 32$ to $f$ and thus $\omega'(f) \ge -3 + 2\times \frac 32 = 0$.

If $\delta(f)\ge 5$, then by (R2.3) $\tau(u\to f)= 1$ for each $u\in V(f)$, and thus $\omega'(f)\ge -3 + 3\times 1 = 0$.

If $\delta(f)= 4$ with $d_H(v)= 4$, then $n_{4^+}(f)= 3$, $n_{5^+}(f)\ge 1$  by Lemma~\ref{lemma03}, and $n_3(x)\le d_H(x)- 2$ for each $x\in \{u, w\}$.
By (R1) $\tau(v\to f)= \frac 12$.
When $\min\{\tau(u\to f), \tau(w\to f)\}\ge \frac 54$, $\omega'(f)\ge - 3+ 2\times\frac 54+ \frac 12 = 0$;
when $\max\{\tau(u\to f), \tau(w\to f)\}= 2$, $\omega'(f)\ge - 3+ 2 + 2\times\frac 12 = 0$.
In the other cases:

{Case 2.1.} $d_H(u)= 4$ and, either $d_H(w)= 6$ with $n_3(w)= 4$ or $d_H(w)= 5$ with $n_3(w)\ge 1$.
In this subcase, $d(u)= d(v)= 4$ by Lemma~\ref{lemma03} and thus one of ($A_{2.4}$), ($A_{5.4}$) and ($A_{6.2}$) exists, a contradiction.

{Case 2.2.} $d_H(u)=d_H(w)= 5$, with $n_3(u)\ge 2$ and $n_3(w)\ge 1$.
In this subcase, $n_3(u)= 2$ and $d(u)= 5$ by Lemma~\ref{lemma04}, and $d(v)= 4$ by Lemma~\ref{lemma03}.
Since $G$ contains no ($A_{2.4}$), we have $d(w)= 5$ and thus ($A_{5.3}$) exists, a contradiction.

{\bf Case 3.} $d(f)= 4$ with $f= [uvxy]$.
In this case, $\omega(f) = d(f)- 6 = -2$.

If $\delta(f)\ge 4$, then by (R1) and (R3) $\omega'(f)\ge -2 + 4\times \frac 12 = 0$.

Consider next $\delta(f)= 3$, that is, $n_3(f)\ge 1$.
Since $G$ contains neither $(A_{1.2})$ nor $(A_{1.4})$, $n_3(f)\le 2$.
Therefore, we have $1 \le n_3(f) \le 2$.

If $n_3(f)= 2$, then $n_{5^+}(f)= 2$ and by (R3) $\omega'(f) = -2 + 2\times 1= 0$.

If $n_3(f)= 1$, then $d(v)= 3$ and, since $G$ contains no $(A_{1.5})$, we have $\max\{d_H(x), d_H(u)\}\ge 5$.
Assuming w.l.o.g. $d_H(u)\ge 5$, by (R1) and (R3),
if $d(x)= 4$, then we have $\tau(y\to f)\ge \frac 12$, $\tau(u\to f)= 1$, and $\tau(x\to f)= \frac 12$, leading to $\omega'(f) \ge -2 + 1+ 2\times\frac 12 = 0$;
or if $d_H(x)\ge 5$, then we have $\tau(u\to f)= \tau(x\to f)= \frac 34$, leading to $\omega'(f) \ge -2 + 2\times \frac 34 + \frac 12 = 0$.

{\bf Case 4.} $d(f)= 5$ with $f= [uvwxy]$.
In this case, $\omega(f) = d(f)- 6 = -1$.

By (R1) and (R4), each $4^+$-vertex in $V(f)$ gives at least $\frac 14$ to $f$.
If $n_{4^+}(f) \ge 4$, then $\omega'(f) \ge -1 + 4\times \frac 14 = 0$.

Otherwise, since $G$ contains neither ($A_{1.2}$) nor ($A_{1.4}$) and by Corollary~\ref{coro01},
we may assume w.l.o.g. $d(u) = d(w) = 3$, $n_{4^+}(f) = 3$, and $d_H(v)\ge 5$.
Then by (R1) and (R4), we have $\tau(v\to f) \ge\frac 12$, $\tau(x\to f)\ge \frac 14$, and $\tau(y\to f)\ge \frac 14$,
leading to $\omega'(f)\ge -1 + 1\times\frac 12+ 2\times \frac 14 = 0$.

This finishes the proof of the lemma.
\end{proof}

\begin{lemma}
\label{lemma07}
For every vertex $u\in V(H)$, $\omega'(u)\ge 0$.
\end{lemma}
\begin{proof}
Recall from Lemma~\ref{lemma01} that $d_H(u)\ge 3$ for every vertex $u \in V(H)$.

If $d_H(u) = 3$, then $u$ is untouched during the discharging process and thus $\omega'(u) = \omega(u) = 2d_H(u)- 6 = 0$.

If $d_H(u) = 4$, then $\omega'(u)= \omega(u)- 4\times \frac 12 = 2d_H(u)- 6- 2= 0$ by (R1).

If $d_H(u)\ge 7$, then using the heaviest weights in (R2.1), (R2.2.1), (R2.3), (R3) and (R4),
we have $\omega'(u)\ge 2d_H(u) - 6 - 2 - 1\times(d_H(u) - 1) = d_H(u) - 7 \ge 0$.

Below we distinguish two cases where $d_H(u) = 5$ and $6$, respectively.
Note that $m_3(u)\le 1$ since $G$ contains no intersecting triangles.

{\bf Case 1.} $d_H(u)= 5$.
In this case, $\omega(u)= 2d_H(u)- 6= 4$.

By Lemma~\ref{lemma04}, $n_3(u)\le 2$, and if $n_3(u)= 2$ then $d(u)= 5$.
When $m_3(u)= 0$, if the number of faces with $\tau(u\to f)\le \frac 12$ is at least $2$,
then by (R3) and (R4) $\omega'(u)\ge 4- 2\times\frac 12- 3\times 1= 0$;
otherwise, there is at most one face $f\in F(u)$ such that $\tau(u\to f)\le \frac 12$, and thus $n_3(u) = 2$.
One sees again by (R3) and (R4) that these two $3$-neighbors of $u$ should not be on the same face,
and we assume w.l.o.g. they are $u_2$ and $u_4$, i.e., $d(u_2) = d(u_4) = 3$.
For each of $u_2$ and $u_4$, at least one of its neighbors besides $u$ is a $5^+$-vertex by Lemma~\ref{lemma02} and the inexistence of ($A_{1.1}$), ($A_{1.2}$) and ($A_{1.5}$).
It follows from again (R3) and (R4) that $\tau(u\to f_1) + \tau(u\to f_2)\le 1+ \frac 34$,
$\tau(u\to f_3) + \tau(u\to f_4)\le 1+ \frac 34$,
$\tau(u\to f_5)\le \frac 12$,
and thus $\omega'(u)\ge 4- 2\times (1+ \frac 34)- \frac 12= 0$.

When $m_3(u)= 1$ and assuming w.l.o.g. $d(f_1)= 3$ with $f_1= [u_1uu_2]$,
we let $s_3(u)$ be the number of $3$-vertices in $N_H(u)\setminus \{u_1, u_2\}$ and we have $s_3(u) \le n_3(u) \le 2$ by Lemma~\ref{lemma04}.
Using (R2.1), (R2.2.3), (R2.3), (R3) and (R4), we discuss the following three subcases for three possible values of $s_3(u)$.

{Case 1.1.} $s_3(u)= 0$.
If $\delta(f_1)= 3$ with $d(u_1)= 3$ (with $d(u_2) = 3$, can be symmetrically argued),
then $\tau(u\to f_1)= \frac 32$, $\tau(u\to f_i)\le \frac 12$, $i\in \{2, 3, 4\}$, $\tau(u\to f_5)\le 1$,
leading to $\omega'(u)\ge 4- \frac 32- 3\times\frac 12- 1= 0$.
If $\delta(f_1)\ge 4$, then $\tau(u\to f_1)\le 2$, $\tau(u\to f_i)\le \frac 12$, $i\in \{2, 3, 4, 5\}$,
leading to $\omega'(u)\ge 4- 2- 4\times\frac 12= 0$.

{Case 1.2.} $s_3(u)= 1$.
By Lemma~\ref{lemma04}, if $n_3(u) = 2$ then $d(u) = 5$ and one of $u_1$ and $u_2$ is a $3$-vertex, which contradicts the inexistence of ($A_{5.2}$).
Therefore, $n_3(u) = 1$, which implies $\delta(f_1)\ge 4$ and further by (R2.2.3) and (R2.3) $\tau(u\to f_1)\le \frac 54$.
If $d(u_3) = 3$ ($d(u_5) = 3$ can be symmetrically argued), then $\tau(u\to f_2) + \tau(u\to f_3)\le 1+ \frac 34$, $\tau(u\to f_i)\le \frac 12$, $i\in \{4, 5\}$,
leading to $\omega'(u)\ge 4- \frac 54- (1+ \frac 34)- 2\times\frac 12= 0$;
if $d(u_4) = 3$, then $\tau(u\to f_3) + \tau(u\to f_4)\le 1+ \frac 34$, $\tau(u\to f_i)\le \frac 12$, $i\in \{1, 5\}$,
leading to $\omega'(u)\ge 4- \frac 54- (1+ \frac 34)- 2\times\frac 12= 0$.

{Case 1.3.} $s_3(u)= 2$.
We have $d(u) = 5$, $n_3(u) = 2$ and $\delta(f_1) \ge 4$, and further by (R2.2.3) and (R2.3) $\tau(u\to f_1)\le 1$.
By Lemma~\ref{lemma02} and the inexistence of ($A_{1.1}$), ($A_{1.2}$) and ($A_{3.1}$), all neighbors of the two $3$-vertices in $N_H(u)$ are $5^+$-vertices.
If $d(u_3) = d(u_4) = 3$ ($d(u_5) = d(u_4) = 3$ can be symmetrically argued),
then $\tau(u\to f_i)\le \frac 34$, $i\in \{2, 4\}$, $\tau(u\to f_3)\le 1$, $\tau(u\to f_5)\le \frac 12$,
leading to $\omega'(u)\ge 4- 2\times1- 2\times\frac 34- \frac 12= 0$;
if $d(u_3) = d(u_5)= 3$, then $\tau(u\to f_j)\le \frac 34$, $j\in \{2, 3, 4, 5\}$, leading to $\omega'(u)\ge 4- 1- 4\times\frac 34 = 0$.

{\bf Case 2.} $d_H(u)= 6$.
In this case, $\omega(u)= 2d_H(u) - 6 = 6$.

When $m_3(u) = 0$, $\omega'(u)\ge 6- 1\times 6= 0$.

When $m_3(u)= 1$ and assuming w.l.o.g. $d(f_1)= 3$ with $f_1= [u_1uu_2]$, we let $t_3(u)$ be the number of $3$-vertices in $N_H(u)\setminus \{u_1, u_2\}$.
By Lemma~\ref{lemma05}, we have $t_3(u) \le n_3(u) \le 4$, and furthermore if $n_3(u) \ge 3$ then $d(u)= 6$.
By Lemma~\ref{lemma02} and the inexistence of ($A_{1.1}$), ($A_{1.2}$), ($A_{1.5}$) and ($A_{3.2}$),
for each $3$-vertex $x\in N_H(u)$, at least one of its neighbors besides $u$ is a $5^+$-vertex;
furthermore, if $n_3(u) \ge 4$, then all its neighbors are $5^+$-vertices.
Using (R2.1), (R2.2.2), (R2.3), (R3) and (R4),
if $\delta(f_1)\ge 4$ then $\tau(u\to f_1)\le 2$,
and if $\delta(f_1)= 3$ then $\tau(u\to f_1)\le \frac 32$.
We discuss the following five subcases for five possible values of $t_3(u)$, respectively.

{Case 2.1.} $t_3(u)= 4$.
Then $n_3(u) = 4$, $d(u)= 6$, and $\delta(f_1)\ge 4$ by Lemma~\ref{lemma05}.
Thus, $\tau(u\to f_1)\le \frac 32$, $\tau(u\to f_i)\le \frac 34$, $i\in \{2, 6\}$, $\tau(u\to f_i)\le 1$, $i\in \{3, 4, 5\}$,
leading to $\omega'(u)\ge 6- \frac 32- 2\times \frac 34- 3\times 1= 0$.

{Case 2.2.} $t_3(u)= 3$.
Then $n_3(u) \ge 3$ and $d(u)= 6$.
If $\delta(f_1)\ge 5$, then $\tau(u\to f_i)\le 1$, $1\le i\le 6$, leading to $\omega'(u)\ge 6- 6\times 1= 0$.
Otherwise, $\delta(f_1)\in \{3, 4\}$ and we assume w.l.o.g. $d(u_3)= d(u_4)= 3$.
Then $\tau(u\to f_3)\le 1$. We can have the following two scenarios for two possible values of $\delta(f_1)$:

(2.2.1.) $\delta(f_1)= 3$.
	 In this scenario, $\tau(u\to f_1)= \frac 32$, and all neighbors of any $3$-vertex in $N_H(u)$ are $5^+$-vertices.
     If $d(u_2)= d(u_5)= 3$, then $\tau(u\to f_i)\le 1$, $i\in \{2, 4, 5\}$, $\tau(u\to f_6)\le \frac 12$,
     leading to $\omega'(u)\ge 6- \frac 32- 4\times 1- \frac 12= 0$.
     If $d(u_1)= d(u_5)= 3$, then $\tau(u\to f_4)\le 1$, $\tau(u\to f_i)\le \frac 34$, $i\in \{2, 5, 6\}$,
     leading to $\omega'(u)\ge 6- \frac 32- 2\times 1- 3\times\frac 34= \frac 14> 0$.
     If $d(u_6)= 3$, then $\tau(u\to f_i)\le \frac 34$, $i\in \{4, 5\}$, and $\tau(u\to f_1) + \tau(u\to f_6)\le 1 + \frac 34$,
     leading to $\omega'(u)\ge 6- \frac 32- 2\times 1- 3\times\frac 34= \frac 14> 0$.

(2.2.2.) $\delta(f_1)= 4$, and assuming that $d_H(x)= 4$ where $\{x, y\}= \{u_1, u_2\}$.
     Note that $\tau(u\to f_i)\le 1$, $i\in \{2, 4\}$.
	 If $d(u_5) = 3$, then $\tau(u\to f_5)\le 1$, $\tau(u\to f_6)\le \frac 12$,
     or if $d(u_6) = 3$, then $\tau(u\to f_5) + \tau(u\to f_6)\le 1+ \frac 34$.
     Thus, if $\tau(u\to f_1)\le \frac 54$, then $\omega'(u)\ge 6- \frac 54- 4\times 1- \frac 34 = 0$.

     Otherwise we have $\tau(u\to f_1)= 2$, and it suffices to assume that $d_H(y)= 4$, or $d_H(y)= 5$ with $n_3(y)\ge 2$.
     By Lemmas~\ref{lemma03} and~\ref{lemma04}, we have $d(x)= 4$, and either $d(y)= 4$, or $d(y)= 5$ with $n_3(y)\ge 2$.
     Lemma~\ref{lemma02} and the inexistence of ($A_{1.1}$), ($A_{1.2}$) and ($A_{3.3}$) state that all neighbors of any $3$-vertex $x\in N_H(u)$ are $5^+$-vertices.
     If $d(u_5) = 3$, then $\tau(u\to f_i)\le \frac 34$, $i\in \{2, 5\}$, $\tau(u\to f_6)\le \frac 12$,
     leading to $\omega'(u)\ge 6- 2- 2\times 1- 2\times \frac 34- \frac 12= 0$;
     if $d(u_6) = 3$, then $\tau(u\to f_i)\le \frac 34$, $i\in \{2, 4, 5, 6\}$,
     leading to $\omega'(u)\ge 6- 2- 1\times 1- 4\times\frac 34= 0$.

{Case 2.3.} $t_3(u)= 2$.
By symmetry ($u_3$ and $u_6$ have symmetrical roles, and $u_4$ and $u_5$ have symmetrical roles), we can have the following four scenarios:

(2.3.1.) $d(u_3)= d(u_4)= 3$.
    In this scenario, $\tau(u\to f_i)\le 1$, $i\in \{3, 4\}$, and $\tau(u\to f_5)\le \frac 12$.
    If $\delta(f_1)\ge 4$, then $\tau(u\to f_2)\le 1$, $\tau(u\to f_6)\le \frac 12$,
    leading to $\omega'(u)\ge 6- 2- 3\times 1- 2\times \frac 12= 0$.
    If $\delta(f_1)= 3$, then $\tau(u\to f_i)\le 1$, $i\in \{2, 6\}$,
    leading to $\omega'(u)\ge 6- \frac 32- 4\times 1- \frac 12= 0$.

(2.3.2.) $d(u_3)= d(u_5)= 3$.
    In this scenario, $\tau(u\to f_4) + \tau(u\to f_5)\le 1+ \frac 34$.
    If $\delta(f_1)\ge 4$, then $\tau(u\to f_2) + \tau(u\to f_3)\le 1+ \frac 34$, $\tau(u\to f_6)\le \frac 12$,
    leading to $\omega'(u)\ge 6- 2- 2\times (1+ \frac 34)- \frac 12= 0$.
    If $d(u_2)= 3$, then $\tau(u\to f_i)\le 1$, $i\in \{2, 3\}$, $\tau(u\to f_6)\le \frac 12$,
    leading to $\omega'(u)\ge 6- \frac 32- 2\times 1- (1+ \frac 34)- \frac 12= \frac 14$.
    If $d(u_1)= 3$, $\tau(u\to f_2) + \tau(u\to f_3)\le 1+ \frac 34$, $\tau(u\to f_6)\le 1$,
    leading to $\omega'(u)\ge 6- \frac 32- 2\times (1+ \frac 34)- 1= 0$.

(2.3.3.) $d(u_3)= d(u_6)= 3$.
    In this scenario, $\tau(u\to f_4)\le \frac 12$.
    If $\delta(f_1)\ge 4$, then $\tau(u\to f_2) + \tau(u\to f_3)\le 1+ \frac 34$,
    $\tau(u\to f_5) + \tau(u\to f_6)\le 1+ \frac 34$, leading to $\omega'(u)\ge 6- 2- 2\times (1+ \frac 34)- \frac 12= 0$.
    If $\delta(f_1)= 3$ and assuming w.l.o.g. $d(u_1)= 3$,
    then $\tau(u\to f_i)\le 1$, $i\in \{2, 3\}$, $\tau(u\to f_5) + \tau(u\to f_6)\le 1+ \frac 34$,
	leading to $\omega'(u)\ge 6- \frac 32- 2\times 1- (1+ \frac 34)- \frac 12= \frac 14$.

(2.3.4.) $d(u_4)= d(u_5)= 3$.
    In this scenario, $\tau(u\to f_i)\le 1$, $i\in \{3, 4, 5\}$.
    If $\delta(f_1)\ge 4$, then $\tau(u\to f_i) \le \frac 12$, $i\in \{2, 6\}$,
    leading to $\omega'(u)\ge 6- 2- 3\times 1- 2\times\frac 12= 0$.
    If $\delta(f_1)= 3$, then $\tau(u\to f_2) + \tau(u\to f_6)\le 1+ \frac 12$,
    leading to $\omega'(u)\ge 6- \frac 32- 4\times 1- \frac 12= 0$.

{Case 2.4.} $t_3(u)= 1$ and assuming w.l.o.g. $d(u_3)= 3$ or $d(u_4)= 3$.

(2.4.1.) $d(u_3)= 3$.
    In this scenario, $\tau(u\to f_i)\le \frac 12$, $i\in \{4, 5\}$.
    If $\delta(f_1)\ge 4$, then $\tau(u\to f_6)\le \frac 12$,
    $\tau(u\to f_2)+ \tau(u\to f_3) \le 1 + \frac 34$,
    leading to $\omega'(u)\ge 6- 2- (1+ \frac 34)- 3\times\frac 12= \frac 34> 0$.
    If $d(u_2)= 3$, then $\tau(u\to f_i)\le 1$, $i\in \{2, 3\}$,
    $\tau(u\to f_6)\le \frac 12$, leading to $\omega'(u)\ge 6- \frac 32- 2\times 1- 3\times \frac 12= 1$.
    If $d(u_1)= 3$, then $\tau(u\to f_2)+ \tau(u\to f_3) \le 1 + \frac 34$, $\tau(u\to f_6)\le 1$,
	leading to $\omega'(u)\ge 6- \frac 32- 1\times 1- (1+ \frac 34)- 2\times \frac 12= \frac 34$.

(2.4.2.) $d(u_4)= 3$.
    In this scenario, $\tau(u\to f_3)+ \tau(u\to f_4) \le 1 + \frac 34$,
     $\tau(u\to f_5)\le \frac 12$.
     If $\delta(f_1)\ge 4$, then $\tau(u\to f_i)\le \frac 12$, $i\in \{2, 6\}$,
     leading to $\omega'(u)\ge 6- 2- (1+ \frac 34)- 3\times\frac 12= \frac 34> 0$.
     If $\delta(f_1)= 3$, then $\tau(u\to f_2)+ \tau(u\to f_6) \le 1 + \frac 12$,
     leading to $\omega'(u)\ge 6- \frac 32- \frac 12- (1+ \frac 34)- (1+ \frac 12)= \frac 34$.

{Case 2.5.} $t_3(u)= 0$.
We have $\tau(u\to f_1) \le 2$, $\tau(u\to f_i)\le \frac 12$, $i\in \{3, 4, 5\}$, and $\tau(u\to f_2)+ \tau(u\to f_6)\le 1+ \frac 12$.
Thus, $\omega'(u)\ge 6- 2- 3\times\frac 12- (1 + \frac 12) = 1$.

This finishes the proof of the lemma that for every vertex $u \in V(H)$, $\omega'(u) \ge 0$.
\end{proof}

Lemmas~\ref{lemma06} and~\ref{lemma07} together contradict the negative total weight of $-12$ stated in Eq.~(\ref{eq01}),
and thus prove Theorem~\ref{thm02}.

\section{Acyclic edge coloring}
In this section, we show how to derive an acyclic edge coloring, by an induction on $|E(G)|$ and by recoloring certain edges in each specified local structure.
Recall that $G$ is a simple $2$-connected planar graph without intersecting triangles.
The following lemma gives the starting point.

\begin{lemma}{\rm (\cite{Sku04,AMM12,BC09,SWMW19,WMSW19})}
\label{lemma08}
If $\Delta \in \{3, 4\}$, then $a'(G) \le \Delta + 2$, and an acyclic edge $(\Delta + 2)$-coloring can be obtained in polynomial time.
\end{lemma}

Given a partial acyclic edge $k$-coloring $c(\cdot)$ of the graph $G$ using the color set $C = \{1, 2, \ldots, k\}$,
for a vertex $v\in V(G)$, let $C(v)$ denote the set of colors assigned to the edges incident at $v$ under $c$.
If the edges of a path $P= ux\ldots v$ are alternatively colored $i$ and $j$, we call it an {\em $(i, j)_{(u, v)}$-path}.
Furthermore, if $uv\in E(G)$ is also colored $i$ or $j$, we call $ux\ldots vu$ an {\em $(i, j)_{(u, v)}$-cycle}.

For simplicity, we use {\em $\{e_1, e_2, \ldots, e_m\} \to a$} to state that all the edges $e_1$, $e_2$, $\ldots$, $e_m$ are colored $a$,
use simply $e_1 \to a$ to state that $e_1$ is colored $a$,
use $e_1 \to S$ ($S \ne\emptyset$) to state that $e_1$ is colored with a color in $S$,
and use $(e_1, e_2, \ldots, e_m) \to (a_1, a_2, \ldots, a_m)$ to state that $e_j$ is colored $a_j$, for $j = 1, 2, \ldots, m$.
We also use $(e_1, e_2, \ldots, e_m)_c = (a_1, a_2, \ldots, a_m)$ to denote that $c(e_j) = a_j$, for $j = 1, 2, \ldots, m$.
This way, when $(uv, xy)_c = (a, b)$, switching their colors is represented as $(uv, xy) \to (b, a)$.

The following lemma is obvious (via a simple contradiction):

\begin{lemma}{\rm \cite{SWMW19}}
\label{lemma09}
Suppose $G$ has an acyclic edge coloring $c(\cdot)$,
and $P = uv_1v_2$-$\ldots$-$v_kv_{k+1}$ is a maximal $(a, b)_{(u, v_{k+1})}$-path with $c(uv_1) = a$ and $b \not\in C(u)$.
Then there is no $(a, b)_{(u, w)}$-path for any vertex $w \not\in V(P)$.
\end{lemma}

The rest of the section is devoted to the proof of Theorem~\ref{thm01}, by induction on the number $|E(G)|$ of edges.
The flow of the proof is depicted in Figure~\ref{fig02}.

For the base cases in the induction,
one sees that when $|E(G)| \le \Delta + 2$, coloring each edge by a distinct color gives a valid acyclic edge $(\Delta + 2)$-coloring;
when $\Delta \le 4$, an acyclic edge $(\Delta + 2)$-coloring is guaranteed by Lemma~\ref{lemma08}.

\begin{figure}[]
\begin{center}
\vskip -2pt
\includegraphics[width=5.7in]{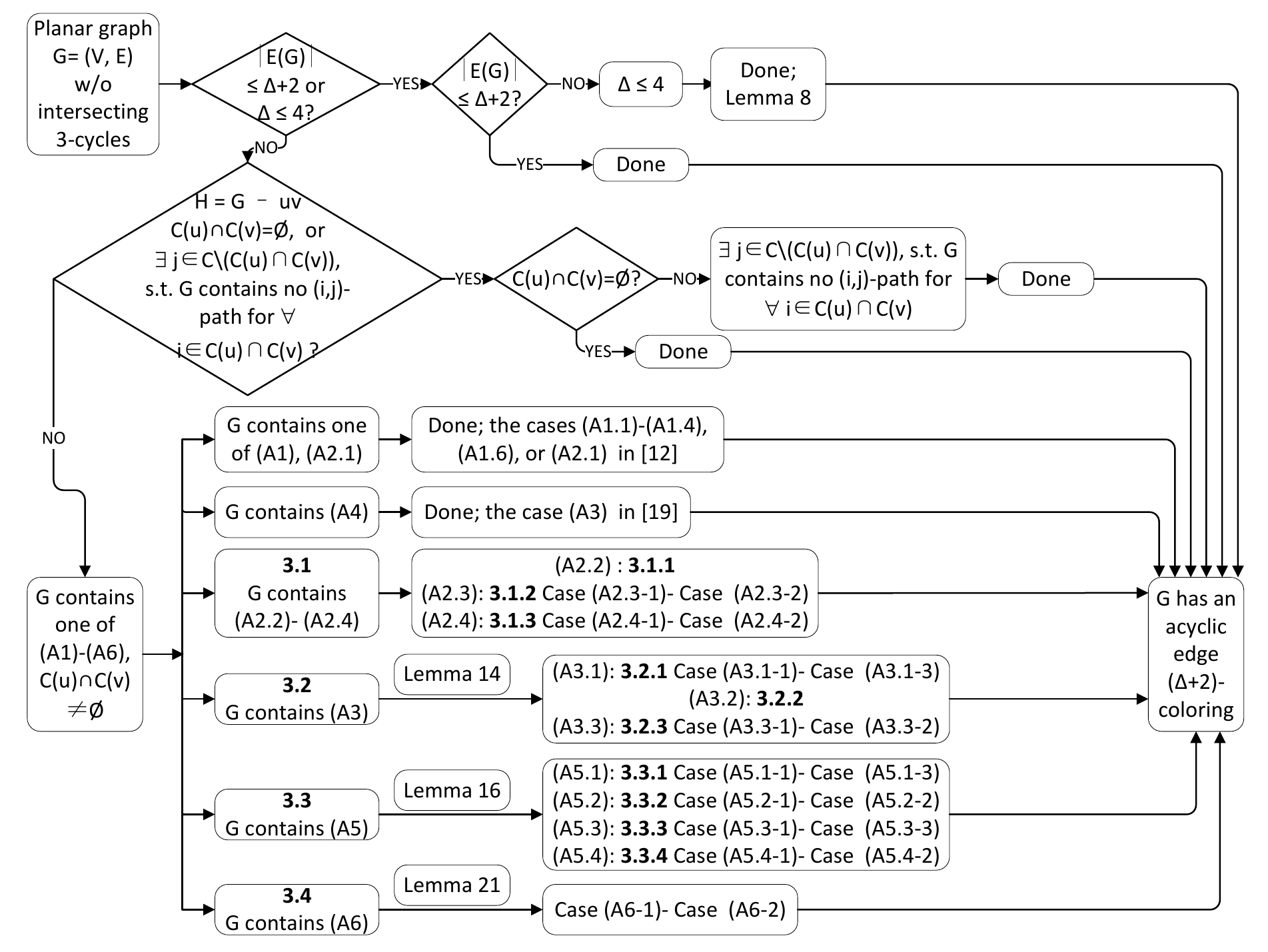}
\vskip 2pt
\caption{The flow of the proof of Theorem~\ref{thm01} by induction on $|E(G)|$.\label{fig02}}
\end{center}
\end{figure}

In the sequel we consider a simple $2$-connected planar graph $G$ without intersecting triangles such that $\Delta \ge 5$ and $|E(G)| \ge \Delta + 3$.
Theorem~\ref{thm02} states that $G$ contains at least one of the configurations ($A_1$)--($A_6$).
One sees that for the edge $uv$ in each of the configurations (see Figure~\ref{fig01}),
we have the degree $d(v) \le 3$ (in ($A_{1.1}$)--($A_{1.3}$), ($A_{1.5}$), ($A_2$), ($A_5$), ($A_6$), respectively),
or the degree $d(u) = 3$ (in ($A_{1.2}$), ($A_{1.4}$), ($A_3$), respectively),
or the sum of their degrees $d(u) + d(v) \le 8$ (in ($A_1$), ($A_3$)--($A_5$), respectively).
We pick the edge $uv$ and let $H = G - uv$ (that is, $H$ is obtained by removing the edge $uv$ from $G$, or $H = (V(G), E(G) \setminus\{uv\})$).
The graph $H$ is also planar, contains no intersecting triangles, and $4 \le \Delta - 1 \le \Delta(H) \le \Delta$.%
\footnote{Note that $H$ might not be $2$-connected, but that is OK since we do not need its $2$-connectivity.}
By the {\em inductive hypothesis}, $H$ has an acyclic edge $(\Delta(H) + 2)$-coloring $c(\cdot)$, and thus a $(\Delta + 2)$-coloring,
using the color set $C = \{1, 2, \ldots, \Delta+2\}$.
If $d(u) + d(v) \le 8$, then $d_H(u) + d_H(v) \le 6 \le \Delta + 1$;
or otherwise $d_H(u) + d_H(v) \le (\Delta - 1) + 2 = \Delta + 1$.
That is, either way we have $d_H(u) + d_H(v) \le \Delta + 1 < |C|$, or equivalently $C \setminus (C(u)\cup C(v)) \ne \emptyset$.

One clearly sees that,
if $C(u)\cap C(v) = \emptyset$, then $uv\to C\setminus (C(u)\cup C(v))$ gives an acyclic edge $(\Delta + 2)$-coloring for $G$;
or if there exists a $j\in C\setminus (C(u)\cup C(v))$ such that, for all $i\in C(u)\cap C(v)$, $H$ contains no $(i, j)_{(u, v)}$-path,
then $uv\to j$ gives an acyclic edge $(\Delta + 2)$-coloring for $G$ too;
and we are done.%

In the remaining case, for each $i\in C(u)\cap C(v)$, assume that $c(u u_s) = c(v v_t) = i$,
and we define the color set
\begin{equation}
\label{eq02}
B_i = \{j \mid \mbox{there is an } (i, j)_{(u_s, v_t)}\mbox{-path in } H\}.
\end{equation}
Therefore, for each $j \in B_i$, there is a neighbor $x$ of $u_s$ and a neighbor $y$ of $v_t$, respectively,
such that $c(u_s x) = c(v_t y) = j$, and each of $x$ and $y$ is incident with an edge colored $i$.
One sees that $B_i \subseteq C(u_s)\cap C(v_t)$.
In the sequel we continue the proof with the following Proposition~\ref{prop3001},%
\footnote{This and many succeeding propositions are numbered in subsubsections, and they are assumed true for the following-up proof, respectively;
	that is, if any of them is false, then an acyclic edge $(\Delta + 2)$-coloring for $G$ has already been obtained.}
or otherwise we are done:

\begin{proposition}
\label{prop3001}
	$C(u)\cap C(v) \ne \emptyset$, and
	$C\setminus (C(u)\cup C(v))\subseteq \bigcup _{i\in C(u)\cap C(v)}B_i$.
\end{proposition}

For each degree $k\in \{2, 3, 4, 5\}$, we collect the colors of the edges connecting $u$ and its $k$-neighbors,
into $S_k = \{c(u w) \mid w \in N(u) \setminus \{v\}, d(w) = k\}$.

We affirm the readers, and it is not hard to see, that if $G$ contains ($A_1$) or ($A_{2.1}$),
then the proof for the configurations ($A_{1.1}$)--($A_{1.4}$), ($A_{1.6}$) and ($A_{2.1}$) in \cite{SWW12}
can be adopted to derive an acyclic edge $(\Delta + 2)$-coloring for $G$;
and if $G$ contains ($A_4$), then the proof for the configuration ($A_3$) in \cite{WSW4} can be adopted.%
The next four subsections deal with the other configurations in ($A_2$), ($A_3$), ($A_5$), and ($A_6$), respectively.

\subsection{Configurations ($A_{2.2}$)--($A_{2.4}$)}
In this subsection we prove the inductive step for the case where $G$ contains one of the configurations ($A_{2.2}$)--($A_{2.4}$).
We have the following revised ($A_2$):
\begin{description}
\parskip=0pt
\item[$(A_2)$]
    A $6^+$-vertex $u$ is adjacent to a $2$-vertex $v$.
	Let $u_1, u_2, \ldots, u_{d(u)-1}$ be the other neighbors of $u$.
	At least one of the configurations ($A_{2.2}$)--($A_{2.4}$) occurs.
\end{description}
See Figure~\ref{fig01} for illustrations.
Let $w$ denote the other neighbor of $v$.
For each $1 \le i \le d(u)-1$, if $d(u_i)= 2$, let $x_i$ denote the other neighbor of $u_i$.

We first characterize two useful particularities of the $2$-vertices in $N(u)$.
For ease of presentation we let $u_{d(u)} = v$ and $x_{d(u)} = w$.

\begin{lemma}
\label{lemma10}
For any $2$-vertex $u_i \in N(u)$, if $u x_i \not\in E(G)$, and $x_i u_j \notin E(G)$ for every $u_j \in N(u) \setminus \{u_i\}$,
then $G$ admits an acyclic edge $(\Delta + 2)$-coloring.
\end{lemma}
\begin{proof}
Let $G' = G - u_i + u x_i$, that is, merging the two edges $u u_i$ and $u_i x_i$ into a new edge $u x_i$ while eliminating the vertex $u_i$.

The graph $G'$ is a planar graph without intersecting triangles, $\Delta(G') = \Delta$, and contains one less edge than $G$.
By the inductive hypothesis, $G'$ has an acyclic edge $(\Delta + 2)$-coloring $c(\cdot)$ using the color set $C = \{1, 2, \ldots, \Delta + 2\}$,
in which $C(u)$ is the color set of the edges incident at $u$ and we assume that $c(u x_i) = 1$.
To transform back an edge coloring for $G$,
we keep the color for every edge of $E(G') \setminus \{u x_i\}$, $u_i x_i \to 1$, and $u u_i \to C \setminus C(u)$.
One may trivially verify that no bichromatic cycle would be introduced, that is, the achieved edge $(\Delta + 2)$-coloring for $G$ is acyclic.
\end{proof}

\begin{lemma}
\label{lemma11}
For any two $2$-vertices $u_i, u_j \in N(u)$,
if in an edge coloring $c(\cdot)$ we have $c(u u_i) = i$, $c(u u_j) = j$,
and there exist an $(\alpha, i)_{(u, u_i)}$-cycle and an $(\alpha, j)_{(u, u_j)}$-cycle for some color $\alpha \not\in \{i, j\}$,
then we can switch the colors of $\{u u_i, u u_j\}$
such that $G$ contains neither an $(\alpha, j)_{(u, u_i)}$-cycle nor an $(\alpha, i)_{(u, u_j)}$-cycle and no new bichromatic cycle is introduced.
\end{lemma}
\begin{proof}
First of all, since $u_i$ and $u_j$ are both $2$-vertices, we have $c(u_i x_i) = c(u_j x_j) = \alpha$.
It follows that switching the colors of $\{u u_i, u u_j\}$, that is $(u u_i, u u_j) \to (j, i)$, is feasible.

Secondly, any newly introduced bichromatic cycle would get at least one of $u u_i$ and $u u_j$ involved,
which is impossible by a simple contradiction to the feasibility of $c(\cdot)$.
\end{proof}

We continue with the inductively hypothesized acyclic edge coloring $c(\cdot)$ for $H = G - u v$.
By Proposition~\ref{prop3001}, $c(vw) \in C(u)$, $H$ contains a $(c(v w), i)_{(u, w)}$-path for every $i \in C \setminus C(u)$,
and thus $(C\setminus C(u))\subseteq C(w) \setminus \{c(v w)\}$.
Note from $|C(u)| \le \Delta - 1$ that $|C \setminus C(u)| \ge 3$.
If there is a color $i_1 = c(u u_i) \in (S_2\cup S_3) \setminus (C(w) \setminus \{c(v w)\})$,
then $v w \to i_1$%
\footnote{Note that $i_1$ could be $c(v w)$ itself, but we recolor anyways.}
and $u v \to (C \setminus C(u))\setminus C(u_i)$ give rise to an acyclic edge $(\Delta + 2)$-coloring for $G$, and we are done.%
\footnote{Note that $(C \setminus C(u))\setminus C(u_i)$ is non-empty.}%
Otherwise, we proceed with the following proposition:

\begin{proposition}
\label{prop3101}
	$(C \setminus C(u)) \cup S_2 \cup S_3 \subseteq C(w) \setminus \{c(vw)\}$.
\end{proposition}

\subsubsection{Configuration ($A_{2.2}$)}
In this configuration, $n_2(u)+ n_3(u) = d(u) - 3$ and $n_3(u) \le 3$.

Let $c(u u_i) = i$ for each $1 \le i \le d(u) - 1$, in the acyclic edge $(\Delta + 2)$-coloring $c(\cdot)$ for $H = G - uv$.

From $n_2(u)+ n_3(u) = d(u) - 3$, we assume that $d(u_1), d(u_2), d(u_3) \ge 4$,
$d(u_j)= 3$ for $4\le j\le n_3(u) + 3$,
and $d(u_j)= 2$ for $n_3(u) + 4\le j \le d(u)- 1$.
It follows from Proposition~\ref{prop3101} that $C(w) = \{c(v w), 4, 5$, $\ldots, \Delta + 2\}$, where $v w$ can be arbitrarily colored by any $i \in \{1, 2, 3\}$;
thus $d(w) = \Delta$.
Using $C(u) = \{1, 2, \ldots, d(u) - 1\}$ and $C(v) = \{i\}$ in Proposition~\ref{prop3001}, we have $C \setminus C(u) \subseteq B_i$, for any $i \in \{1, 2, 3\}$.
Let $B^* = B_1 \cap B_2 \cap B_3$.
Then, for each color $j \in B^*$, we have $j \in C(u_1) \cap C(u_2) \cap C(u_3)$,
there is a neighbor $y$ of $w$ with $c(w y) = j$, and $y$ is incident with three edges colored $1, 2$, and $3$, respectively,
see for an illustration in Figure~\ref{fig03}.
Note that $C \setminus C(u) \subseteq B^*$.

When there exists $j = c(u u_j) \in S_2$ such that $j\not\in B_i$ for some $i \in \{1, 2, 3\}$,
if $c(u_j x_j) \not\in S_2\cup S_3$, then let $S = (C \setminus C(u)) \setminus C(u_j)$;
if $c(u_j x_j) = k \in S_2\cup S_3$, then let $S = (C \setminus C(u)) \setminus C(u_k)$.
It follows that $S \ne \emptyset$, and thus $u u_j \to S$ and $(u v, vw) \to (j, i)$ give rise to an acyclic edge $(\Delta+2)$-coloring for $G$, and we are done.%

In the other case, we have $B_i \subseteq (C(u_i) \setminus \{i\}) \cap (C(w) \setminus \{c(v w)\})$ for every $i \in \{1, 2, 3\}$,
and thus we proceed with the following proposition: 

\begin{proposition}
\label{prop3111}
	$(C \setminus C(u)) \cup S_2 \subseteq B^* \subseteq \left(\bigcap_{i = 1}^3 (C(u_i) \setminus \{i\})\right) \cap (C(w) \setminus \{c(v w)\})$.
\end{proposition}

If $S_3\subseteq B^*$, then from Assumption~$(*_{2.2})$ we have $(C \setminus C(u)) \cup S_2 \cup S_3 \subseteq B^*$.
Since $(C \setminus C(u)) \cup S_2 \cup S_3 = C \setminus \{1, 2, 3\}$,
we conclude that $B^* = C \setminus \{1, 2, 3\} = B_i = C(u_i) \setminus \{i\} = C(w) \setminus \{c(v w)\}$,
and consequently $C(u_i) = \{i, 4, 5, \ldots, \Delta+2\}$, for every $i \in \{1, 2, 3\}$.
See for an illustration in Figure~\ref{fig03} with $B^* = \{4, 5, \ldots, \Delta+2\}$,
from which we conclude that $w$ does not collide into any of $u_i$, $i = 1, 2, 3$, neither any of $u_j$, $j = 4, 5, \ldots, d(u) - 1$ (which have degree $2$ or $3$),
that is, $u w \not\in E(G)$.
Since $B^* = C \setminus \{1, 2, 3\}$, $wu_j\not\in E(G)$ for every $j\in \{1, 2, \ldots, d(u)- 1\}$.
It follows from Lemma~\ref{lemma10} that $G$ admits an acyclic edge $(\Delta + 2)$-coloring, and we are done.%

\begin{figure}[]
\begin{center}
\vskip -2pt
\includegraphics[width=3.5in]{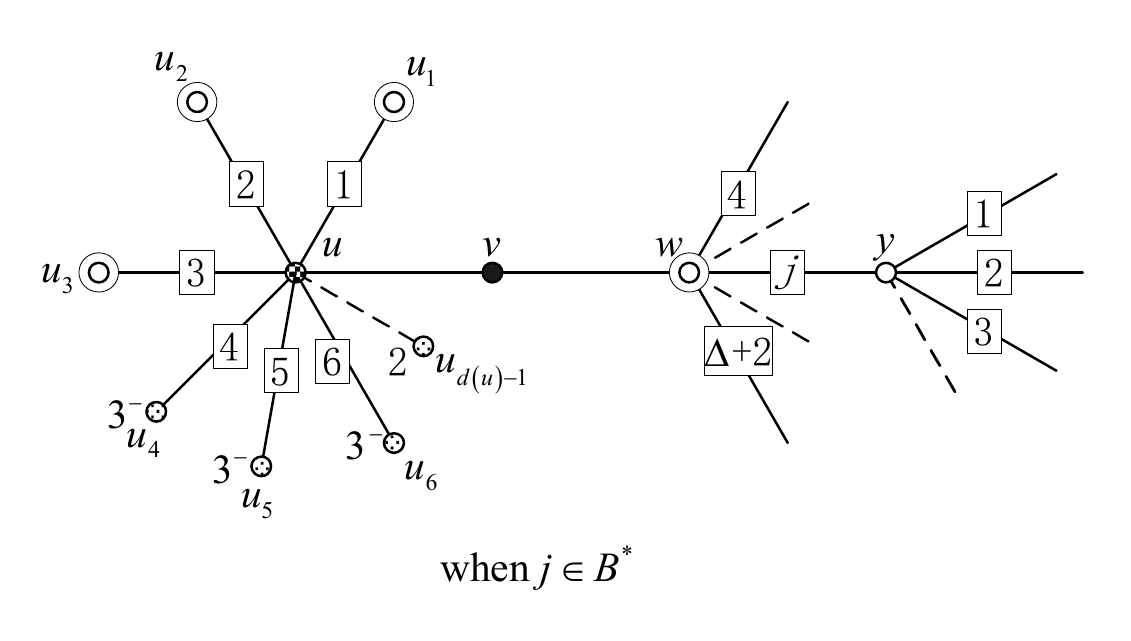}
\vskip 2pt
\caption{An illustration: for any color $j \in B^*$,
	there is an edge $wy$ colored $j$ and $y$ is incident with at least three edges colored $1, 2, 3$, respectively.
	Consequently, there are at least $|B^*|$ neighbors of $w$ each incident with at least three edges colored $1, 2, 3$, respectively.
	In the figure, a color is frame-boxed.\label{fig03}}
\end{center}
\end{figure}

Next we consider $S_3 \setminus B^* \ne \emptyset$, which is completely dealt with in the following Lemma~\ref{lemma12}.

\begin{lemma}
\label{lemma12}
{\em Condition~(C1)} states that $d(u) = \Delta = 6$, $n_3(u) = 2$ with $S_3 = \{4, 5\}$,
$C(u_4) = \{4, 5, \alpha_6\}$, and $C(u_5)= \{5, \alpha_7, \alpha_8\}$, where $\{\alpha_6, \alpha_7, \alpha_8\} = \{6, 7, 8\}$.

{\rm (1)} If (C1) holds, then $G$ admits an acyclic edge $(\Delta + 2)$-coloring.

{\rm (2)} If $S_3 \setminus B^* \ne \emptyset$, then either $G$ admits an acyclic edge $(\Delta + 2)$-coloring, or (C1) holds.
\end{lemma}
\begin{proof}
We first show that under Condition~(C1), $G$ admits an acyclic edge $(\Delta+2)$-coloring.%

In (C1), we have $S_2 = \emptyset$ due to $n_2(u) + n_3(u) = d(u) - 3$ and $H = G - u v$.

If the color $5 \not\in B_i$ for some $i \in \{1, 2, 3\}$, then $(u u_5, u v, v w) \to (\alpha_6, 5, i)$ gives an acyclic edge $(\Delta + 2)$-coloring for $G$.

In the other case, $5 \in B_i$ for every $i\in \{1, 2, 3\}$, that is, $|S_3 \cap B_i| \ge 1$.
By the definition of $B^*$, we have $|S_3 \cap B^*| \ge 1$ too.
It follows from Proposition~\ref{prop3111} that $\{5, 6, 7, 8\} \subseteq B^*$.
Since $w$ does not collide into any of $u_i$, $i = 1, 2, 3$, neither $u_4$ or $u_5$ (each has degree $3$), we conclude that $u w \not\in E(G)$;
and since $\{5, 6, 7, 8\} \subseteq B^*$, $w u_4, w u_5\not\in E(G)$.
If $w u_j\not\in E(G)$ for every $j\in \{1, 2, 3\}$, then by Lemma~\ref{lemma10} $G$ admits an acyclic edge $(\Delta + 2)$-coloring.
Otherwise we assume w.l.o.g. that $wu_1\in E(G)$;
then $c(wu_1) = 4$, $C(u_1) = C \setminus \{2, 3\}$, and thus the $(2, j)_{(u_2, w)}$-path, for any $j \in \{5, 6, 7, 8\}$, does not pass through $u_1$.
$(w u_1, v w, u v)\to (2, 4, \alpha_7)$ gives an acyclic edge $(\Delta + 2)$-coloring for $G$.

To prove the second item of the lemma, from Proposition~\ref{prop3111} we assume w.l.o.g. that the color $4 \in S_3 \setminus B_i$ for some $i\in \{1, 2, 3\}$.
Let $C(u_4)= \{4, a, b\}$.
Note that $C \setminus C(u) = \{d(u), d(u) + 1, \ldots, \Delta + 2\}$ contains at least three colors.
In all the following acyclic edge $(\Delta + 2)$-coloring schemes for $G$, $(v w, u v) \to (i, 4)$,
together with changing the colors of some other edges to be specified.

\begin{itemize}
\parskip=0pt
\item
	If $\{a, b\}\cap (S_2\cup S_3) = \emptyset$,
	then $u u_4 \to (C\setminus C(u))\setminus C(u_4)$ which, together with $(v w, u v) \to (i, 4)$ above,
	gives an acyclic edge $(\Delta + 2)$-coloring for $G$.
\item
	If $\{a, b\}\cap (S_2\cup S_3)= \{c(uu_s)\}\subseteq S_2$,
	then $u u_4 \to (C\setminus C(u))\setminus (C(u_4)\cup C(u_s))$ gives an acyclic edge $(\Delta + 2)$-coloring for $G$.
\item
	If $\{a, b\}\cap (S_2\cup S_3) =\{c(u u_s), c(u u_t)\}\subseteq S_2$,
	then $u u_4 \to (C\setminus C(u))\setminus (C(u_s)\cup C(u_t))$ gives an acyclic edge $(\Delta + 2)$-coloring for $G$.
\end{itemize}
In the other case, we have $\{a, b\} \cap S_3 \ne \emptyset$.
Since $4 \in S_3 \setminus B_i$, $|S_3| \ge 2$;
and we assume that $a = 5 \in S_3$, that is, $C(u_4) = \{4, 5, b\}$.
We distinguish the following three cases according to $b$.

{Case 1.} $b \in \{1, 2, 3\} \cup S_2$.
If $b \in \{1, 2, 3\}$, then $u u_4 \to (C\setminus C(u))\setminus C(u_5)$ gives an acyclic edge $(\Delta + 2)$-coloring for $G$.

When $b = j \in S_2$, we have $C(u_4) \subseteq C(u)$.
If there is a color $\alpha \in (C\setminus C(u)) \setminus (C(u_5)\cup C(u_j))$,
then $u u_4 \to \alpha$ gives an acyclic edge $(\Delta + 2)$-coloring for $G$.
Otherwise, we have $C\setminus C(u) \subseteq C(u_5)\cup C(u_j)$.
It follows that $d(u) = \Delta$, $C(u_4)= \{4, 5, j\}$, $C(u_5) = \{5, \Delta, \Delta+ 1\}$, and $C(u_j)= \{j, \Delta+ 2\}$.
Then $(u u_j, u u_4) \to (\Delta, \Delta + 2)$ gives an acyclic edge $(\Delta + 2)$-coloring for $G$.

{Case 2.} $b = 6 \in S_3$.
In this case we have $C(u_4) \subseteq C(u)$ too and $n_3(u)= 3$.
If there is a color $\alpha \in C\setminus C(u)$ such that $H$ does not contain any $(5, \alpha)_{(u_4, u_5)}$-path or any $(6, \alpha)_{(u_4, u_6)}$-path,
then $u u_4 \to \alpha$ gives an acyclic edge $(\Delta + 2)$-coloring for $G$.
Otherwise, for every color $j \in C \setminus C(u)$, $G$ contains a $(5, j)_{(u_4, u_5)}$-path or a $(6, j)_{(u_4, u_6)}$-path.
Consider the three colors of $\{\Delta, \Delta + 1, \Delta + 2\} \subseteq C \setminus C(u)$:
we assume w.l.o.g. that $C(u_5) = \{5, \Delta, \Delta + 1\}$, $C(u_6)= \{6, \Delta + 2, s\}$,
and there are a $(5, \Delta)_{(u_4, u_5)}$-path, a $(5, \Delta + 1)_{(u_4, u_5)}$-path, and a $(6, \Delta + 2)_{(u_4, u_6)}$-path in $H$.
If $s = c(u u_j)\in S_2$, then $u u_6 \to \{\Delta, \Delta + 1\}\setminus C(u_j)$;
otherwise, $u u_6 \to \{\Delta, \Delta + 1\}\setminus C(u_6)$.
Then, $u u_4 \to \Delta + 2$ gives an acyclic edge $(\Delta + 2)$-coloring for $G$.

{Case 3.} $b \in C \setminus C(u)$.
If there is a color $\alpha \in (C\setminus C(u))\setminus C(u_4)$ such that $H$ contains no $(5, \alpha)_{(u_4, u_5)}$-path,
then $u u_4\to \alpha$ gives an acyclic edge $(\Delta + 2)$-coloring for $G$.
Otherwise, for every $j\in (C\setminus C(u))\setminus C(u_4)$, $H$ contains a $(5, j)_{(u_4, u_5)}$-path.
It follows from $d(u_5) = 3$ that $d(u) = \Delta$, $C \setminus C(u) = \{\Delta, \Delta + 1, \Delta + 2\}$,
and we assume w.l.o.g. that $C(u_4)= \{4, 5, \Delta\}$, $C(u_5)= \{5, \Delta + 1, \Delta + 2\}$,
and there are a $(5, \Delta + 1)_{(u_4, u_5)}$-path and a $(5, \Delta + 2)_{(u_4, u_5)}$-path in $H$.
One sees that if $\Delta = 6$, then we have just arrived at Condition~(C1).

If $\Delta > 6$, then $(S_2\cup S_3)\setminus \{4, 5\}\ne\emptyset$.
When $n_3(u)= 3$ with $C(u_6) = \{6, s, t\}$,
if $s \in S_2$ and $t \in S_2$,
then $u u_6 \to \{\Delta, \Delta + 1, \Delta + 2\}\setminus (C(u_s) \cup C(u_t))$;
if $s \in S_2$ but $t \notin S_2$,
then $u u_6 \to \{\Delta, \Delta + 1, \Delta + 2\}\setminus (C(u_6) \cup C(u_s))$;
if $s \notin S_2$ and $t \notin S_2$,
then $u u_6 \to \{\Delta, \Delta + 1, \Delta + 2\}\setminus C(u_6)$.
When $n_3(u)= 2$ with $C(u_6)= \{6, s\}$,
if $s \in S_2$, then $u u_6 \to \{\Delta, \Delta + 1, \Delta + 2\}\setminus C(u_s)$;
if $s \not\in S_2$, then $u u_6 \to \{\Delta, \Delta + 1 , \Delta + 2\}\setminus C(u_6)$.
This re-coloring of $u u_6$, together with $(u u_4, v w, u v) \to (6, i, 4)$, gives an acyclic edge $(\Delta + 2)$-coloring for $G$.
\end{proof}

Lemma~\ref{lemma12} says that when $S_3 \setminus B^* \ne \emptyset$,
an acyclic edge $(\Delta + 2)$-coloring for $G$ can be obtained in $O(1)$ time from the inductively hypothesized acyclic edge $(\Delta + 2)$-coloring for $H = G - uv$,
or otherwise Condition~(C1) holds;
and if Condition~(C1) holds, then again an acyclic edge $(\Delta + 2)$-coloring for $G$ can be obtained in $O(1)$ time.
This finishes the inductive step for the case where $G$ contains the configuration ($A_{2.2}$).


\subsubsection{Configuration ($A_{2.3}$)}
In this configuration, $n_2(u)= d(u)- 4$, $d(u_j)= 2$ for $5 \le j \le d(u)- 1$, $n_2(u_4)\in \{d(u_4)- 4, d(u_4)- 5\}$, and $u_3u_4\in E(G)$.

From $n_2(u)= d(u) - 4$ and ($A_{2.1}$)--($A_{2.2}$), we assume $d(u_i) \ge 4$ for $1 \le i \le 4$.
Let $c(u u_i) = i$ for each $1 \le i \le d(u) - 1$, in the acyclic edge $(\Delta + 2)$-coloring $c(\cdot)$ for $H = G - uv$.
Since $u_3 u_4 \in E(G)$ and $G$ contains no intersecting triangles, we have $w, x_j\not\in N(u)$,
where $x_j$ is the other neighbor of $u_j$ other than $u$, for $j = 5, 6, \ldots, d(u) - 1$.

We prove the inductive step by first treating $u_1, u_2, u_3, u_4$ indistinguishably to obtain an acyclic edge $(\Delta + 2)$-coloring for $G$,
until impossible, by then (after Lemma~\ref{lemma13}) to distinguish $u_4$ from the other three vertices $u_1, u_2, u_3$.

Note that $S_2 = \{5, 6, \ldots, d(u)-1\}$ and $S_3 = \emptyset$.
From Proposition~\ref{prop3101} we have $(C \setminus C(u)) \cup S_2 = C \setminus \{1, 2, 3, 4\} \subseteq C(w) \setminus \{c(vw)\}$.
Therefore, $c(vw) \in \{1, 2, 3, 4\}$, $d(w) \ge \Delta - 1$, and $1 \le |C(w) \cap \{1, 2, 3, 4\}| \le 2$.
Since $vw$ can be colored by an arbitrary color $i \in C \setminus (C(w) \setminus c(vw))$,
we conclude that it can be colored by at least three of the four colors $\{1, 2, 3, 4\}$.
Then, for each such color $i$, by Proposition~\ref{prop3001} we have $C \setminus C(u) \subseteq B_i \subseteq C(u_i)$.

Let us tentatively un-color the edge $vw$;
later we will re-color it together with the edge $uv$ to present an acyclic edge $(\Delta + 2)$-coloring for $G$, in $O(1)$ time.

Consider first $d(w) = \Delta - 1$,
here we have $C(w) = C \setminus \{1, 2, 3, 4\}$ and $C \setminus C(u) \subseteq B_i \subseteq C(u_i)$ for every color $i \in \{1, 2, 3, 4\}$.

When there exists a color $j = c(u u_j) \in S_2 \setminus B_i$ for some $i \in \{1, 2, 3, 4\}$,
if $c(u_j x_j)\not\in S_2$, then let $S = (C\setminus C(u))\setminus C(u_j)$;
if $c(u_j x_j)= c(uu_k)\in S_2$, then let $S= (C\setminus C(u))\setminus C(u_k)$.
It follows that $|S| \ge 2$, and $u u_j \to S$ and $(u v, v w) \to (j, i)$ give rise to an acyclic edge $(\Delta + 2)$-coloring for $G$, and we are done.%

In the other case where $S_2 \setminus B_i = \emptyset$,
we have a stronger conclusion that $(C\setminus \{1, 2, 3, 4\}) \subseteq B_i\subseteq C(u_i)$ for every color $i\in \{1, 2, 3, 4\}$,
and consequently $\{1, 2, 3, 4\} \subseteq C(w_j)$ for every $w_j \in N(w)\setminus \{v\}$
(a similar illustration as shown in Figure~\ref{fig03} with $y$ now incident with at least four edges colored $1, 2, 3, 4$, respectively).
Since $|\{1, 2, 3, 4\} \cap C(u_j)| \le 2$ for every $j = 1, 2, \ldots, d(u)-1$, $u_j$ cannot be adjacent to $w$, that is, $w u_j\not\in E(G)$.
Together with $u w, u x_j\not\in E(G)$ for for every $j = 5, 6, \ldots, d(u)-1$,
it follows from Lemma~\ref{lemma10} that $G$ admits an acyclic edge $(\Delta + 2)$-coloring, and we are done.%

Consider next $d(w)= \Delta$, here we assume w.l.o.g. that $C(w)= C\setminus \{1, 2, 3\}$,
and its $\Delta - 1$ neighbors (other than $v$) are $w_4, w_5, \ldots, w_{\Delta+2}$ with $c(w w_j)= j$, $j = 4, 5, \ldots, \Delta+2$;
we have $C \setminus C(u) \subseteq B_i \subseteq C(u_i)$ for every color $i \in \{1, 2, 3\}$.
Let $B^*= B_1\cap B_2\cap B_3$;
then for each color $j = c(w w_j) \in B^*$, $w_j$ is incident with at least three edges colored $1, 2, 3$, respectively, and $d(w_j)\ge 4$
(a similar illustration as shown in Figure~\ref{fig03}).

\begin{lemma}
\label{lemma13}
When $d(w)= \Delta$ and assume w.l.o.g. that $C(w)= \{4, 5, \ldots, \Delta+2\}$,
either $G$ admits an acyclic edge $(\Delta + 2)$-coloring, or the following hold:
\begin{itemize}
\parskip=-4pt
\item[{(1)}]
	For every color $j= c(u u_j)\in S_2$, if $c(u_j x_j)\ne 4$ or $x_j\ne w$, then $j\in B^*$.
\item[{(2)}]
	$S_2\subseteq B^*$;
	or there exists exactly a color $j = c(u u_j)\in S_2$ such that $j\notin B^*$.
	In the latter case, $c(u_j x_j)= 4$, $x_j= w$, $C\setminus \{1, 2, 3, 4, j\}\subseteq C(u_4)$,
	and for every $k \in C\setminus \{1, 2, 3, 4, j\}$, there is a $(4, k)_{(u_4, w)}$-path, $\{1, 2, 3, 4\}\subseteq C(w_k)$, and $d(w_k)\ge 5$.
\item[{(3)}]
	$4\in B^*$;
	or $(((C\setminus C(u))\cup S_2)\cap B^*)\subseteq C(u_4)$.
\end{itemize}
\end{lemma}
\begin{proof}
To prove (1), assume to the contrary that $j\not\in B_i$ for some $i\in \{1, 2, 3\}$.
Firstly, if $c(u_j x_j)= 4$ and $x_j\ne w$,
then $( \{1, 2, 3\} \cup \{5, 6, \ldots, \Delta+2\} ) \setminus (\{j\} \cup C(x_j)) \ne \emptyset$.
We let $u_j x_j\to ( \{1, 2, 3\} \cup \{5, 6, \ldots, \Delta+2\} ) \setminus (\{j\} \cup C(x_j))$;
thus, it suffices to consider $c(u_j x_j) \ne 4$.

If $c(u_j x_j)\in \{1, 2, 3\}\cup (C\setminus C(u))$, then let $S= (C\setminus C(u))\setminus C(u_j)$;
if $c(u_j x_j) = k = c(u u_k)\in S_2$, then let $S= (C\setminus C(u))\setminus C(u_k)$.
It follows that $|S| \ge 2$;
$u u_j\to S$ and $(u v, v w)\to (j, i)$ give rise to an acyclic edge $(\Delta + 2)$-coloring for $G$, and we are done.%

To prove (2), we see from (1) that if there exists a color $j = c(u u_j)\in S_2 \setminus B^*$, then $c(u_j x_j)= 4$ and $x_j = w$.
Consequently there can be at most one such color.
Recalling that $C \setminus C(u) \subseteq B^*$,
we have $C\setminus \{1, 2, 3, 4, j\} = (C \setminus C(u)) \cup (S_2 \setminus \{j\}) \subseteq B^*$.

Denote $C^* = C\setminus \{1, 2, 3, 4, j\}$.
Assume that $j \not\in B_i$ for some $i\in \{1, 2, 3\}$.
Since we can recolor $u_j w\to i$, followed by a similar proof of (1),
if $G$ contains no $(4, k)_{(u_4, w)}$-path for some $k\in C^*$, then an acyclic edge $(\Delta + 2)$-coloring for $G$ can be achieved and we are done.%
\footnote{In the similar proof, the pair $(4, k)$ replaces $(i, j)$.} 
This proves that $G$ contains a $(4, k)_{(u_4, w)}$-path for every $k\in C^*$;
and thus $C^*\subseteq C(u_4)$, $\{1, 2, 3, 4\}\subseteq C(w_k)$ and $d(w_k) \ge 5$ for every $k\in C^*$.

To prove (3), one sees that in the latter case of (2), $(((C\setminus C(u))\cup S_2)\cap B^*)\subseteq (C \setminus \{1,2,3,4,j\}) \subseteq C(u_4)$;
therefore, it suffices to consider the former case of (2), that is, $S_2 \subseteq B^*$,
and consequently $((C\setminus C(u))\cup S_2)\cap B^* = (C\setminus C(u))\cup S_2 = C \setminus \{1,2,3,4\}$.

Assume that $4\not\in B_i$ for some $i\in \{1, 2, 3\}$, and there exists a $k\in ((C\setminus C(u))\cup S_2)\setminus C(u_4)$.
If $k= c(u u_k)\in S_2$, then letting $u u_k\to (C\setminus C(u))\setminus C(u_k)$ when $c(u_k x_k)\not\in S_2$,
or letting $u u_k\to (C\setminus C(u))\setminus C(u_\ell)$ when $c(u_k x_k)= c(u u_\ell)\in S_2$.
We thus consider below only the case where $k\in (C \setminus C(u)) \setminus C(u_4)$.
It follows that there is no $(k, i)_{(u_4, u_i)}$-path, since $k \notin C(u_4)$ and the $(k, i)_{(u, v)}$-path ends at $w$.
We re-color $(u u_4, u v, v w)\to (k, 4, i)$.

If there is no $(k, \ell)_{(u_4, u_\ell)}$-cycle for every $\ell \in S_2$, then this is an acyclic edge $(\Delta + 2)$-coloring for $G$, and we are done.%
%
In the other case, the re-coloring produces a $(k, \ell)_{(u_4, u_\ell)}$-cycle for some $\ell\in S_2$,
and by Lemma~\ref{lemma11} we can assume that there is exactly one such color $\ell$.
Subsequently, $c(u_\ell x_\ell) = k$ and $\ell \in C(x_\ell) \cap C(u_4)$.
We note that $x_\ell \ne w$ since $k \in B^*$.
Let $p$ be a color in $(C \setminus \{4\}) \setminus C(x_\ell)$, which is non-empty.
If $p\in (C\setminus C(u))\cup \{1, 2, 3\}$, then let $S= (C\setminus C(u))\setminus \{k, p\}$;
if $p\in S_2$, then let $S= (C\setminus C(u))\setminus (\{k\} \cup C(u_p))$.
It follows that $S\ne \emptyset$;
and the re-coloring $uu_\ell \to S$ and $u_\ell x_\ell \to p$ gives rise to an acyclic edge $(\Delta + 2)$-coloring for $G$, and we are done.%

This finishes the proof.
\end{proof}

The rest of the subsection deals with the un-settled cases described in Lemma~\ref{lemma13},
and we distinguish two cases on whether or not there is a color $j \in S_2 = \{5, 6, \ldots, d(u)-1\}$ such that $x_j = w$,
that is, the $2$-vertex $u_j$ is adjacent to $w$.

\paragraph{Case ($A_{2.3}$-1).} There exists a color $j \in S_2$ such that $x_j = w$ (i.e. $u_j \in N(w)$).

Since we do not differentiate $u_5, u_6, \ldots, u_{d(u)-1}$, we assume w.l.o.g. that $x_5 = w$.

From Lemma~\ref{lemma13}(2), we know that $|S_2 \setminus B^*| \le 1$ (or equivalently $|B^*| \ge \Delta - 3$);
recall that for each color $j = c(w w_j) \in B^*$, $d(w_j)\ge 4$ and $w_j$ is incident with at least three edges colored $1, 2, 3$, respectively.
At most two among $N(w) \setminus \{v\} = \{w_4, w_5, \ldots, \Delta + 2\}$ can have degree $2$.
Let $k = 1$ if there is another color $j \in S_2 \setminus \{5\}$ such that $x_j = w$, or $k = 0$ otherwise.
It follows that $\Delta \ge d(u)\ge 6 + k$.

{Case 1.1. $k= 1$.}
In this case, $\Delta \ge 7$.
By Lemma~\ref{lemma13}(2), it suffices to assume that $5 \notin B^*$;
that is, $w_4 = u_5$, $B^* = \{6, 7, \ldots, \Delta + 2\} \subseteq C(u_4)$, and $\{1, 2, 3, 4, j\}\subseteq C(w_j)$ for every $j\in B^*$.

Below we will take the edge between two of $\{u_1, u_2, u_3, u_4\}$ specified in ($A_{2.3}$) into consideration, say $u_{i_1} u_{i_2} \in E(G)$.
Denote the color-$j$ edge incident at $u_{i_1}$ ($u_{i_2}$, respectively) as $y_j u_{i_1}$ ($z_j u_{i_2}$, respectively), for every $j \in B^*$.
Then,
(i) $c(u_{i_1} u_{i_2})= \alpha\in \{1, 2, 3, 4, 5\}\setminus \{i_1, i_2\}$, or
(ii) $c(u_{i_1} u_{i_2})= \beta \in B^*$ and there is a color-$i_2$ (color-$i_1$, respectively) edge incident at $u_{i_1}$ ($u_{i_2}$, respectively),
denoted as $y_{i_2} u_{i_1}$ ($z_{i_1} u_{i_2}$, respectively).
We assume w.l.o.g. that $n_2(u_{i_1}) \in \{d(u_{i_1}) - 4, d(u_{i_1}) - 5\}$.
One sees that $d(u_{i_1})\ge \Delta- 1\ge 6$.
In (i) $|C(u_{i_1})\setminus (B^*\cup \{i_1, \alpha\})|\le 1$, and thus there is a color, say $j_0 \in B^*$, such that $d(y_{j_0})= 2$;
in (ii) $|C(u_{i_1})\setminus (B^*\cup \{i_1, i_2\})|\le 1$, and thus there is a color, say $j_0 \in (B^*\setminus \{\beta\})\cup \{i_2\}$, such that $d(y_{j_0})= 2$.
For every $y_j \in N(u_{i_1})$, if $d(y_j)= 2$ then let $y'_j\ne u_{i_1}$ denote its the other neighbor.

For (ii), if $d(y_{i_2})= 2$, then $c(y_{i_2} y'_{i_2})= \beta$.
Thus, $y_{i_2} u_{i_1}\to \{1, 2, 3, 4, 5\}\setminus (C(u_{i_1})\cup C(u_{i_2}))$ and $(u_5 w, u v, v w)\to (i_1, \beta, i_2)$
give rise to an acyclic edge $(\Delta + 2)$-coloring for $G$, and we are done.%

For (i) and (ii) where $j_0\in S_2$,
if $c(u_{j_0} x_{j_0})\not\in S_2$, then let $S = (C\setminus C(u))\setminus C(u_{j_0})$;
if $c(u_{j_0} x_{j_0})= c(u u_\ell)\in S_2$, then let $S = (C\setminus C(u))\setminus C(u_\ell)$.
It follows $S\ne \emptyset$ and letting $u u_{j_0}\to S$ reduces to the case where $j_0\in C\setminus C(u)$. }
Below we assume w.l.o.g. that $j_0 = \Delta$, i.e., $d(y_\Delta) = 2$.
It follows that $c(y_\Delta y'_\Delta)= i_1$.

If $5\not\in C(u_{i_1})$, then $(y_\Delta u_{i_1}, u v, v w)\to (5, \Delta, i_1)$ gives rise to an acyclic edge $(\Delta + 2)$-coloring for $G$, and we are done.%

If $5\in C(u_{i_1})$, then let $y_\Delta u_{i_1}\to \{1, 2, 3, 4\}\setminus (C(u_{i_1}) \cup \{i_2\})$ and $u u_{i_1}\to \Delta$.
Note that this does not create any $(\Delta, i)_{(u_{i_1}, u_i)}$-cycle for each $i\in \{1, 2, 3, 4\}$, due to $\Delta \in B^*$.

If there is no $(\Delta, j)_{(u_{i_1}, x_j)}$-cycle for every $j\in S_2$,
then $(u v, v w)\to (\Delta + 1, i_1)$ gives rise to an acyclic edge $(\Delta + 2)$-coloring for $G$, and we are done.%
%
Otherwise, by Lemma~\ref{lemma11} we can assume that there is exactly one color $j \in S_2$ such that there is a $(\Delta, j)_{(u_{i_1}, x_j)}$-cycle.
One sees that $c(u_j x_j)= \Delta$ and $j\in C(x_j)$, and thus $j \ne 5$ and $x_j \ne w$.
Note that there exists a color $p\in (C\setminus\{i_1\})\setminus C(x_j)$.
If $p\in C\setminus S_2$, then let $S = \{\Delta + 1, \Delta + 2\}\setminus \{p\}$;
if $p= c(u u_\ell)\in S_2$, then let $S = \{\Delta + 1, \Delta + 2\}\setminus C(u_\ell)$.
It follows that $S\ne \emptyset$;
$u_j x_j\to p$, $u u_j\to q\in S$, $uv\to \{\Delta + 1, \Delta + 2\}\setminus \{q\}$ and $vw\to i_1$
give rise to an acyclic edge $(\Delta + 2)$-coloring for $G$, and this finishes the discussion for Case 1.1.%

{Case 1.2.} $k= 0$.
In this case, $\Delta\ge 6$, and if $S_2 \setminus B^* \ne \emptyset$ then by Lemma~\ref{lemma13}(2) $S_2 \setminus B^* = \{5\}$.
That is, $\{6, 7, \ldots, \Delta+2\} \subseteq B^*$, and $G$ contains an $(i, j)_{(u_i, w_j)}$-path for every $i\in \{1, 2, 3\}$ and $j\in B^*$.

If there exists an $i\in \{1, 2, 3\}\setminus C(w_5)$,
then $(w w_5, u v, v w)\to (i, \Delta, 5)$ gives rise to an acyclic edge $(\Delta + 2)$-coloring for $G$, and we are done.%
Otherwise, $\{1, 2, 3\}\subset C(w_5)$ and thus $d(w_j) \ge 4$ for every $j\in \{5, 6, \ldots, \Delta+2\}$.
One sees that $w_5 \ne u_4$;
$w_5\notin \{u_1, u_2, u_3\}$ since $|\{1, 2, 3, 5\} \cup B^*| \ge \Delta+1$.
We conclude that $w_5 \notin N(u)$.

Since $4 \notin B^*$, by Lemma~\ref{lemma13}(3) we have $((C\setminus C(u))\cup S_2)\cap B^*)\subseteq C(u_4)$.
If $5\in B^*$, then $5\in C(u_4)$ too;
if $5\not\in B^*$, then by Lemma~\ref{lemma13}(2) $4\in C(w_j)$ for each $j\in \{6, 7, \ldots, \Delta + 2\}$.
Consequently, from $\{1, 2, 3, j\}\subseteq N(w_j)$ and $d(w_j) \le \Delta$ we conclude that $w_j \notin N(u)$ for each $j\in \{6, 7, \ldots, \Delta + 2\}$.

That is, $w_j \notin N(u)$ for each $j\in \{5, 6, \ldots, \Delta + 2\}$.
Next, let $G' = G - \{v, u_5\} + u w$, which is a planar graph without intersecting triangles, $\Delta(G') \le \Delta$, and contains three less edges than $G$.
By the inductive hypothesis, $G'$ has an acyclic edge $(\Delta + 2)$-coloring $c(\cdot)$ using the color set $C = \{1, 2, \ldots, \Delta + 2\}$,
in which we assume that $c(u w) = 1$, $2\not\in C(u)$, and $3\not\in C(w)$.
We keep the color for every edge of $E(G')$ except $u w\notin E(G)$,
and color the other edges of $E(G)$ as $(u v, v w, u u_5, u_5 w)\to (2, 1, 1, 3)$.
This does not introduce any bichromatic cycle, that is, the achieved edge $(\Delta + 2)$-coloring for $G$ is acyclic.%

\paragraph{Case ($A_{2.3}$-2).} $x_5, x_6, \ldots, x_{d(u)-1}$ and $w$ are distinct from each other.

By Lemma~\ref{lemma13}(2), $\{5, 6, \ldots, \Delta+2\} \subseteq B^*$ and $\{1, 2, 3, j\}\subseteq C(w_j)$ for every $j \in B^*$.
It follows that $w_j\not\in N(u)\setminus\{u_4\}$ for every $j \in B^*$.
Below we do not distinguish the colors of $\{5, 6, \ldots, \Delta+2\}$, and refer them to as $\alpha_5, \alpha_6, \ldots, \alpha_{(\Delta+2)}$, respectively.%
\footnote{That is, $(\alpha_5, \alpha_6, \ldots, \alpha_{(\Delta+2)})$ is an arbitrary permutation of $(5, 6, \ldots, \Delta+2)$.}
The edge between two of $\{u_1, u_2, u_3, u_4\}$ specified in ($A_{2.3}$) have two possibilities:
(Case 2.1) it is between two of $\{u_1, u_2, u_3\}$ and we assume w.l.o.g. the edge is $u_1 u_2$,
or (Case 2.2) it is between $u_4$ and one of $\{u_1, u_2, u_3\}$ and we assume w.l.o.g. the edge is $u_1 u_4$.

{Case 2.1.} $u_1 u_2\in E(G)$.
In this case, denote the color-$\alpha_j$ edge incident at $u_1$ ($u_2$, respectively) as $y_j u_1$ ($z_j u_2$, respectively),
for every $j \in \{6, 7, \ldots, \Delta+2\}$.
Similar to Case 1.1, either (i) $c(u_1 u_2)\in \{3, 4\}$ and $c(y_5 u_1) = c(z_5 u_2) = \alpha_5$,
or (ii) $c(u_1 u_2)= \alpha_5$, $c(y_5 u_1) = 2$, and $c(z_5 u_2) = 1$.
One sees that in either case both $u_1$ and $u_2$ have reached degree $\Delta \ge 6$ (by the non-existence of ($A_{1.1}$)) and thus $4 \notin B_i$ for $i = 1, 2$;
by Lemma~\ref{lemma13}(3), we have $B^* = \{5, 6, \ldots, \Delta+2\} \subseteq C(u_4)$.

It follows that $w_j \ne u_4$, implying $w_j \notin N(u)$ for every $j \in B^*$, and $w_4 \notin \{u_1, u_2, u_4\}$.
When $w_4 \ne u_3$,
$w u_j \notin E(G)$ for every $j = 1, 2, \ldots, d(u)-1$;
it follows from Lemma~\ref{lemma10} that $G$ admits an acyclic edge $(\Delta + 2)$-coloring, and we are done.%
%
When $w_4 = u_3$,
$C(u_3)= C\setminus \{1, 2\}$.

If there is no $(4, j)_{(u_4, w_j)}$-path for some $j\in B^*$,
then $(w u_3, uv, vw)\to (1, j, 4)$ gives rise to an acyclic edge $(\Delta + 2)$-coloring for $G$, and we are done.%
Otherwise, there is a $(4, j)_{(u_4, w_j)}$-path for every $j\in B^*$.
Note that the prior $u_1$ and $u_2$ are indistinguishable to us, and we assume w.l.o.g. below that $n_2(u_1) \in \{d(u_1)-4, d(u_1)-5\}$.
Since $n_2(u_1) \ge 1$, we assume that $d(y_{i_0})= 2$ for some $i_0$ and let $y'_{i_0} \ne u_1$ denote its the other neighbor.
In (i), $i_0\in B^*$, $c(y'_{i_0} y_{i_0})= 1$, and we let $u_1 y_{i_0}\to 2$.
In (ii), if $i_0= 5$, then $c(y'_{i_0} y_{i_0})= \alpha_5$, and we let $(u_1 y_{i_0}, v w)\to (3, 2)$;
if $i_0\in \{6, 7, \ldots, \Delta + 2\}$, then $c(y'_{i_0} y_{i_0})= 1$, and we let $u_1 y_{i_0}\to 3$.

If $\alpha_{i_0}\notin C(u)$, then $u v\to \alpha_{i_0}$ gives rise to an acyclic edge $(\Delta + 2)$-coloring for $G$, and we are done.%
Otherwise, $\alpha_{i_0}= c(uu_j)\in S_2$.
If $c(u_j x_j)\not\in S_2$, then let $S= \{\Delta, \Delta + 1, \Delta + 2\}\setminus C(u_j)$;
if $c(u_j x_j)= c(u u_\ell)\in S_2$, then let $S= \{\Delta, \Delta + 1, \Delta + 2\}\setminus C(u_\ell)$.
It follows that $S\ne \emptyset$,
and $u u_j\to S$ and $uv\to \alpha_{i_0}$ give rise to an acyclic edge $(\Delta + 2)$-coloring for $G$, and we are done.%

{Case 2.2.} $u_1 u_4\in E(G)$.
If there exists $i\in \{2, 3\}\setminus C(w_4)$, then we recolor $(ww_4, vw)\to (i, 4)$ to reduce the discussion to the above Case 2.1, and we are done.%

We thus consider below $\{2, 3, 4\}\subseteq C(w_4)$, which implies that $w_4 \notin N(u)$.

If $4 \notin B^*$, then by Lemma~\ref{lemma13}(3) $B^* = \{5, 6, \ldots, \Delta+2\} \subseteq C(u_4)$;
thus $w_j \notin N(u)$ for every $j \in \{5, 6, \ldots, \Delta+2\}$ either.
By Lemma~\ref{lemma10}, we obtain an acyclic edge $(\Delta + 2)$-coloring for $G$, and we are done.%

If $4 \in B^*$, that is, $B^* = \{4, 5, \ldots, \Delta+2\}$,
then we denote the color-$\alpha_j$ edge incident at $u_1$ as $y_j u_1$, for every $j \in \{6, 7, \ldots, \Delta+2\}$,
and then $c(u_1 u_4) = \alpha_5$, $c(z_5 u_4) = 1$, and $c(y_5 u_1) = 4$.
We further assume that $w_{\alpha_6}= u_4$ (otherwise by Lemma~\ref{lemma10} we are done%
), and $\{1, 2, 3, 4, \alpha_5, \alpha_6\}\subseteq C(u_4)$ with $c(z_6 u_4)= 2$, $c(z_7 u_4)= 3$.

First assume that there exists an $i_0$ such that $d(y_{i_0})= 2$ for some $i_0\in \{5, 6, \ldots, \Delta+2\}$.
If $i_0 = 5$, then $c(y'_5 y_5)= 1$;
the re-coloring $y_5 u_1\to 2$ makes $4\not\in B_1$, which together with $\{5, 6, \ldots, \Delta+2\}\setminus C(u_4) \ne \emptyset$
enables us to obtain an acyclic edge $(\Delta + 2)$-coloring of $G$ by Lemma~\ref{lemma13}(3).%
If $i_0\in \{6, 7, \ldots, \Delta+2\}$, then $c(y'_{i_0} y_{i_0})= 1$;
we let $(y_{i_0} u_1, u v, v w)\to (2, \alpha_{i_0}, 1)$.
If $\alpha_{i_0}\notin C(u)$, then this is an acyclic edge $(\Delta + 2)$-coloring for $G$, and we are done.%
Otherwise, $\alpha_{i_0}= c(u u_j)\in S_2$.
One sees that there exists a color $p\in (C\setminus \{4\})\setminus C(x_j)$.
If $p\not\in S_2$, then let $S = \{\Delta, \Delta + 1, \Delta + 2\}\setminus \{p\}$;
if $p = c(u u_\ell)\in S_2$, then let $S= \{\Delta, \Delta + 1, \Delta + 2\}\setminus C(u_\ell)$.
It follows that $S\ne \emptyset$;
$u_j x_j\to p$ and $u u_j\to S$ give rise to an acyclic edge $(\Delta+2)$-coloring for $G$, and we are done.%

Next assume that $d(y_j)\ge 3$ for every $j \in \{5, 6, \ldots, \Delta+2\}$.
It follows from $d(u_4) \ge 6$ that $n_2(u_4)\in \{d(u_4)- 4, d(u_4)- 5\}$ is at least $1$.
Since $\{1, \alpha_5, \alpha_6\}\subseteq C(z_5)$, at least one of $z_6$ and $z_7$ is a $2$-vertex, and we assume w.l.o.g. $d(z_6) = 2$.
It follows that $c(z_6 z'_6) = \alpha_6$.
If $\alpha_6\in C\setminus C(u)$, then let $S= C\setminus (C(u) \cup \{\alpha_6\})$;
if $\alpha_6= c(uu_j)\in S_2$, then let $S= C\setminus (C(u) \cup C(u_j))$.
Note that $S \ne \emptyset$.
Then $z_6 u_4\to C\setminus C(u_4)$, $(w u_4, v w)\to (2, \alpha_6)$ and $uv\to S$ give rise to an acyclic edge $(\Delta + 2)$-coloring for $G$, and we are done.%

This finishes the inductive step for the case where $G$ contains the configuration ($A_{2.3}$).


\subsubsection{Configuration ($A_{2.4}$)}
In this configuration, $n_2(u) = d(u)- 5$, $d(u_j) = 2$ for $6 \le j \le d(u)-1$, $d(u_5)= 3$, $d(u_4)= 4$, $d(u_3)\le 5$, and $u_3 u_4\in E(G)$.

From $n_2(u)= d(u) - 5$ and ($A_{2.1}$)--($A_{2.2}$), we have $n_3(u) = 1$ and thus $d(u_i) \ge 4$ for $1 \le i \le 3$.
Let $c(u u_i)= i$ for each $1\le i\le d(u)- 1$.
By Propositions~\ref{prop3001} and \ref{prop3101}, $\{5, 6, \ldots, \Delta+2\} = (C\setminus C(u)) \cup S_2\cup S_3\subseteq C(w)$ and $c(v w) \in \{1, 2, 3, 4\}$.
If $c(v w)\in \{1, 2\}$, then we recolor $vw\to \{3, 4\}\setminus C(w)$;
it suffices to deal with the following two subcases.

\paragraph{Case ($A_{2.4}$-1).} $c(vw)= 4$.

In this case, if $C(u) \cup C(u_4) \ne C$, then $uv\to C\setminus (C(u) \cup C(u_4))$ gives rise to an acyclic edge $(\Delta + 2)$-coloring for $G$, and we are done;%
we consider below $d(u) = \Delta$, $C(u_4) = \{4, \Delta, \Delta + 1, \Delta + 2\}$ with $c(u_3 u_4)= \Delta$,
and furthermore $4\in C(u_3)$, $C(w) = C\setminus \{1, 2\}$,
and $G$ contains an $(i, j)_{(u_i, w)}$-path for every pair $i\in \{1, 2\}$ and $j\in \{\Delta, \Delta + 1, \Delta + 2\}$.

Let $C(u_5)= \{5, a, b\}$.
When $3\not\in C(u_5)$, we can set up a color set $S$ as follows:
If $\{a, b\}\cap S_2 = \emptyset$, then let $S = \{\Delta, \Delta + 1, \Delta + 2\}\setminus C(u_5)$;
if $\{a, b\}\cap S_2 = \{a\}$, then let $S = \{\Delta, \Delta + 1, \Delta + 2\}\setminus (C(u_5)\cup C(u_a))$;
if $\{a, b\}\cap S_2 = \{a, b\}$, then let $S = \{\Delta, \Delta + 1, \Delta + 2\}\setminus (C(u_a)\cup C(u_b))$.
It follows that $S\ne \emptyset$,
and $u v\to 5$ and $u u_5\to S$ give rise to an acyclic edge $(\Delta + 2)$-coloring for $G$, and we are done.%

When $3\in C(u_5)$, we consider the following two subcases on whether $\Delta\in C(u_5)$ or not.

{Case 1.1.} $\Delta\not\in C(u_5)$.
We first re-color $(u u_5, u v)\to (\Delta, 5)$.
If there is no $(\Delta, j)_{(u_5, x_j)}$-cycle for every $j\in S_2$, then we are done.%
Otherwise, $T_\Delta = \{j\in S_2 \mid G \mbox{ contains a } (\Delta, j)_{(u_5, x_j)}\mbox{-cycle}\} \ne \emptyset$.
By Lemma~\ref{lemma11}, we can assume that there is exactly one color in $T$ and w.l.o.g. $T_\Delta = \{6\}$;
that is, $c(u_6 x_6)= \Delta$ and $C(u_5)= \{5, 3, 6\}$.
Then, $u u_6\to \Delta + 1$ gives rise to an acyclic edge $(\Delta + 2)$-coloring for $G$, and we are done.%

{Case 1.2.} $\Delta\in C(u_5)$ and thus $C(u_5) = \{5, 3, \Delta\}$.
If there exists a color $\alpha\in \{\Delta + 1, \Delta + 2\}\setminus C(u_3)$,
then $(u u_5, u v)\to (\alpha, 5)$ gives rise to an acyclic edge $(\Delta + 2)$-coloring for $G$, and we are done.%
Otherwise, $C(u_3) = \{3, 4, \Delta, \Delta + 1, \Delta + 2\}$.
Then, $(u u_5, u u_3, u v)\to (\Delta + 1, 5, 3)$ gives rise to an acyclic edge $(\Delta + 2)$-coloring for $G$, and we are done.%

\paragraph{Case ($A_{2.4}$-2).} $c(v w) = 3$, $d(u_3)= 5$ (by symmetry) and $C(w)= C\setminus \{1, 2\}$.

In this case, $G$ contains an $(i, j)_{(u_i, w)}$-path for every pair $i\in \{1, 2, 3\}$ and $j\in C\setminus C(u)$.
Therefore, if $c(u_3 u_4)\in C\setminus C(u)$, then $3\in C(u_4)$;
either way, $(C\setminus C(u))\setminus C(u_4)\ne\emptyset$.

When there exists a color $j = c(u u_j)\in S_2$ such that $G$ contains no $(i, j)_{(u_i, w)}$-path for some $i\in \{1, 2, 3\}$,
if $c(u_j x_j)\not\in S_2\cup \{4, 5\}$, then let $S = (C\setminus C(u))\setminus C(u_j)$;
or if $c(u_j x_j)= c(u u_k)\in S_2\cup \{4, 5\}$, then let $S = (C\setminus C(u))\setminus C(u_k)$.
Then $u v\to j$, $v w\to i$, and $u u_j\to S$ give rise to an acyclic edge $(\Delta + 2)$-coloring for $G$.%

When $S_2 = \emptyset$, or $G$ contains an $(i, j)_{(u_i, w)}$-path for every pair $i\in \{1, 2, 3\}$ and $j\in S_2$,
we conclude from $d(u_3)= 5$ that $3 \le |(C\setminus C(u))\cup S_2|\le 4$.
By a similar discussion as in Case ($A_{2.4}$-1), if $4\not\in C(u_5)$, then $G$ contains a $(3, 5)_{(u_3, w)}$-path.
We distinguish the two cases for $|(C\setminus C(u))\cup S_2|$.

{Case 2.1.} $|C\setminus C(u)|= 4$ or $S_2 = \{6\}$.
In this case, $\Delta = 7$, $4\not\in C(u_3)$, $3\in C(u_4)$, and $C(u_5) = \{5, 4, b\}$.
One sees that $5\not\in B_3$.

When $c(u_3 u_4)\in S_2$,
if $b\not\in S_2$, then let $S = (C\setminus C(u))\setminus (C(u_5)\cup C(u_4))$;
or if $b = c(u u_j)\in S_2$, then let $S = (C\setminus C(u))\setminus (C(u_4)\cup C(u_j))$.
It follows that $S\ne \emptyset$,
and thus $u v\to 5$ and $u u_5\to S$ give rise to an acyclic edge $(\Delta + 2)$-coloring for $G$.%

When $c(u_3 u_4)\notin S_2$, and assume w.l.o.g. that $c(u_3 u_4)= \Delta$,
if $b = \Delta$, then $u u_5\to (C\setminus C(u))\setminus C(u_4)$;
or if $b\in S_2$ and $c(u_b x_b)= \Delta$, then $u u_5\in (C\setminus C(u))\setminus (C(u_4)\cup C(u_b))$;
otherwise, $u u_5\to \Delta$.
This together with $u v\to 5$ gives rise to an acyclic edge $(\Delta + 2)$-coloring for $G$, and we are done.%

{Case 2.2.} $|C\setminus C(u)|= 3$ and $S_2 = \emptyset$.
In this case, $d(u) = \Delta = 6$.

{Case 2.2.1.} $G$ contains a $(3, 5)_{(u_3, w)}$-path and let $C(u_4) = \{3, 4, c(u_3 u_4), b\}$.
Then $C(u_3)= \{3, 5, 6, 7, 8\}$.

(1) $c(u_3 u_4)= 6$.
If $b\ne 5$, then $u u_4\to \{7, 8\}\setminus C(u_4)$ and $u v\to 4$ give rise to an acyclic edge $(\Delta + 2)$-coloring for $G$.%
If $b = 5$ and $\{7, 8\}\setminus C(u_5)\ne\emptyset$,
then $u u_4\to \{7, 8\}\setminus C(u_5)$ and $uv\to 4$ give rise to an acyclic edge $(\Delta + 2)$-coloring for $G$.%
Otherwise, we have $b = 5$ and $C(u_5) = \{5, 7, 8\}$.

(2) $c(u_3 u_4)= 5$.
If there exists a color $\alpha\in \{6, 7, 8\}\setminus (C(u_4)\cup C(u_5))$,
then $(uu_4, uv)\to (\alpha, 4)$ gives rise to an acyclic edge $(\Delta + 2)$-coloring for $G$.%
Otherwise, $\{6, 7, 8\}\subseteq C(u_4)\cup C(u_5)$ and we assume w.l.o.g. that $C(u_4)= \{3, 4, 5, 6\}$, $C(u_5)= \{5, 7, 8\}$.

In the leftover situation of (1) and (2), we may further assume that $G$ contains an $(i, 5)_{(u_i, w)}$-path for every $i\in \{1, 2\}$.
Note that $wu_i\not\in E(G)$ for every $i\in \{3, 4, 5\}$.
By Lemma~\ref{lemma10}, $w u_i\in E(G)$ for some $i\in \{1, 2\}$ and we assume w.l.o.g. that $w u_1\in E(G)$ with $c(w u_1)= 4$.
Then $(w u_1, v w, u v)\to (2, 4, 7)$ gives rise to an acyclic edge $(\Delta + 2)$-coloring for $G$.%

{Case 2.2.2.} $G$ contains no $(3, 5)_{(u_3, w)}$-path and it suffices to assume that $4\in C(u_5)$.
Furthermore, when there exists a color $\alpha\in \{6, 7, 8\}\setminus (C(u_4)\cup C(u_5))$,
$(u u_5, u v)\to (\alpha, 5)$ gives rise to an acyclic edge $(\Delta + 2)$-coloring for $G$.%
We consider below the case where $\{6, 7, 8\} \subseteq C(u_4)\cup C(u_5)$, and discuss the two possible values of $c(u_3 u_4)$.

(1) $c(u_3 u_4)\not\in \{6, 7, 8\}$.
It follows that $C(u_5) = \{4, 5, 8\}$ and $6, 7\in C(u_4)\setminus \{c(u_3 u_4)\}$.
Then $(u u_3, u u_4)\to (4, 3)$ and $u v\to 6$ give rise to an acyclic edge $(\Delta + 2)$-coloring for $G$.%

(2) $c(u_3 u_4)\in \{6, 7, 8\}$ and we assume w.l.o.g. $c(u_3 u_4)= 6$.
It follows that $C(u_4)= \{3, 4, 6, 7\}$ and $C(u_5)= \{5, 4, 8\}$.
If $G$ contains no $(i, 4)_{(u_i, w)}$-path for some $i\in \{1, 2, 3\}$,
then $(u u_4, u v)\to (8, 4)$ and $v w\to i$ give rise to an acyclic edge $(\Delta + 2)$-coloring for $G$.%
Otherwise, $G$ contains an $(i, 4)_{(u_i, w)}$-path and $4\in C(u_i)$ for every $i\in \{1, 2, 3\}$,
which implies that $C(u_3)= \{3, 4, 6, 7, 8\}$ and $4\in C(u_1)\cap C(u_2)$.
Note that $w u_j\not\in E(G)$ for every $j\in \{3, 4, 5\}$.
By Lemma~\ref{lemma10}, $w u_i\in E(G)$ for some $i\in \{1, 2\}$, and we assume w.l.o.g. that $w u_1\in E(G)$ with $c(w u_1)= 5$.
Then $(w u_1, v w, u v)\to (2, 5, 7)$ gives rise to an acyclic edge $(\Delta + 2)$-coloring for $G$.%

This finishes the inductive step for the case where $G$ contains the configuration ($A_{2.4}$),
and thus completes the discussion for the configuration ($A_2$).

The discussion for the configurations ($A_3$), ($A_5$) and ($A_6$) follows the same scheme to identify certain local edge coloring,
followed by re-coloring of a minimal subset of edges and at the same time assigning a color for the edge $u v$,
to give rise to an acyclic edge $(\Delta + 2)$-coloring for $G$.
We show, as in the above, that there are only a finitely many situations, and in each situation such a process takes $O(1)$ time.
We hope that the above detailed analysis for the configuration ($A_2$) has convinced the readers the correctness,
with the similar analyses for the configurations ($A_3$), ($A_5$) and ($A_6$) moved to the Appendix,
and so that Theorem~\ref{thm01} is proved.


\section{Conclusions}
Proving the acyclic edge coloring conjecture (AECC) on general graphs has a long way to go,
but on planar graphs it is very close with recent efforts.
In particular, the AECC has been confirmed true for planar graphs without $i$-cycles, for each $i \in \{3, 4, 5, 6\}$.
In this paper, we dealt with the planar graphs containing $3$-cycles but no intersecting $3$-cycles.
Prior to our work, their acyclic chromatic index was shown to be at most $\Delta + 3$,
and in this work we closed the AECC by confirming it affirmatively.
The main technique is a discharging method to show at least one of the several specific local structures must exist in the graph,
followed by recoloring certain edges in each such local structure and by an induction on the number of edges in the graph.

Our detailed construction and proofs are built on top of several existing works,
and we believe our new design of the local structures and the associated discharging method can be helpful for proving the AECC on general planar graphs.
On the other hand, for general graphs, improvements from the approximation algorithm perspective,
that is, to reduce the coefficient for $\Delta$ in the upper bound (from the currently best $3.74$),
are impactful pursuit.

\subsection*{Acknowledgements.}
QS is supported by the ZJNSF Grants LY20F030007 and LQ15A010010 and the NSFC Grant 11601111.
GL is supported by the NSERC Canada.
EM is supported by the KAKENHI Grant JP17K00016 and the JST CREST JPMJR1402.



\newpage
\section*{Appendix: Inductive proofs for the configurations $A_3$, $A_5$ and $A_6$}
\setcounter{section}3
\setcounter{subsection}{1}
\subsection{Configuration ($A_3$)}
In this subsection we prove the inductive step for the case where $G$ contains the configuration ($A_3$), that is,
\begin{description}
\parskip=0pt
\item[$(A_3)$]
	A path $u_2 u v$ with $d(u)= 3$, $d(v)= 4$.
	Let $w_1, w_2, \ldots, w_{d(u_2)-1}$ be the neighbors of $u_2$ other than $u$, sorted by their degrees.
    At least one of the configurations ($A_{3.1}$)--($A_{3.3}$) occurs.
\end{description}

Let $v_1, v_2, v_3$ be the neighbors of $v$ other than $u$, and assume that $c(v v_i)= i$ for $i= 1, 2, 3$;
let $u_1$ be the neighbor of $u$ other than $u_2$ and $v$.
Before we discuss the detailed configurations ($A_{3.1}$)--($A_{3.3}$),
we do not distinguish $u_1$ and $u_2$ but use $x_1$ to refer to any one of $u_1$ and $u_2$ while $x_2$ to refer to the other ($\{x_1, x_2\} = \{u_1, u_2\}$).


Denote $C^*= C\setminus\{1, 2, 3\}$.
One sees that $|C^*| = \Delta - 1\ge 4$,
from the inexistence of ($A_{1.3}$) that $\{x_1, x_2\}\cap \{v_1, v_2, v_3\}= \emptyset$,
and from Proposition~\ref{prop3001} that $1\le|C(u)\cap C(v)|\le 2$.

\begin{lemma}
\label{lemma14}
If $|C(u)\cap C(v)| = 2$, then in $O(1)$ time either an acyclic edge $(\Delta + 2)$-coloring for $G$ can be obtained,
or an acyclic edge $(\Delta + 2)$-coloring for $H$ can be obtained such that $|C(u)\cap C(v)| = 1$.
\end{lemma}
\begin{proof}
Assume w.l.o.g. that $(u x_1, u x_2)_c= (1, 2)$.

When there exists a color $j\in C^*\setminus C(x_i)$ for some $i \in \{1, 2\}$,
$u x_i\to j$ gives rise to an acyclic edge $(\Delta + 2)$-coloring for $H$ such that $|C(u)\cap C(v)| = 1$.

In the other case, $C(x_1)= C^*\cup \{1\}$ and $C(x_2) = C^*\cup \{2\}$.
Since $4\in B_1\cup B_2$, we assume w.l.o.g. that $4\in B_1$.
If $G$ contains no $(j, 4)_{(x_1, v_j)}$-path for some $j\in \{2, 3\}$,
then $(u x_2, u x_1, u v)\to (1, j, 4)$ gives rise to an acyclic edge $(\Delta + 2)$-coloring for $G$.%
Otherwise, $G$ contains a $(j, 4)_{(x_1, v_j)}$-path and $4\in C(v_j)$ for every $j\in \{1, 2, 3\}$.
Similarly, we have $C^*\subseteq C(v_j)$ and thus $C(v_j) = C^*\cup \{j\}$ for every $j\in \{1, 2, 3\}$.
Then by Lemma~\ref{lemma09}, $(u x_1, u x_2)\to (2, 1)$, $(v v_2, v v_3)\to (3, 2)$, and $u v\to 4$ give rise to an acyclic edge $(\Delta + 2)$-coloring for $G$.%
\end{proof}

By Lemma \ref{lemma14}, hereafter we assume that $|C(u)\cap C(v)| = 1$, and w.l.o.g. that $C(u) = \{1, 4\}$, $c(u x_1) = 1$, $c(u x_2) = 4$,
and further from Proposition~\ref{prop3001} that $(C\setminus \{1, 2, 3, 4\})\subseteq B_1 \subseteq C(x_1)\cap C(v_1)$;
it follows that $\{2, 3\}\setminus C(x_1)\ne\emptyset$.

When $G$ contains no $(1, 4)_{(x_1, v_1)}$-path, $u v\to 4$;
if $C^*\setminus C(x_2)\ne \emptyset$, then $u x_2\to C^*\setminus C(x_2)$,
or if $1\not\in C(x_2)$, then $u x_2\to \{2, 3\}\setminus C(x_2)$,
or if $C(x_2)= C\setminus \{2, 3\}$, then $u x_2\to \{2, 3\}\setminus C(x_1)$.
This gives rise to an acyclic edge $(\Delta + 2)$-coloring for $G$, and we are done.%
We proceed with $4 \in B_1$ and the following proposition: 	

\begin{proposition}
\label{prop3201}
	$G$ contains a $(1, j)_{(x_1, v_1)}$-path that cannot pass through $x_2$ for every $j\in C^*\subseteq B_1$, 
	and $C(x_1) = C(v_1) = C\setminus\{2, 3\}$.
\end{proposition}

If $C^*\subseteq C(v_i)$ for some $i\in \{2, 3\}$, then $(v v_1, v v_i, uv)\to (i, 1, 5)$ gives rise to an acyclic edge $(\Delta + 2)$-coloring for $G$.%
We thus consider below $C^*\setminus C(v_2)\ne\emptyset$ and $C^*\setminus C(v_3)\ne\emptyset$.
If $G$ contains no $(3, \alpha)_{(v_2, v_3)}$-path for some $\alpha\in C^*\setminus C(v_2)$,
then $(v v_2, u v)\to (\alpha, 2)$ gives rise to an acyclic edge $(\Delta + 2)$-coloring for $G$.%
Symmetrically, if $G$ contains no $(2, \beta)_{(v_2, v_3)}$-path for some $\beta\in C^*\setminus C(v_3)$,
then $(v v_3, u v)\to (\beta, 3)$ gives rise to an acyclic edge $(\Delta + 2)$-coloring for $G$.%
Otherwise, $G$ contains a $(3, \alpha)_{(v_2, v_3)}$-path for every $\alpha\in C^*\setminus C(v_2)$ and
a $(2, \beta)_{(v_2, v_3)}$-path for every $\beta\in C^*\setminus C(v_3)$.
It follows that $\{3\}\cup (C^*\setminus C(v_3))\subseteq C(v_2)$ and $\{2\}\cup (C^*\setminus C(v_2)) \subseteq C(v_3)$.
If $1\not\in C(v_2)\cup C(v_3)$, then $(v v_2, v v_1, u v)\to (1, 2, 5)$ gives rise to an acyclic edge $(\Delta + 2)$-coloring for $G$.%
Otherwise, we assume w.l.o.g. that $1\in C(v_2)$.
Therefore, $|C^*\setminus C(v_2)| \ge 2$ and $|C^*\setminus C(v_3)| \ge 1$,
and we let $\alpha_1, \alpha_2\in C^*\setminus C(v_2)$ be two distinct colors and $\beta_1\in C^*\setminus C(v_3)$.
Note that $|C^*\setminus \{\alpha_1, \alpha_2, \beta_1\}|\ge 1$.

Thus, we proceed with the following proposition:

\begin{proposition}
\label{prop3202}
\begin{itemize}
\parskip=0pt
\item[{(1)}]
	$\{1, 2, 3, \beta_1\} \subseteq C(v_2)$, $\{2, 3, \alpha_1, \alpha_2\}\subseteq C(v_3)$, $\max\{d(v_2), d(v_3)\}\ge 5$.
\item[{(2)}]
	$G$ contains a $(3, \alpha_1)_{(v_2, v_3)}$-path and a $(2, \beta_1)_{(v_2, v_3)}$-path, which don't pass through $x_1$.
\end{itemize}
\end{proposition}

If $1\not\in C(x_2)$, then $(u x_1, u x_2, u v)\to (2, 1, \alpha_1)$ gives rise to an acyclic edge $(\Delta + 2)$-coloring for $G$.%
If $G$ contains no $(2, j)_{(x_1, x_2)}$-path for some $j\in \{4\}\cup (C^*\setminus C(x_2))$,
then $(u x_1, u x_2)\to (2, j)$ and $u v\to \{\alpha_1, \beta_1\}\setminus \{j\}$ give rise to an acyclic edge $(\Delta + 2)$-coloring for $G$.%
Symmetrically, if $G$ contains no $(3, j)_{(x_1, x_2)}$-path for some $j\in \{4\}\cup (C^*\setminus C(x_2))$,
then we can obtain an acyclic edge $(\Delta + 2)$-coloring of $G$ too.%
In the other case, $1\in C(x_2)$ and $G$ contains an $(i, j)_{(x_1, x_2)}$-path for every $i\in \{2, 3\}$, $j\in \{4\}\cup (C^*\setminus C(x_2))$.
We thus have $\{1, 2, 3, 4\}\subseteq C(x_2)$ with $c(u x_2) = 4$.

Let $y_4, y_5, \ldots, y_{\Delta + 2}$ be the neighbors of $x_1$ other than $u$, and $c(x_1 y_j)= j$ for every $4\le j\le \Delta + 2$.
Let $z_1, z_2, \ldots, z_{d(x_2)- 1}$ be the neighbors of $x_2$ other than $u$, and $c(x_2 z_i)= i$ for every $i\in \{1, 2, 3\}$;
and let
\[
B_{23}= \{4\}\cup (C^*\setminus C(x_2)) \mbox{ and } B'_{23}= \{j \mid \mbox{there is an } (i, j)_{(x_1, z_i)}\mbox{-path for every } i\in \{2, 3\}\}.
\]
It follows that $B_{23}\subseteq B'_{23}$ and $|B_{23}|= |C^*\setminus C(x_2)| + 1 \ge 2 + 1 = 3$.
We assume w.l.o.g. $\Delta + 1, \Delta + 2\in C^*\setminus C(x_2)$,
and summarize the above into the following proposition:

\begin{proposition}
\label{prop3203}
	$\{4\}\cup (C^*\setminus C(x_2)) = B_{23}\subseteq B'_{23}$,
	$d(z_i)\ge |B_{23}|+ 1\ge 4$ for every $i\in\{2, 3\}$,
	and $d(y_j)\ge |\{1, 2, 3, j\}|\ge 4$ for every $j\in B_{23}$.
\end{proposition}

One sees from $d(x_2) \ge 4$ and $d(v_i) \ge 4$ for every $i\in \{1, 2, 3\}$ that, if $d(y_j) = 3$ for some $j\in C^*$, then $y_j\not\in N(v)\cup \{x_2\}$.
Since $C^*\in B_1$, we have $1\in C(y_j)$ for every $j\in C^*$.
When $\Delta = 5$ and $d(y_5) = 3$, $x_1 y_5\to \{2, 3\}\setminus C(y_5)$;
When $\Delta = 6$ and $d(y_5) = d(y_6) = 3$,
if $6\not\in C(y_5)$, then $x_1 y_5\to \{2, 3\}\setminus C(y_5)$;
or if $6\in C(y_5)$, then $x_1 y_5\to \{2, 3\}\setminus C(y_6)$.
This together with $u v\to 5$ gives rise to an acyclic edge $(\Delta + 2)$-coloring for $G$.%
We proceed with the following proposition:

\begin{proposition}
\label{prop3204}
	If $\Delta = 5$, then $n_3(x_1)= 1$;
	if $\Delta = 6$, then $n_3 (x_1)\le 2$.
\end{proposition}

We summarize the above many cases after Proposition~\ref{prop3202} in the following Lemma~\ref{lemma15},
in each of which an acyclic edge $(\Delta+2)$-coloring for $G$ can be obtained in $O(1)$ time.

\begin{lemma}
\label{lemma15}
If the acyclic edge $(\Delta+2)$-coloring $c$ for the graph $H = G - u v$ satisfies one of the following,
then we can obtain an acyclic edge $(\Delta+2)$-coloring for $G$ in $O(1)$ time:
\begin{itemize}
\parskip=-3pt
\item[{(1)}]
	$1\not\in C(x_2)$;
\item[{(2)}]
	$2\not\in C(x_2)$ or $3\not\in C(x_2)$;
\item[{(3)}]
	$B_{23}\setminus C(x_2)\ne\emptyset$ and $c(x_2 z_k) = i$ with $k\ne i$ for some $i\in \{2, 3\}$;
\item[{(4)}]
	$c(x_2 z_2) = 2$, $c(x_2 z_3) = 3$, and there exists a color $j\in (C^*\setminus C(x_2))\setminus B'_{23}$.
\end{itemize}%
\end{lemma}
\begin{proof}
For (1), let $(u x_2, u x_1)\to (1, 2)$;
for (2), let $u x_1\to \{2, 3\}\setminus C(x_2)$ and $u x_2\to C^*\setminus C(x_2)$;
for (3), let $u x_2\to j\in B_{23}\setminus C(x_2)$, and $u x_1\to i$;
for (4), if $G$ contains no $(\gamma, j)_{(z_\gamma, x_1)}$-path for some $\gamma\in \{2, 3\}$, then let $u x_2\to j$ and $u x_1\to \gamma$.
The above together with $u v\to \{\alpha_1, \beta_1\}\setminus C(u)$ gives rise to an acyclic edge $(\Delta + 2)$-coloring for $G$,
by Proposition~\ref{prop3202} and Lemma~\ref{lemma09}.
\end{proof}

We continue to discuss each of the configurations ($A_{3.1}$)--($A_{3.3}$) separately in the following subsections,
under Propositions~\ref{prop3201}--\ref{prop3204} and specifically with $x_1 = u_1$;
that is, $(u u_1, u u_2)_c = (1, 4)$, $C(u_1) = \{1, 4, 5, \ldots, \Delta + 2\}$ and $C(u_2) = \{1, 2, 3, 4, \ldots, d(u_2)\}$.
We tentatively un-color the two edges $u u_1$ and $u u_2$, and will recolor them at the end.

\subsubsection{Configuration ($A_{3.1}$): $d(u_2) = 5$, $d(w_4) = 3$, and $w_2 w_3\in E(G)$}
From Propositions~\ref{prop3201}--\ref{prop3204}, we have $c(u_2 w_4)\not\in \{2, 3\}$.
Let $s_5, s_6, \ldots, s_{\Delta + 2}$ be the neighbors of $w_2$ other than $u_2$ and $w_3$,
and $t_5, t_6, \ldots, t_{\Delta + 2}$ be the neighbors of $w_3$ other than $u_2$ and $w_2$.
Note that $N(u_2)\cap N(u)\ne\emptyset$, $w_4\not\in \{v_1, v_2, v_3\}$,
and for each $i\in \{1, 2, 3\}$ any re-coloring of $u_2 w_i$ cannot change $C(v) = \{1, 2, 3\}$.
By symmetry, it suffices to consider the following three cases {\bf ($A_{3.1}$-1)}, {\bf ($A_{3.1}$-2)} and {\bf ($A_{3.1}$-3)}.

\paragraph{Case ($A_{3.1}$-1).} $(u_2 w_2, u_2 w_3)_c = (2, 3)$.

Since $C^*\setminus \{5\}\subseteq B_{23}$,
it suffices to assume that $c(w_2 s_j)= c(w_3 t_j)= j$ for every $j\in \{6, 7, \ldots, \Delta + 2\}$ and,
either (i) $c(w_2 s_5) = c(w_3 t_5) = 4$, $c(w_2 w_3)\in \{1, 5\}$
or (ii) $c(w_2 w_3) = 4$, $c(w_2 s_5) = 3$, $c(w_3 t_5) = 2$.
One sees that there is no $(i, j)$-path passing through $w_i$ for every $i\in \{2, 3\}$ and $j\in \{1, 5\}$.
We assume w.l.o.g. that $(u_2 w_1, u_2 w_4) = (5, 1)$.

If $G$ contains no $(5, i)_{(w_1, w_4)}$-path for some $i\in B_{23}\setminus C(w_4)$, then let $u_2 w_4\to i$;
or if $1\not\in C(w_1)$, then let $u_2 w_1\to 1$ and $u_2 w_4\to B_{23}\setminus C(w_4)$.
It follows from Lemma~\ref{lemma15}(1) and (4), respectively, that we are done.

If $1\in C(w_1)$ and $G$ contains a $(5, i)_{(w_1, w_4)}$-path for every $i\in B_{23}\setminus C(w_4)$,
then $B_{23}\subseteq C(w_1)\cup C(w_4)$ and $5\in C(w_4)$.
It follows that $\{2, 3\}\setminus (C(w_1)\cup C(w_4))\ne\emptyset$ and assume w.l.o.g. that $2\not\in C(w_1)\cup C(w_4)$.
For (i), if $c(w_2 w_3) = 5$, then let $(u_2 w_2, u_2 w_4)\to (1, 2)$;
or if $c(w_2 w_3) = 1$, then let $(u_2 w_2, u_2 w_1) \to (5, 2)$.
For (ii), if $G$ contains no $(2, 3)_{(w_3, w_4)}$-path, then let $(u_2 w_2, u_2 w_4)\to (1, 2)$;
or otherwise, let $(u_2 w_2, u_2 w_1)\to (5, 2)$.
It follows from Lemma~\ref{lemma15}(3) that we are done.

\paragraph{Case ($A_{3.1}$-2).} $(u_2 w_2, u_2 w_1, u_2 w_3, u_2 w_4)_c = (2, 3, 5, 1)$.

When $G$ contains no $(5, i)_{(w_3, w_4)}$-path for some $i\in B_{23}\setminus C(w_4)$, let $u_2 w_4\to i$ and we are done by Lemma \ref{lemma15}(1).
In the other case, $G$ contains a $(5, i)_{(w_3, w_4)}$-path for every $i\in B_{23}\setminus C(w_4)$,
which implies that $5\in C(w_4)$ and $(B_{23}\setminus C(w_4))\subseteq C(w_3)$.
We assume w.l.o.g. that $c(w_2 s_i) = i\in \{6, \ldots, \Delta + 2\}$,
and consider the following two subcases on the different values of $c(w_2 w_3)$.

{Case 2.1.}\ $c(w_2 s_5) = 4$ and $c(w_2 w_3)\in \{1, 3\}$.
One sees that $G$ contains no $(2, 5)_{(w_2, u_1)}$-path in this case.
If $B_{23}\subseteq C(w_3)$, then switch the colors of $\{u_2 w_2, u_2 w_3\}$ and we are done by Lemma \ref{lemma15}(3).
Otherwise, $B_{23}\setminus C(w_3)\ne\emptyset$ and we assume w.l.o.g. that $4\not\in C(w_3)$.
It follows that $C(w_4) = \{1, 4, 5\}$ and $(B_{23}\setminus \{4\})\subseteq C(w_3)$.
If $1\not\in (w_3)$, then let $u_2 w_3\to 4$;
if $1\in C(w_3)\setminus \{c(w_2 w_3)\}$, then switch the colors of $\{u_2 w_2, u_2 w_3\}$;
we are done by Lemma \ref{lemma15}(4) and (3), respectively.
Otherwise, it suffices to assume that $c(w_2 w_3) = 1$.

If $2\not\in C(w_3)$, then let $(u_2 w_2, u_2 w_3, u_2 w_4)\to (5, 4, 2)$;
otherwise $C(w_3) = C\setminus \{3, 4\}$.
If $2\not\in C(w_1)$, then let $(u_2 w_1, u_2 w_2, u_2 w_3)\to (2, 5, 3)$;
otherwise, we have $C(w_1) = C\setminus \{1, 5\}$ and let $(u_2 w_1, u_2 w_3)\to (5, 3)$.
It follows from Lemma~\ref{lemma15}(1), (3) and (3), respectively, that $G$ admits an acyclic edge $(\Delta + 2)$-coloring.

{Case 2.2.}\ $(w_2 w_3, w_3 t_5) = (4, 2)$.
In this case, we assume w.l.o.g. that $7\not\in C(w_4)$ and $G$ contains a $(5, 7)_{(w_3, w_4)}$-path.
We consider the following two subcases on the different values of $B_{23}\setminus C(w_3)$.

{Case 2.2.1.}\ $B_{23}\subseteq C(w_3)$.
Then $C(w_3) = C\setminus \{1, 3\}$.
If $3, 5\not\in C(w_2)$, then let $(u_2 w_2, u_2 w_3, u_2 w_4)$ $\to (5, 1, 7)$ and we are done by Lemma \ref{lemma15}(2);
otherwise, $\{3, 5\}\cap C(w_2)\ne \emptyset$ and thus $1\not\in C(w_2)$.

If $5\not\in B'_{23}$, then let $(u_2 w_3, u_2 w_4)\to (1, 7)$;
it follows from Lemma~\ref{lemma15}(4) that $G$ admits an acyclic edge $(\Delta + 2)$-coloring.
In the other case, we have $5\in B'_{23}$, $C(w_1) = C\setminus \{1, 2\}$ and $C(w_2)= C\setminus \{1, 3\}$.
If $3\not\in C(w_4)$, then switch the colors of $\{u_2 w_1, u_2 w_4\}$;
otherwise, $3\in C(w_4)$ and $C(w_4) = \{1, 3, 5\}$.
If there is no $(2, 5)_{(w_1, w_3)}$-path, then switch the colors of $\{u_2 w_2, u_2 w_1\}$;
or otherwise, switch the colors of $\{u_2 w_2, u_2 w_4\}$.
It follows from Lemma~\ref{lemma15}(3) that $G$ admits an acyclic edge $(\Delta + 2)$-coloring.

{Case 2.2.2.}\ $B_{23}\setminus C(w_3)\ne\emptyset$.
We assume w.l.o.g. that $6\not\in C(w_3)$.
It follows that $C(w_4) = \{1, 5, 6\}$ and $C(w_2) = C\setminus \{1, 3\}$ with $c(w_2 s_5) = 5$.
If $3\not\in C(w_3)$, then we first let $(u_2 w_3, u_2 w_4)\to (3, 7)$.
Next, if $1\not\in C(w_1)$, then let $u_2 w_1\to 1$;
or otherwise, we have $C(w_1) = C\setminus \{2, 5\}$ and let $u_2 w_1\to 5$.
We are done by Lemma \ref{lemma15}(3) and (1), respectively.

In the other case, $3\in C(w_3)$ and thus $C(w_3) = C\setminus \{1, 6\}$.
If $5\not\in B'_{23}$, then let $u_2 w_3\to 6$ and we are done by Lemma \ref{lemma15}(4);
otherwise, $5\in C(w_1)$ and thus $C(w_1)= C\setminus \{1, 2\}$.
If $G$ contains no $(2, 5)_{(w_3, w_4)}$-path, then switch the colors of $\{u_2 w_2, u_2 w_4\}$;
if $G$ contains no $(3, 5)_{(w_3, w_4)}$-path, then switch the colors of $\{u_2 w_1, u_2 w_4\}$;
or otherwise, $G$ contains a $(2, 5)_{(w_3, w_4)}$-path and a $(3, 5)_{(w_3, w_4)}$-path,
and switch the colors of $\{u_2 w_1, u_2 w_2\}$.
It follows from Lemma~\ref{lemma09} and Lemma~\ref{lemma15}(3) that $G$ admits an acyclic edge $(\Delta + 2)$-coloring, and we are done.

\paragraph{Case ($A_{3.1}$-3).} $(u_2 w_2, u_2 w_1, u_2 w_3, u_2 w_4)_c = (2, 3, 1, 5)$.

When $G$ contains no $(5, i)_{(w_3, w_4)}$-path for some $i\in B_{23}\setminus C(w_3)$,
let $u_2 w_3 \to i$ and we are done by Lemma \ref{lemma15}(1).
In the other case, we assume that if $B_{23}\setminus C(w_3)\ne \emptyset$, then $G$ contains a $(5, i)_{(w_3, w_4)}$-path for every $i\in B_{23}\setminus C(w_3)$,
and thus $5\in C(w_3)$ and $(B_{23}\setminus C(w_3))\subseteq C(w_4)$.

If $5\not\in B'_{23}$ and $G$ contains no $(1, i)_{(w_3, w_4)}$-path for some $i\in B_{23}\setminus C(w_4)$,
then let $u_2 w_4\to i$ and we are done by Lemma \ref{lemma15}(4).
Otherwise, either $5\in B'_{23}$,
or $5\not\in B'_{23}$ and then $G$ contains a $(1, i)_{(w_3, w_4)}$-path for every $i\in B_{23}\setminus C(w_4)$, $1\in C(w_4)$,
and $(B_{23}\setminus C(w_3))\subseteq C(w_4)$.

We assume w.l.o.g. that $c(w_2 s_i)= i\in \{6, \ldots, \Delta + 2\}$, and consider the following two subcases on the different values of $c(w_2 w_3)$.

{Case 3.1.}\ $c(w_2 s_5) = 4$ and $c(w_2 w_3)\in \{3, 5\}$.
If $B_{23}\subseteq C(w_3)$, then switch the colors of $\{u_2 w_2, u_2 w_3\}$ and we are done by Lemma \ref{lemma15}(3).
Otherwise, $B_{23}\setminus C(w_3)\ne\emptyset$ and we assume w.l.o.g. that $6\not\in C(w_3)$.
It follows that $5\in C(w_3)$, $6\in C(w_4)$ and $G$ contains a $(5, 6)_{(w_3, w_4)}$-path.

First, assume that $5\in C(w_3)\setminus \{c(w_2 w_3)\}$.
One sees from $5\not\in C(w_2)$ that $5\not\in B'_{23}$.
It follows that $1\in C(w_4)$, $C(w_4)= \{1, 5, 6\}$, and $C(w_3)= C\setminus\{2, 6\}$.
Then switch the colors of $\{u_2 w_2, u_2 w_3\}$ and we are done by Lemma \ref{lemma15}(3).

Next, assume that $c(w_2 w_3)= 5$.
If $2, 3\not\in C(w_3)$, then switch the colors of $\{u_2 w_2, u_2 w_3\}$ and we are done by Lemma \ref{lemma15}(3);
otherwise, $\{2, 3\}\cap C(w_3)\ne \emptyset$.
When $2\not\in C(w_3)$,
since $5\not\in B'_{23}$, we have $C(w_4) = \{5, 1, 6\}$ and $C(w_3) = C\setminus \{2, 6\}$,
let $(u_2 w_2, u_2 w_3, u_2 w_4)\to (1, 6, 2)$ and we are done by Lemma \ref{lemma15}(3).
Now it suffices to assume that $2\in C(w_3)$.

When $1\in C(w_4)$ and thus $C(w_3) = C\setminus \{3, 6\}$,
if $1\not\in C(w_1)$, then let $(u_2 w_1$, $u_2 w_3$, $u_2 w_4)\to (1, 6, 3)$;
if $C(w_1) = C\setminus \{2, 5\}$, then let $(u_2 w_1, u_2 w_3, u_2 w_4)\to (5, 3, 7)$;
we are done by Lemma \ref{lemma15}(3) and (1), respectively.

In the other case, $1\not\in C(w_4)$.
It follows that $5\in B'_{23}$ and $C(w_1) = C\setminus \{1, 2\}$.
If $3\not\in C(w_3)$, then let $(u_2 w_1, u_2 w_3)\to (1, 3)$;
otherwise, $3\in C(w_3)$.
Since $B_{23}\subseteq C(w_3)\cup C(w_4)$,
we assume w.l.o.g. that $C(w_4)= \{4, 5, 6\}$ and $C(w_3) = C\setminus \{4, 6\}$.
Then switch the colors of $\{u_2 w_1, u_2 w_2\}$.
It follows from Lemma~\ref{lemma15}(3) that $G$ admits an acyclic edge $(\Delta + 2)$-coloring.

{Case 3.2.}\ $(w_2 w_3, w_3 t_5) = (4, 2)$.
According to $5\in B'_{23}$ or not, there are two subcases.

{Case 3.2.1.}\ $5\not\in B'_{23}$.
In this case, $1\in C(w_4)$ and $B_{23}\subseteq C(w_3)\cup C(w_4)$.

First, assume that $1\not\in C(w_2)$.
It follows that $C(w_4) = \{5, 1, 4\}$, $C(w_3)= C\setminus \{3, 5\}$, and $G$ contains a $(1, 6)_{(w_3, w_4)}$-path.
Let $(u_2 w_3, u_2 w_4)\to (5, 6)$.
If $5\not\in C(w_2)$ and thus $C(w_2)= C\setminus \{1, 3\}$, then let $u_2 w_2\to 1$.
We are done by Lemma \ref{lemma15}(1) and (2), respectively.

Next, assume that $1\in C(w_2)$ and $C(w_2) = C\setminus \{3, 5\}$.
If $B_{23}\subseteq C(w_3)$ and thus $C(w_3) = C\setminus \{3, 5\}$, then let $u_2 w_3\to 5$ and $u_2 w_4\to B_{23}\setminus C(w_4)$;
we are done by Lemma \ref{lemma15}(1).
In the other case, $B_{23}\setminus C(w_3)\ne \emptyset$ and we assume w.l.o.g. that $6\in C(w_4)\setminus C(w_3)$.
It follows that $C(w_4) = \{5, 1, 6\}$, $C(w_3) = C\setminus \{3, 6\}$ and $G$ contains a $(5, 6)_{(w_3, w_4)}$-path.
Let $(u_2 w_2, u_2 w_3, u_2 w_4)\to (5, 6, 4)$,
we are done by Lemma~\ref{lemma15}(2).

{Case 3.2.2.}\ $5\in B'_{23}$ and $C(w_2) = C\setminus \{1, 3\}$, $C(w_1) = C\setminus \{1, 2\}$.
If $5\not\in C(w_3)$, then let $u_2 w_3\to 5$ and $u_2 w_4\to B_{23}\setminus C(w_4)$;
we are done by Lemma \ref{lemma15}(1).

In the other case, $5\in C(w_3)$ and we assume w.l.o.g. that $6\in C(w_4)\setminus C(w_3)$.
If $G$ contains neither a $(5, 3)_{(w_2, w_4)}$-path nor a $(5, 2)_{(w_1, w_4)}$-path,
then switch the colors of $\{u_2 w_1, u_2 w_2\}$;
otherwise, $G$ contains a $(5, \gamma)_{(w_\gamma, w_4)}$-path, where $\{2, 3\}\cap C(w_4) = \{\gamma\}$.
It follows that $1\not\in C(w_4)$ and $C(w_3) = C\setminus \{3, 6\}$, and then let $(u_2 w_3, u_2 w_1)\to (3, 1)$.
One sees from Lemma~\ref{lemma09} that there is no $(5, 3)_{(w_3, w_4)}$-path even if $3\in C(w_4)$.
We are done by Lemma~\ref{lemma15}(3).

\subsubsection{Configuration ($A_{3.2}$): $d(u_2)= 6$, $d(w_3)= d(w_4)= d(w_5)= 3$}
From Propositions~\ref{prop3201}--\ref{prop3204},
we assume that $c(u_2 w_i) = i + 1$, $i\in \{1, 2, 4, 5\}$, $c(u_2 w_3) = 1$, and $C(w_3) = \{1, a, b\}$.
Note that $w_3, w_4, w_5\not\in \{v_1, v_2, v_3, u_1, v\}$ and $|B_{23}|\ge 3$.

If $G$ contains neither an $(a, i)_{(w_3, u_2)}$-path nor a $(b, i)_{(w_3, u_2)}$-path for some $i\in B_{23}\setminus C(w_3)$,
then let $u_2 w_3\to i$ and we are done by Lemma \ref{lemma15}(1).
Otherwise, $G$ contains an $(a, i)_{(w_3, u_2)}$-path or a $(b, i)_{(w_3, u_2)}$-path for every $i\in B_{23}\setminus C(w_3)$.
By symmetry, we consider the following two possibilities:
\begin{itemize}
\parskip=0pt
\item[(i)]
	$\{a, b\}\cap \{5, 6\} = \{5\}$, $C(w_3) = \{1, 4, 5\}$, $C(w_4) = \{5, 7, 8\}$, and $G$ contains a $(5, i)_{(w_3, w_4)}$-path for every $i\in \{7, 8\}$.
    Then let $u_2 w_3\to 6$ and $u_2 w_5\to B_{23}\setminus C(w_5)$.
\item[(ii)]
	$C(w_3) = \{1, 5, 6\}$, $C(w_4) = \{5, 4, 7\}$, $8\in C(w_5)$, and $G$ contains $(5, i)_{(w_3, w_4)}$-path for every $i\in \{4, 7\}$.
     Then let $u_2 w_3\to 8$ and $u_2 w_5\to \{4, 7\}\setminus C(w_5)$.
\end{itemize}
By Lemma~\ref{lemma15}(1) we are done.

\subsubsection{Configuration ($A_{3.3}$): $d(u_2) = 6$, $d(w_4) = d(w_5) = 3$, $d(w_2)\le 5$, $d(w_3) = 4$, and $w_2 w_3\in E(G)$}
From Propositions~\ref{prop3201}--\ref{prop3204},
we assume that $(u_2 w_1, u_2 w_4, u_2 w_5)_c = (3, \xi, 5)$, $\{c(u_2 w_2), c(u_2 w_3)\} = \{2, \eta\}$, where $\{\xi, \eta\} = \{1, 6\}$.
Note that $N(u_2)\cap N(u)\ne\emptyset$, $w_4, w_5\not\in \{v_1, v_2, v_3\}$, and any re-coloring $u_2 w_3$ cannot change $C(v)$.
Denote $C(w_5) = \{5, a, b\}$ and $C(w_4) = \{\xi, j, k\}$.

\paragraph{Case ($A_{3.3}$-1).}\ $5, 6\not\in B'_{23}$.

If $G$ contains neither an $(a, i)_{(w_5, u_2)}$-path nor a $(b, i)_{(w_5, u_2)}$-path for some $i\in B_{23}\setminus C(w_5)$,
then let $u_2 w_5\to i$ and we are done by Lemma \ref{lemma15}(4).
Otherwise, $G$ contains an $(a, i)_{(w_5, u_2)}$-path or a $(b, i)_{(w_5, u_2)}$-path for every $i\in B_{23}\setminus C(w_5)$.
Symmetrically, $G$ contains an $(j, i)_{(w_4, u_2)}$-path or a $(k, i)_{(w_4, u_2)}$-path for every $i\in B_{23}\setminus C(w_4)$.
It follows that $C(w_5)\cap \{\xi, \eta\}\ne \emptyset$ and $C(w_4)\cap \{5, \eta\}\ne \emptyset$.
One sees that there exists $\gamma\in B_{23}\setminus (C(w_4)\cup C(w_5))$.
Next, if $G$ contains no $(\eta, \gamma)_{(u_2, w_5)}$-path, then let $u_2 w_5\to \gamma$;
otherwise, let $u_2w_4\to\gamma$;
either way we are done by Lemma \ref{lemma15}(4).

\paragraph{Case ($A_{3.3}$-2).}\ $\{5, 6\}\cap B'_{23}\ne\emptyset$.

In this case, $c(u_2 w_2) = 2$, $2\in C(w_3)$, and $|\{5, 6\}\cap B'_{23}| = 1$.
Let $(u_2 y_1, u_2 y_5, u_2 y_6)_c = (1, 5, 6)$, where $\{y_1, y_5, y_6\} = \{w_3, w_4, w_5\}$.
For $k\in \{1, 5, 6\}$, if $y_k\in \{w_4, w_5\}$, then let $C(y_k) = \{k, a_k, b_k\}$;
or if $y_k = w_3$, then let $C(y_k) = \{k, 2, a_k, b_k\}$.
If $G$ contains neither an $(a_1, i)_{(y_1, u_2)}$-path nor a $(b_1, i)_{(y_1, u_2)}$-path for some $i\in B_{23}\setminus C(y_1)$,
then let $u_2 y_1\to B_{23}\setminus C(y_1)$ and we are done by Lemma \ref{lemma15}(1);
otherwise, $G$ contains an $(a_1, i)_{(y_1, u_2)}$-path or a $(b_1, i)_{(y_1, u_2)}$-path for every $i\in B_{23}\setminus C(y_1)$.

We assume w.l.o.g. that $\{a_5, b_5\} = \{7, 8\}$, $G$ contains a $(5, i)_{(y_1, y_5)}$-path for every $i\in \{7, 8\}$ and
either (i)\ $\{a_1, b_1\} = \{4, 5\}$
or (ii)\ $\{a_1, b_1\} = \{5, 6\}$ with $4\in C(y_6)$.
If $5\not\in B'_{23}$, then let $u_2 y_5\to 4$ and we are done by Lemma~\ref{lemma15}(4).
Otherwise, $6\not\in B'_{23}$ and first let $u_2 y_6\to B_{23}\setminus C(y_6)$, and then let $u_2 y_1\to 6$ for (i), or $u_2 y_1\to 4$ for (ii);
by Lemma~\ref{lemma15}(1) we are done.

This finishes the inductive step for the case where $G$ contains the configuration ($A_3$).

\subsection{Configuration ($A_5$)}
In this subsection we prove the inductive step for the case where $G$ contains one of the configurations described in ($A_5$):

\begin{description}
\parskip=0pt
\item[$(A_5)$]
	A $5$-vertex $u$ adjacent to $u_1, u_2, u_3, u_4$ and a $3$-vertex $v$, sorted by their degrees.
    At least one of the configurations ($A_{5.1}$)--($A_{5.4}$) occurs.
\end{description}

The same as in the above, we first not to distinguish the four neighbors $u_1, u_2, u_3, u_4$ but refer to them as $x_1, x_2, x_3, x_4$;
we assume w.l.o.g. that $c(u x_i) = i$ for $1\le i\le 4$.
Let $v_1, v_2$ denote the other two neighbors of $v$ and $C^*_5 = C\setminus\{1, 2, 3, 4\} = \{5, 6, \ldots, \Delta+2\}$,
and similarly we refer to these $\Delta-2$ colors indistinguishably as $\gamma_5, \gamma_6, \ldots, \gamma_{\Delta + 2}$.
Assume from Proposition~\ref{prop3001} that $1\le|C(u)\cap C(v)|\le 2$.

\begin{lemma}
\label{lemma16}
If $|C(u)\cap C(v)|=2$, then in $O(1)$ time either an acyclic edge $(\Delta + 2)$-coloring for $G$ can be obtained,
or an acyclic edge $(\Delta + 2)$-coloring for $H$ can be obtained such that $|C(u)\cap C(v)| = 1$.
\end{lemma}
\begin{proof}
Assume $(v v_1, v v_2)_c = (1, 2)$.

When there exists $i\in C^*_5\setminus C(v_1)$, $v v_1\to i$ reduces to the case $|C(u)\cap C(v)| = 1$.
In the other case, $C^*_5\setminus C(v_1) = \emptyset$,
that is, $C^*_5\cup \{1\}\subseteq C(v_1)$ and likewise, $C^*_5\cup \{2\}\subseteq C(v_2)$.

Assume that $1, 2\not\in (C(v_1)\cup C(v_2))\setminus \{c(vv_1), c(vv_2)\}$.
Since $j\in B_1\cup B_2$ for every $j\in C^*_5$, we assume w.l.o.g. that $j\in B_1\cap C(x_1)$.
If $G$ contains no $(2, j)_{(x_2, v_1)}$-path, then $(v v_1, v v_2, u v)\to (2, 1, j)$ gives rise to an acyclic edge $(\Delta + 2)$-coloring for $G$, and we are done.
Otherwise, $G$ contains a $(2, j)_{(x_2, v_1)}$-path and $j\in C(x_2)$.
Thus, $C^*_5\subseteq C(x_1)\cup C(x_2)$ and $G$ contains an $(i, j)_{(u, v)}$-path for every $i\in \{1, 2\}$, $j\in C^*_5$.

Assume that $C^*_5\setminus C(x_2)\ne\emptyset$ and assume w.l.o.g. that $5\not\in C(x_2)$.
Then $5\in B_1$ and $\{1, 2\}\cap (C(v_1)\cup C(v_2))\setminus \{c(v v_1), c(v v_2)\}\ne\emptyset$.
If $G$ contains no $(i, 5)_{(x_2, x_i)}$-path for every $i\in \{3, 4\}$, then $u x_2\to 5$ reduces to the case $|C(u)\cap C(v)| = 1$.
Otherwise, $G$ contains a $(i, 5)_{(x_2, x_i)}$-path for some $i\in \{3, 4\}$ and we assume w.l.o.g. that $G$ contains a $(3, 5)_{(x_2, x_3)}$-path.
The re-coloring $u v\to 5$ and,
if $1\in C(v_2)$ and $2\not\in C(v_1)$, then $(v v_1, v v_2)\to (2, 3)$,
or if $2\in C(v_1)$ and $1\not\in C(v_2)$, then $(v v_1, v v_2)\to (3, 1)$,
or otherwise $v v_1\to 3$,
give rise to an acyclic edge $(\Delta + 2)$-coloring for $G$, and we are done.

We proceed with the following proposition:

\begin{proposition}
\label{prop3301}
\begin{itemize}
\parskip=0pt
\item[{(1)}]
	$C^*_5\cup \{1\}\subseteq C(v_1)\cap C(x_1)$ and $C^*_5\cup \{2\}\subseteq C(v_2)\cap C(x_2)$.
\item[{(2)}]
	$d(x_1)\ge \Delta- 1\ge4$, $d(x_2)\ge \Delta- 1\ge 4$, and $x_1, x_2\not\in S_2\cup S_3$.
\item[{(3)}]
	If $1, 2\not\in (C(v_1)\cup C(v_2))\setminus \{c(v v_1), c(v v_2)\}$, then $G$ contains an $(i, j)_{(x_i, v)}$-path for every pair $i\in \{1, 2\}$, $j\in C^*_5$.
\end{itemize}
\end{proposition}

Note that if re-coloring some edges $E'\subseteq E(H)$ doesn't produce any new bichromatic cycles,
and $|C(u)\cap C(v)| = 2$ with $(C(u)\cap C(v))\cap S_3\ne \emptyset$,
then we are done.

Assume that $x_3$ is a $3$-vertex.
If $3\not\in C(v_2)$ and $G$ contains no $(1, 3)_{(v_1, v_2)}$-path, then let $v v_2\to 3$ and we are done.
If $3\not\in C(v_1)$ and $G$ contains no $(2, 3)_{(v_1, v_2)}$-path, then let $v v_1\to 3$ and we are done.
Otherwise, $3\in C(v_1)\cup C(v_2)$.
By symmetry, we assume w.l.o.g. that $C(v_1) = C^*_5\cup \{1, 3\}$,
and either (i)\ $C(v_2) = C^*_5\cup \{2, 1\}$,
or (ii)\ $C(v_2) = C^*_5\cup \{2, 3\}$.

For (i), let $(v v_1, v v_2)\to \{2, 3\}$ and we are done.
For (ii), from Proposition~\ref{prop3301}, $G$ contains an $(i, j)_{(x_i, v)}$-path for every pair $i\in \{1, 2\}$, $j\in C^*_5$.
One sees that the same argument applies if $c(v v_2) = 4$, and thus $G$ contains a $(4, j)_{(x_4, v)}$-path for every $j\in C^*_5$, and $C^*_5\subseteq C(x_4)$.
We assume w.l.o.g. that $5\not\in C(x_3)$ and let $u x_3\to 5$.
By similar discussion, we conclude that $G$ contains a $(3, i)_{(x_i, v)}$-path and $3\in C(x_i)$ for every $i\in \{1, 2, 4\}$.
Thus, $(C^*_5\cup \{3\})\subseteq C(x_i)$ and $d(x_i) = \Delta$ for every $i\in \{1, 2, 4\}$,
$C(v_1) = C^*_5\cup \{1, 3\}$, $C(v_2) = C^*_5\cup \{2, 3\}$, and $v_1, v_2\not\in N(u)$.

In such a case, none of ($A_{5.1}$)--($A_{5.3}$) holds, and we assume that ($A_{5.4}$) holds.
From Proposition~\ref{prop3301}, $C(u)\cap C(v)\ne \{c(u u_3), c(u u_4)\}$.
By symmetry, we only need to discuss the following two cases.

{Case 1.}\ $(u u_1, u u_3, u u_2, u u_4)_c = (1, 2, 3, 4)$.

From Proposition~\ref{prop3301}, $\Delta = 5$, $C(u_3) = C^*_5\cup \{2\}$ and we assume w.l.o.g. that $c(u_3 u_4) = 5$.
If $4\not\in C(v_1)$ and $G$ contains no $(2, 4)_{(v_1, v_2)}$-path, then $v v_1\to 4$;
if $4\in C(v_1)\setminus C(v_2)$, then $(v v_1, v v_2)\to (2, 4)$;
this reduces to the case where $C(u)\cap C(v) = \{c(u u_3), c(u u_4)\}$.

Thus, we have either (i)\ $4\in C(v_1)\cap C(v_2)$,
or (ii)\ $4\not\in C(v_1)$, which implies that $G$ contains a $(2, 4)_{(v_1, v_2)}$-path, $C(v_1) = C^*_5\cup \{1, 2\}$ and $C(v_2) = C^*_5\cup \{2, 4\}$.

For (i), since $1, 2\not\in (C(v_1)\cup C(v_2))\setminus \{c(v v_1), v(v v_2)\}$, we have $2\in C(u_4)$.
One sees that the same argument applies if $c(v v_1) = 3$, and thus $G$ contains a $(3, j)_{(u_2, v)}$-path for every $j\in C^*_5$.
Let $u u_4\to \rho\in C^*_5\setminus C(u_4)$ and $u v\to 4$.
Note that there is no $(i, \rho)_{(u_4, u)}$-cycle for every $i\in \{1, 2, 3\}$.
If $4\not\in B_1$, we are done;
or otherwise, switch the colors of $\{v v_1, v v_2\}$ and we are also done.

For (ii),
one sees that the same argument applies if $(v v_1, v v_2)_c =(4, 1)$, and thus we may further assume that $C(u_4) = C^*_5\cup \{4\}$ and $5\in B_1$.
Then $(v v_1, v v_2, u v)\to (4, 1, 5)$ gives rise to an acyclic edge $(\Delta + 2)$-coloring for $G$, and we are done.

{Case 2.}\ $(u u_1, u u_2)_c = (1, 2)$.

If there exists a color $k\in \{3, 4\}\setminus (C(v_1)\cup C(v_2))$, then $vv_1\to k$ reduces the proof to the above Case 1.
Otherwise, $3, 4\in C(v_1)\cup C(v_2)$,
and we assume w.l.o.g. that $C(v_1) = C^*_5\cup \{1, 3\}$ and $C(v_2) = C^*_5\cup \{2, 4\}$.
Then $v v_2\to 3$ reduces the proof to the above Case 1 too.
\end{proof}

By Lemma \ref{lemma16} we consider below the case where $|C(u)\cap C(v)| = 1$,
and $G$ contains the configuration ($A_5$) with $(v v_1, v v_2)_c = (a, 5)$ and $a\in C(u)$.
From Proposition~\ref{prop3001}, $C\setminus \{1, 2, 3, 4, 5\} \subseteq B_a\subseteq C(u_1)\cap C(v_1)$.

\begin{lemma}
\label{lemma17}
If $5\not\in B_a$, then an acyclic edge $(\Delta + 2)$-coloring for $G$ can be obtained in $O(1)$ time.
\end{lemma}
\begin{proof}
We assume w.l.o.g. that $a = 1$.

When there exists a color $i\in C^*_5\setminus C(v_2)$, or $G$ contains no $(1, i)_{(v v_1, v v_2)}$-path for some $i\in \{2, 3, 4\}\setminus C(v_2)$,
$(v v_2, u v)\to (i, 5)$ gives rise to an acyclic edge $(\Delta + 2)$-coloring for $G$, and we are done.
In the other case, $(1\cup C^*_5)\subseteq C(v_2)$, $2, 3, 4\in C(v_1)\cup C(v_2)$ and
$G$ contains a $(1, i)_{(v v_1, v v_2)}$-path for every $i\in \{2, 3, 4\}\setminus C(v_2)$.
Hence, we assume w.l.o.g. that $C(v_1)= (C^*_5)\setminus \{5\}\cup \{1, 2, 3\}$,
$C(v_2) = C^*_5\cup \{1, 4\}$, and $G$ contains a $(1, i)_{(v v_1, v v_2)}$-path for every $i\in \{2, 3\}$.

One sees that the same argument applies if $c(v v_1)= 4$, or if $(v v_1, v v_2)_c= (5, i)$ for any $i\in \{2, 3\}$,
we thus may further assume that $G$ contains an $(i, j)_{(x_i, v)}$-path for every pair $i\in \{1, 2, 3, 4\}$, $j\in (C^*_5\setminus \{5\})\subseteq C(x_i)$,
and $G$ contains an $(i, j)_{(v_1, v_2)}$-path for every pair $i\in \{1, 4\}$, $j\in \{2, 3\}$.

If $5\not\in C(x_1)$ and $G$ contains no $(5, i)_{(x_1, x_i)}$-path for every $i\in \{2, 3, 4\}$,
then $(u x_1, v v_1, u v)\to (5, 4, 1)$ gives rise to an acyclic edge $(\Delta + 2)$-coloring for $G$.
If $5\not\in C(x_2)$ and $G$ contains no $(5, i)_{(x_2, x_i)}$-path for every $i\in \{1, 3, 4\}$,
then $(u x_2, u v)\to (5, 2)$ gives rise to an acyclic edge $(\Delta + 2)$-coloring for $G$.
Otherwise, for every $i\in \{1, 2, 3, 4\}$,
we have $5\in C(x_i)$, or $G$ contains a $(5, j)_{(x_i, x_j)}$-path for some $j\in (C(x_i) \cap \{1, 2, 3, 4\})\setminus \{i\}$.
It follows that $d(x_i)\ge 4$ for every $i\in \{1, 2, 3, 4\}$.

In such a case, none of ($A_{5.1}$)--($A_{5.3}$) holds, and we assume that ($A_{5.4}$) holds.
Furthermore, we assume w.l.o.g. that $(u u_3, u u_4)_c= (3, \gamma)$, where $\gamma\in \{1, 2, 4\}$.
One sees that $6, 7\in C(u_3)$, and if $c(u_3 u_4)\in \{6, 7\}$, then $C(u_3) = C(u_4) = \{6, 7, 3, \gamma\}$.
It follows that $c(u_3 u_4) = 5$, $C(u_3)\setminus \{3\} = C(u_4)\setminus \{\gamma\} = \{5, 6, 7\}$ and
$G$ contains a $(3, i)_{(u_3, v)}$-path for every $i\in \{6, 7\}$ which cannot pass through $u_4$.
Note that $N(u)\cap N(v)= \emptyset$ and $G$ contains a $(1, 3)_{(v_1, v_2)}$-path.
Then $(u_3 u_4, u u_3, uv)\to (3, 5, 3)$ gives rise to an acyclic edge $(\Delta + 2)$-coloring for $G$, and we are done.
\end{proof}

By Lemma \ref{lemma17}, we assume below $a = 1$, $c(u x_1) = 1$, $5\in B_1$ and $C^*_5\subseteq C(x_1)\cap C(v_1)$.
Recall that by Lemma~\ref{lemma09},
if there is a $(1, j)_{(v_1, x_1)}$-path for some $j\in C\setminus C(u)$ and $1\not\in C(x_i)$ where $i\ne 1$,
then there is no $(1, j)_{(v_1, x_2)}$-path.
We immediately have the following two useful Lemmas.

\begin{lemma}
\label{lemma18}
If re-coloring some edges $E'\subseteq E(H)$ incident to $u$ or $v$ does not produce any new bichromatic cycles,
but at least one of the following holds:
\begin{itemize}
\parskip=0pt
\item[{(1)}]
	$C(u)\cap C(v) = \emptyset$,
\item[{(2)}]
	$C(u)\cap C(v) = \{1\}$ with $c(u x_i) = 1$, where $i\ne 1$,
\end{itemize}
then an acyclic edge $(\Delta + 2)$-coloring for $G$ can be obtained in $O(1)$ time.
\end{lemma}

\begin{lemma}
\label{lemma19}
If re-coloring some edges $E'\subseteq E(H)$ does not produce any new bichromatic cycles but $C(u)\cap C(v) = \{i\}$,
then $G$ contains an $(i, j)_{(x_{i_0}, v_{j_0})}$-path for every $j\in C\setminus C(u)$, where $c(u x_{i_0})= c(v v_{j_0}) = i$.
\end{lemma}

Next assume that there exists $i_0\in \{2, 3, 4\}\setminus B_1$, say $i_0 = 4$.
If there exists some $j\in C^*_5\setminus C(x_4)$ such that $G$ contains no $(i, j)_{(x_i, x_4)}$-path for every $i\in \{2, 3\}$,
then $(u x_4, u v)\to (j, 4)$ gives rise to an acyclic edge $(\Delta + 2)$-coloring for $G$, and we are done.
Otherwise, $C^*_5\subseteq C(x_4)$,
or for every $j\in C^*_5\setminus C(u_4)$, $G$ contains an $(i_0, j)_{(u_{i_0}, u_4)}$-path for some $i_0\in \{2, 3\} \cap C(u_4)$.

Assume that there exists $i_0\in \{2, 3, 4\}$ such that there is an $(i_0, j)_{(x_{i_0} x_{j_0})}$-path for some $j\in (C^*_5\setminus \{5\})\setminus C(x_{j_0})$ where $j_0\in \{2, 3, 4\}\setminus \{i_0\}$, say $i_0= 2$.
If $2\not\in C(v_1)$ and $G$ contains no $(2, 5)_{(v_1, v_2)}$-path,
then $(v v_1, u v)\to (2, j)$ gives rise to an acyclic edge $(\Delta + 2)$-coloring for $G$, and we are done;
otherwise, $2\in C(v_1)$, or $G$ contains a $(2, 5)_{(v_1, v_2)}$-path.

For each $i\in \{2, 3, 4\}$, if recoloring $vv_k\to i$ doesn't produce any bichromatic cycle for some $k \in \{1, 2\}$,
then we define
\[
B_i = \{j \in C^*_5 \mid \mbox{there is an } (i, j)_{(x_i, v)}\mbox{-path in } H\};
\]
otherwise, $B_i = \emptyset$.
We note such a definition is consistent with the set $B_i$ defined in Eq.~(\ref{eq02}).

Assume that there exists $i_0\in \{2, 3, 4\}$ such that $|B_{i_0}|\le 2$ and $i_0\not\in C(v_1)$, say $i_0 = 4$.
Then first let $v v_1\to 4$.
If $G$ contains no $(5, 4)_{(v_1, v_2)}$-path,
then by Lemma~\ref{lemma19} $G$ admits an acyclic edge $(\Delta + 2)$-coloring, and we are done.
If $1\not\in C(v_2)$, then $v v_2\to 1$ and $u v\to C^*_5\setminus B_4$ give rise to an acyclic edge $(\Delta + 2)$-coloring for $G$, and we are done.
Otherwise, $G$ contains a $(5, 4)_{(v_1, v_2)}$-path and $1, 4\in C(v_2)$.
One sees that if there exists $\gamma\in \{2, 3\}\setminus (C(v_1)\cup C(v_2))$, then no new bichromatic cycles will be produced if letting $v v_1\to \gamma$.
Since $\{1\}\cup C^*_5\subseteq C(v_1)$, we have $|\{2, 3\}\cap C(v_1)|\le 1$.
It follows that $|\{2, 3\}\cap C(v_2)|\ge 1$ and $C^*_5\setminus C(v_2)\ne \emptyset$.
Thus, we proceed with the following proposition:

\begin{proposition}
\label{prop3302}
\begin{itemize}
\parskip=0pt
\item[{(1)}]
	For every $i\in \{2, 3, 4\}\setminus B_1$, $C^*_5\subseteq C(x_i)$;
	or for every $j\in C^*_5\setminus C(x_i)$, $G$ contains an $(i, j)_{(x_i, x_j)}$-path for some $j\in (\{2, 3, 4\}\cap C(x_i))\setminus \{i\}$.
\item[{(2)}]
	If there exists an $(i_0, j)_{(x_{i_0}, x_{j_0})}$-path for some $i_0\ne j_0\in \{2, 3, 4\}$ and $j\in (C^*_5\setminus \{5\})\setminus C(x_{j_0})$,
	then {\rm (i)}\ $i_0\in C(v_1)$, or {\rm (ii)}\ $G$ contains an $(i_0, 5)_{(v_1, v_2)}$-path.
\item[{(3)}]
	If $|B_{i_0}|\le 2$ or $d(x_0)= 3$ for some $i_0\in \{2, 3, 4\}$,
	then {\rm (i)}\ $i_0\in C(v_1)$, or {\rm (ii)}\ $G$ contains a $(5, i_0)_{(v_1, v_2)}$-path,
	$\{1, 2, 3, 4\}\subseteq (C(v_1)\cup C(v_2))\setminus \{c(vv_1), c(vv_2)\}$, and $C^*_5\setminus C(v_2)\ne\emptyset$.
\item[{(4)}]
	If $|B_i|\le 2$ and $|B_j|\le 2$, or $d(x_i) = d(x_j) = 3$ for some $i, j\in \{2, 3, 4\}$,
	then $\{1, 2, 3, 4\}\subseteq (C(v_1)\cup C(v_2))\setminus \{c(v v_1), c(v v_2))\}$ and $C^*_5\setminus C(v_2)\ne\emptyset$.
\end{itemize}
\end{proposition}

To introduce the next important proposition, we introduce a few symbols used in \cite{BLSNHT11}.
A {\em multiset} is a generalized set where a member can appear multiple times.
If an element $x$ appears $t$ times in the multiset $MS$, then we say that the {\em multiplicity} of $x$ in $MS$ is $t$, and write mult$_{MS}(x) = t$.
The {\em cardinality} of a finite multiset $M S$, denoted by $\|M S\|$, is defined as $\|M S\|= \sum_{x\in M S}$mult$M S(x)= t$.
Let $M S_1$ and $M S_2$ be two multisets.
The {\em join} of $M S_1$ and $M S_2$, denoted $M S_1\biguplus M S_2$, is a multiset that has all the members of $M S_1$ and of $M S_2$.
For $x\in M S_1 \biguplus MS_2$, mult$_{M S_1 \biguplus M S_2}(x) =$ mult$_{M S_1}(x) + $mult$_{M S_2}(x)$.
Clearly, $\|M S_1 \biguplus M S_2\| = \|M S_1\| + \|M S_2\|$.
Specifically, we define
\[
S_u = \biguplus\limits_{i\in \{2, 3, 4\}}(C(x_i)\setminus \{c(u x_i)\}).
\]

When there exist $i_0, j_0\in \{2, 3, 4\}\setminus B_1$, and there exists a $j\in C^*_5\setminus (C(x_{i_0})\cup C(x_{j_0}))$, say $3, 4\not\in B_1$,
if $G$ contains no $(2, j)_{(x_2, x_3)}$-path, then $(u x_3, u v)\to (j, 3)$,
or otherwise, $(u x_4, u v)\to (j, 4)$;
the re-coloring gives rise to an acyclic edge $(\Delta + 2)$-coloring for $G$, and we are done.

In the other case, we proceed with the following proposition and let $(u u_1, u u_2, u u_3, u u_4)_c = (1, 2, 3, 4)$:

\begin{proposition}
\label{prop3303}
\begin{itemize}
\parskip=0pt
\item[{(1)}]
	If there exist $i, j\in \{2, 3, 4\}\setminus B_1$, then $C^*_5\subseteq C(x_i)\cup C(x_j)$.
\item[{(2)}]
	If $2, 3, 4\not\in B_1$, then mult$_{S_u}(i)\ge 2$ for every $j\in C^*_5$.
\end{itemize}
\end{proposition}

\subsubsection{Configuration ($A_{5.1}$): $d(u_3) = d(u_4) = 3$}
If $3, 4\not\in B_1$, then $C^*_5\subseteq C(u_3)\cup C(u_4)$ by Proposition~\ref{prop3303}(1).
We assume w.l.o.g. that $C(u_3) = \{3, \gamma_5, \gamma_6\}$.
Then $(u u_3, u v)\to (\gamma_7, 3)$ gives rise to an acyclic edge $(\Delta + 2)$-coloring for $G$, and we are done.
Otherwise, $\{3, 4\}\cap B_1\ne \emptyset$ and assume w.l.o.g. that $3\in B_1$.
It follows that $C(u_1) = C\setminus \{2, 4\}$.

If $G$ contains no $(2, i)_{(u_2, u_4)}$-path for every $i \in C^*_5\setminus C(u_4)$,
then $C(u_4) = \{4, 3, \gamma_5\}$, $C(u_3) = \{3, \gamma_6, \gamma_7\}$ by Proposition~\ref{prop3302}(1) and
then $(u u_1, u u_4, u v)\to (4, 1, 6)$ gives rise to an acyclic edge $(\Delta + 2)$-coloring for $G$, and we are done.
Otherwise, $2\in C(u_4)$ and $G$ contains a $(2, \alpha_0)_{(u_2, u_4)}$-path for some $\alpha_0\in C^*_5\setminus C(u_4)$.
One sees that $4\not\in C(v_1)$.
By Proposition~\ref{prop3302}(3) $G$ contains a $(5, 4)_{(v v_1, v v_2)}$-path, and $C^*_5\setminus C(v_2)\ne\emptyset$.

Let $T_2 = \{i \mid G \mbox{ contains a } (2, i)_{(u_2, u_4)}\mbox{-path}\}$,
and $T_3 = \{i \mid G \mbox{ contains a } (3, i)_{(u_3, u_4)}\mbox{-path}\}$.
We deal with the following three cases of $C(u_4)$.

\paragraph{Case ($A_{5.1}$-1).}\ $C(u_4) = \{4, 2, 1\}$.

By Proposition~\ref{prop3302}(1), $G$ contains a $(2, i)_{(u_2, u_4)}$-path for every $i\in C^*_5$.
If $3\not\in T_2$, then $(u u_4, u v)\to (3, 4)$ and $u u_3\to C^*_5\setminus C(u_3)$ give rise to an acyclic edge $(\Delta + 2)$-coloring for $G$;
otherwise, $3\in T_2$ and $C(u_2) = C^*_5\cup \{2, 3\}$.
If $4\not\in C(u_3)$, then $(u u_2, u u_4, u v)\to (4, 6, 2)$ gives rise to an acyclic edge $(\Delta + 2)$-coloring for $G$;
otherwise, $4\in C(u_3)$.

If $G$ contains neither a $(3, 1)_{(u_2, u_3)}$-path nor a $(3, 2)_{(u_1, u_3)}$-path,
then switch the colors of $\{u u_1, u u_2\}$ and we are done by Lemma \ref{lemma18};
otherwise, $G$ contains a $(3, 1)_{(u_2, u_3)}$-path and $1\in C(u_3)$, or $G$ contains a $(3, 2)_{(u_1, u_3)}$-path and $2\in C(u_3)$.
If $1\in C(u_3)$, then $(u u_1, u u_2, u u_3, u u_4)\to (4, 1, 2, 3)$;
if $2\in C(u_3)$, then $(u u_1, u u_2, u u_3, u u_4)\to (2, 4, 1, 3)$;
either way we are done by Lemma \ref{lemma18}.

\paragraph{Case ($A_{5.1}$-2).}\ $C(u_4) = \{4, 2, 6^*\}$.
(By $6^*$, we mean that the same discussion applies to the symmetric case where $C(u_4) = \{4, 2, 5\}$, after recoloring $v v_2\to 6$.
	The succeeding $^*$'s have the same meaning.)

By Proposition~\ref{prop3302}, $(C^*_5\setminus \{6\})\subseteq T_2$.
When $3\not\in T_2$, if $(C^*_5\setminus \{6\})\setminus C(u_3)\ne\emptyset$, then $u u_3\to (C^*_5\setminus \{6\})\setminus C(u_3)$;
or if $C(u_3) = \{3, 5, 7\}$, then $u u_3\to 6$;
together with $(u u_4, u v)\to (3, 4)$, this gives rise to an acyclic edge $(\Delta + 2)$-coloring for $G$, and we are done.
In the other case, $3\in T_2$.

{Case 2.1.}\ $G$ contains no $(3, 4)_{(u_1, u_3)}$-path.
If $1\not\in T_2$, then $(u u_1, u u_4)\to (4, 1)$ and we are done by Lemma \ref{lemma18}.
Otherwise, $1\in T_2$, $C(u_2) = C\setminus \{4, 6\}$, and $G$ contains a $(3, 6)_{(u_2, u_3)}$-path by Proposition~\ref{prop3302}(1).
If $4\not\in C(u_3)$, then $(u u_2, u u_4, u v)\to (4, 7, 2)$;
or if $C(u_3) = \{3, 4, 6\}$, then $(u u_2, u u_3, u u_4, u v)\to (6, 2, 3, 4)$;
this gives rise to an acyclic edge $(\Delta + 2)$-coloring for $G$, and we are done.

{Case 2.2.}\ $G$ contains a $(3, 4)_{(u_1, u_3)}$-path and $c(u_3 y_1) = 4$ with $3\in C(y_1)$.
If $4\not\in C(u_2)$, then $(u u_2, u u_4, u v)\to (4, 7, 2)$ gives rise to an acyclic edge $(\Delta + 2)$-coloring for $G$;
otherwise, $4\in C(u_2)$ and $C(u_2) = C\setminus \{1, 6\}$.
If $2\not\in C(u_3)$, then $(u u_1, u u_2, u u_3, u u_4)\to (4, 1, 2, 3)$ and we are done by Lemma \ref{lemma18};
otherwise, $2\in C(u_3)$ and $C(u_3) = \{3, 4, 2\}$.
If there exists a color $j\in C^*_5\setminus C(y_1)$,
then $(u u_1, u u_3, u u_4)\to (4, j, 1)$ and $vv_2\to C^*_5\setminus C(v_2)$ if $j = 5$, and we are done by Lemma \ref{lemma18};
otherwise, $C(y_1) = C^*_5\cup \{3, 4\}$.
One sees that $y_1\not\in N(u)\cup N(v)$.
Then $\{u_3 y_1, u u_2\}\to 1$ and $(u u_1, u u_3, u u_4)\to (2, 4, 3)$, and we are done by Lemma \ref{lemma18}.

\paragraph{Case ($A_{5.1}$-3).}\ $C(u_4) = \{4, 2, 3\}$.

W.l.o.g., let $6^*\not\in C(u_3)$.
One sees from Proposition~\ref{prop3302}(1) that $G$ contains a $(2, 6)_{(u_2, u_4)}$-path.
Then $uu_3\to 6$ reduces the proof to Case ($A_{5.1}$-2).

\subsubsection{Configuration ($A_{5.2}$): $d(u_4) = 3$ and $u_1 v\in E(G)$}
Since $u_1v\in E(G)$, we have $v_2\not\in N(u)$.
By Lemma~\ref{lemma19}, $G$ contains a $(2, i)_{(u_2, v)}$-path for every $i\in C^*_5\subseteq C(u_2)$ and
either (i)\ $(v v_1, v v_2)_c = (2, 5)$,
or (ii)\ $(v v_1, v v_2)_c = (5, 2)$.
For (ii), if $1\not\in C(v_2)$, then $v v_2\to 1$ and we are done by Lemma~\ref{lemma19}.
In the other case, $1\in C(v_2)$ and $C(v_2) = C^*_5\cup \{1, 2\}$.

Let $\{i_0, j_0\} = \{1, 2\}$ with $c(v v_{i_0}) = 2$, $c(v v_{j_0}) = 5$.
One sees that $C(v_{i_0}) = C^*_5\cup \{1, 2\}$ and $4\not\in C(v_{i_0})$ for (i) and (ii).
Proposition~\ref{prop3302}(3) implies that $\{2, 3, 4\}\subseteq C(v_{j_0})$ and $C^*_5\setminus C(v_{j_0})\ne \emptyset$.
We separately consider (i) and (ii) below.

\paragraph{Case ($A_{5.2}$-1).}\ $(v v_1, v v_2)_c = (2, 5)$.

If $G$ contains no $(3, j)_{(u_3, u_4)}$-path for some $j\in C^*_5\setminus C(u_4)$, then $(u u_4, u u_1, u v)\to (j, 4, 1)$;
additionally, if $j = 5$, then $v v_2\to C^*_5\setminus C(v_2)$.
This gives rise to an acyclic edge $(\Delta + 2)$-coloring for $G$, and we are done.
Otherwise, $G$ contains a $(3, j)_{(u_3, u_4)}$-path for every $j\in C^*_5\setminus C(u_4)$ and $3\in C(u_4)$.
W.l.o.g., let $6\not\in C(u_4)$ and then $G$ contains a $(3, 6)_{(u_3, u_4)}$-path.

{Case 1.1.}\ $C(u_4) = \{4, 3, 1\}$.
In this case, $C^*_5\subseteq C(u_3)$ and $G$ contains a $(3, j)_{(u_3, u_4)}$-path for every $j\in C^*_5$.
Two possibilities:
\begin{itemize}
\parskip=0pt
\item
	$3\not\in C(u_2)$.
	We first let $(u u_2, u u_4)\to (3, 2)$.
	Next, if $4\not\in C(u_3)$, then let $u u_3\to 4$;
	otherwise, $C(u_3) = C^*_5\cup \{3, 4\}$ and let $(u u_3, u u_1)\to (1, 4)$.
	By Lemma~\ref{lemma18} $G$ admits an acyclic edge $(\Delta + 2)$-coloring, and we are done.
\item
	$3\in C(u_2)$ and $C(u_2) = C^*_5\cup \{2, 3\}$.
	When $2, 4\not\in C(u_3)$, $(u u_2, u u_4)\to (4, 2)$ and we are done by Lemma \ref{lemma18}.
	In the other case, $\{2, 4\}\cap C(u_3)\ne\emptyset$ and $1\not\in C(u_3)$.
	We first let $(u u_1, u u_4)\to (4, 6)$.
	Next, if $G$ contains no $(1, 6)_{(u_3, u_4)}$-path, then $(u u_3, u v)\to (1, 3)$;
	otherwise, $G$ contains a $(1, 6)_{(u_3, u_4)}$-path, let $(u u_2, u v)\to (1, 7)$;
	this gives rise to an acyclic edge $(\Delta + 2)$-coloring for $G$, and we are done.
\end{itemize}

{Case 1.2.}\ $1\not\in C(u_4)$.
In this case, if $i\not\in C(u_3)$ and $G$ contains no $(2, i)_{(u_2, u_3)}$-path for some $i\in \{1, 4\}$, then let $(u u_3, u u_4, u v)\to (i, 6, 3)$;
additionally, if $i = 1$, then let $uv_1\to 4$;
this gives rise to an acyclic edge $(\Delta + 2)$-coloring for $G$, and we are done.

Next we consider $i\in C(u_3)$ or $G$ contains a $(2, i)_{(u_3, u_2)}$-path for every $i\in \{1, 4\}$.
It follows that $1, 4\in C(u_2)\cup C(u_3)$ and, if $\{1, 4\}\setminus C(u_3)\ne\emptyset$, then $2\in C(u_3)$.
Thus, it suffices to assume that $C(u_4) = \{4, 3, 7^*\}$ and $(C^*_5\setminus\{7\})\subseteq C(u_3)$.
Three possibilities:
\begin{itemize}
\parskip=0pt
\item
	If $2, 4\in C(u_3)$ and $1\in C(u_2)$, then $(u u_3, u u_4, u u_1, u v)\to (7, 1, 3, 4)$ gives rise to an acyclic edge $(\Delta + 2)$-coloring for $G$.
\item
	If $2, 1\in C(u_3)$ and $4\in C(u_2)$, then $(u u_2, u u_3, u u_4)\to (3, 4, 2)$ and we are done by Lemma \ref{lemma18}.
\item
	When $1, 4\in C(u_3)$, if $1\not\in C(u_2)$, then $(u u_2, u u_3, u u_1)\to (1, 2, 3)$;
	otherwise, $C(u_2) = C^*_5\cup \{1, 2\}$ and let $(u u_2, u u_4)\to (4, 2)$.
    We are done by Lemma \ref{lemma18}.
\end{itemize}

\paragraph{Case ($A_{5.2}$-2).}\ $(v v_1, v v_2)_c = (5, 2)$.

Since $\{1, 2, 3, 4\}\subseteq C(v_1)$, we have $|C^*_5\setminus C(v_1)\cup \{5\}|\ge 3$, say $6, 7\not\in C(v_1)$.
Proposition~\ref{prop3302} implies that $3\in C(u_4)$,
and we assume w.l.o.g. that $6\not\in C(u_4)$ and $G$ contains a $(3, 6)_{(u_3, u_4)}$-path. 
Since $6\not\in C(u_4)\cup C(v_1)$, by Proposition~\ref{prop3303}, we have $1\in B_2$ and $C(u_2) = C^*_5\cup \{1, 2\}$.

If $G$ contains no $(4, i)_{(v_1, v_2)}$-path for some $i\in \{1, 6\}$, then $(v v_1, v v_2)\to (i, 4)$;
additionally, if $i = 1$, then let $uv_1\to 6$;
we are done by Lemma~\ref{lemma19}.
Otherwise, $G$ contains a $(4, i)_{(v_1, v_2)}$-path for every $i\in \{1, 6\}$.
If $4\not\in C(u_3)$, then $(u u_3, u u_4, u v)\to (4, 6, 3)$ gives rise to an acyclic edge $(\Delta + 2)$-coloring for $G$.
If $2\not\in C(u_3)\cup C(u_4)$, then $(u u_2, u u_4, u u_1)\to (4, 2, 6)$ and we are done by Lemma \ref{lemma18}.
Otherwise, $C(u_3)\cup C(u_4) = \{2, 3, 4\}\cup C^*_5$ and $1\not\in C(u_3)\cup C(u_4)$.
Then $(u u_3, u u_4, u u_1, u v)\to (1, 6, 7, 4)$ gives rise to an acyclic edge $(\Delta + 2)$-coloring for $G$.

\subsubsection{Configuration ($A_{5.3}$): $d(u_4) = 3$, $d(u_3) = 4$, $d(u_2) = 5$, and $u_3 u_2\in E(G)$}
One sees that $N(u)\cap N(v)\ne \emptyset$, and by Lemma~\ref{lemma19} we let $(v v_1, v v_2)_c = (a, 5)$ where $a\in \{1, 2, 3\}$.
We thus distinguish them into the following three cases of $c(v v_1)$.

\paragraph{Case ($A_{5.3}$-1).}\ $c(v v_1) = 3$.

In this case, $\Delta = 5$.
By Lemma~\ref{lemma19}, $C(u_3) = C^*_5\cup \{3\}$ with $3\in C(u_2)$.
One sees that $|B_2|\le 2$ and $|B_4|\le 2$.
Proposition~\ref{prop3302}(4) implies that $C^*_5\setminus C(v_2)\ne \emptyset$. 
Since $1, 2, 4\not\in B_3$, by Proposition~\ref{prop3303}, we have $C^*_5\subseteq C(u_i)\cup C(u_j)$ for every $i, j\in \{1, 2, 4\}$ and
mult$_{S_u}(j)\ge 2$ for every $j\in C^*_5$. 
W.l.o.g., let $c(u_3u_2) = 6^*$ and let $\{\gamma_5, \gamma_7\} = \{5, 7\}$.

{Case 1.1.}\ Assume that $G$ contains no $(1, i)_{(u_1, u_4)}$-path for every $i\in C^*_5\setminus C(u_4)$.
It follows that $2\in C(u_4)$ and further $C(u_4) = \{4, 2, 6\}$, $C(u_2) = C^*_5\cup \{3, 2\}$ by Proposition~\ref{prop3302}.
Then $(u u_2, u u_4, u v)\to (4, 7, 2)$ gives rise to an acyclic edge $(\Delta + 2)$-coloring for $G$.

{Case 1.2.}\ $G$ contains a $(1, i)_{(u_1, u_4)}$-path for some $i\in C^*_5\setminus C(u_4)$ and $1\in C(u_3)$, say $i= \gamma_6$.
Two possibilities:
\begin{itemize}
\parskip=0pt
\item
	$1\in C(u_2)$.
	Propositions~\ref{prop3302}(1) and \ref{prop3303} imply that $C(u_4) = \{4, 1, \gamma_7\}$,
	$C(u_2) = \{3, 2, 1, 6, \gamma_5\}$, and $C^*_5\subseteq C(u_1)$.
	If $3\not\in C(u_1)$, then $(u u_3, u u_4)\to (4, 3)$ and we are done by Lemma \ref{lemma18};
	otherwise, $C(u_1) = C^*_5\cup \{3, 1\}$,
	and then $(u u_1, u u_4, u v)\to (4, 6, 1)$ gives rise to an acyclic edge $(\Delta + 2)$-coloring for $G$, and we are done.
\item
	$1\not\in C(u_2)$.
	We have two situations (i)\ $C^*_5\subseteq C(u_2)$, or (ii)\ $C(u_4) = \{4, 1, \gamma_7\}$, $C(u_2) = \{3, 2, 4, 6, \gamma_5\}$.
	If $3\not\in C(u_1)\cup C(u_4)$, then $u u_4\to 3$, and $(u u_3, u u_2)\to (2, 4)$ for (i),
	or $u u_3\to 4$ for (ii);
	either gives rise to an acyclic edge $(\Delta + 2)$-coloring for $G$.
	If $4\not\in C(u_1)$, then $(u u_1, u u_2, u u_4, u v)\to (4, 1, \gamma_6, 2)$,
	and additionally if $\gamma_6 = 5$ then $vv_2\to C^*_5\setminus C(v_2)$,
	gives rise to an acyclic edge $(\Delta + 2)$-coloring for $G$.

	In the other case, $(\{3, 4\}\cup C^*_5)\subseteq C(u_1)\cup C(u_4)$.
	It follows that (i) holds and $2\not\in C(u_1)\cup C(u_4)$.
	Then $(u u_1, u u_2, u u_4, u v)\to (2, 1, \gamma_6, 4)$ gives rise to an acyclic edge $(\Delta + 2)$-coloring for $G$.
\end{itemize}

\paragraph{Case ($A_{5.3}$-2).}\ $c(v v_1) = 2$.

One sees from Lemma~\ref{lemma19} that if $c(u_3 u_2)\in C^*_5$, then $2\in C(u_3)$.
It follows that $|C^*_5\cap C(u_3)|\le 2$ and $|B_3|\le 2$.
Together with $|B_4|\le 2$, we have $C^*_5\setminus C(v_2)\ne\emptyset$ by Proposition~\ref{prop3302}(4).
By Proposition~\ref{prop3302}(1), if $4\not\in B_2$, then $\{3, 1\}\cap C(u_4)\ne\emptyset$,
and if $3\not\in B_2$, then $\{1, 4\}\cap C(u_3)\ne\emptyset$.

{Case 2.1.}\ $c(u_3 u_2) = 1$.
We have $C(u_2) = C^*_5\cup \{1, 2\}$, $C(u_4) = \{4, a, \gamma_5\}$, and $C(u_3) = \{1, 3, \gamma_6, \gamma_7\}$, where $a\in \{1, 3\}$;
and then $(u u_2, u u_3)\to (3, 2)$ and we are done by Lemma \ref{lemma18}.

{Case 2.2.}\ $c(u_3 u_2) = 4$.
If $4\not\in B_2$, then $C(u_2) = C^*_5\cup \{4, 2\}$, $C(u_4) = \{4, a, \gamma_5\}$, and $C(u_3) = \{4, 3, \gamma_6, \gamma_7\}$, where $a\in \{1, 3\}$;
and then again $(u u_2, u u_3)\to (3, 2)$ and we are done by Lemma \ref{lemma18}.
Next we assume $4\in B_2$, $2\in C(u_3)$, and $G$ contains a $(5, 3)_{v_1, v_2}$-path.
There are two possibilities of $C(u_3)$, as follows:
\begin{itemize}
\parskip=0pt
\item
	$C(u_3) = \{3, 4, 2, 6^*\}$.
	Propositions~\ref{prop3302}(1) and \ref{prop3303}(1) implies that $C(u_4) = \{4, 5, 7\}$ and $5, 7\in C(u_1)$.
    Then, if $6\not\in C(u_1)$, then $(u u_3, u u_1, u v)\to (1, 6, 3)$,
    or if $2\not\in C(u_1)$, then $(u u_3, u u_1, u u_2)\to (1, 2, 3)$,
    or otherwise, $C(u_1) = C^*_5\cup \{1, 2\}$ and then $(u u_1, u u_4, u v)\to (4, 6, 1)$.
   It gives rise to an acyclic edge $(\Delta + 2)$-coloring for $G$, or we are done by Lemma ~\ref{lemma18}.
\item
	$C(u_3) = \{3, 4, 1, 2\}$ with $c(u_3 y_2) = 1$.
    Proposition~\ref{prop3303}(1) implies that $C^*_5\subseteq C(u_1)$.
    If $G$ contains no $(1, i)_{(u_3, u_1)}$-path for some $i\in C^*_5\cup \{4\}$,
    $(u_2 u_3, u u_3, u v)\to (3, i, 3)$ and $u u_4\to C^*_5\setminus C(u_4)$ if $i = 4$ give rise to an acyclic edge $(\Delta + 2)$-coloring for $G$.
    Otherwise, $G$ contains a $(1, i)_{(u_3, u_1)}$-path for every $i\in C^*_5\cup \{4\}$.
    It follows that $C(y_2) = C^*_5\cup \{1, 4\}$, and $y_2\not\in N(u)\cup N(v)$.
    Then $(u_3 u_2, u u_3)\to (1, 6)$ and $\{u_3 y_2, u v\}\to 3$ give rise to an acyclic edge $(\Delta + 2)$-coloring for $G$.
\end{itemize}

{Case 2.3.}\ $c(u_3u_2)= 6^*$ with $2\in C(u_3)$.
If $3, 4\not\in B_2$, then $\{1, 4\}\cap C(u_3)\ne \emptyset$, $\{3, 1\}\cap C(u_4)\ne \emptyset$,
and thus $C^*_5\setminus (C(u_3)\cup C(u_4))\ne\emptyset$, a contradiction to Proposition~\ref{prop3303}(1).
Otherwise, $\{3, 4\}\cap B_2 \ne \emptyset$.
It follows that $\Delta = 6$, $C(u_2) = C(v_1) = C^*_5\cup \{2, a\}$,
and $C^*_5\subseteq C(u_1)\cup C(u_b)$ by Proposition~\ref{prop3303}(1), where $\{a, b\} = \{3, 4\}$.

{Case 2.3.1.}\ $4\in B_2$.
There are two possibilities since $\{1, 4\}\cap C(u_3)\ne\emptyset$:
\begin{itemize}
\parskip=0pt
\item
	$C(u_3) = \{3, 6, 2, 4\}$.
    In this case, $C(u_4) = \{4, 5, 7\}$, $C^*_5\subseteq C(u_1)$, and $G$ contains $(4, 7)_{(u_3, u_4)}$-path.
    If $4\not\in C(u_1)$, then $(u u_3, u u_1, u u_4, u v)\to (7, 4, 1, 3)$ gives rise to an acyclic edge $(\Delta + 2)$-coloring for $G$.
    Otherwise, $C(u_1) = C^*_5\cup \{1, 4\}$.
    Then $(u u_1, u u_2)\to (2, 1)$ and we are done by Lemma \ref{lemma18}.
\item
	$C(u_3) = \{3, 6, 2, 1\}$ with $c(u_3 y_2) = 1$ and $G$ contains $(1, i)_{(u_3, u_1)}$-path for every $i\in \{5, 7\}$.
	When there exists $i\in \{4, 6\}\setminus C(y_2)$, first let $(u_3 u_2, u u_3, u v)\to (3, i, 3)$,
	and next if $i = 4$ and $C(u_4) = \{4, 5, 7\}$, then $u u_4\to 6$,
	or otherwise, $u u_4\to \{5, 7\}\setminus C(u_4)$.
	This gives rise to an acyclic edge $(\Delta + 2)$-coloring for $G$.
	In the other case, $4, 6\in C(y_2)$.
	It follows that $C(y_2) = C^*_5\cup \{1, 4\}$, and $y_2\not\in N(u)\cup N(v)$.
	Then $(u_3 y_2, u u_3, u v)\to (3, 7, 3)$ gives rise to an acyclic edge $(\Delta + 2)$-coloring for $G$.
\end{itemize}

{Case 2.3.2.}\ $3\in B_2$ and $C(v_1) = C^*_5\cup \{3, 2\}$, $G$ contains a $(5, 4)_{(v_1, v_2)}$-path.
Two possibilities:
\begin{itemize}
\parskip=0pt
\item
	$1\not\in C(u_4)$.
	It follows that $C(u_4) = \{4, 3, \gamma_7\}$, $C(u_3) = \{3, 6, 2, \gamma_5\}$,
	and $6, \gamma_5\in C(u_1)$ by Propositions~\ref{prop3302} and \ref{prop3303}.
	If $\gamma_7\not\in C(u_1)$, then $(u u_1, u u_4, u v)\to (\gamma_7, 1, 4)$ gives rise to an acyclic edge $(\Delta + 2)$-coloring for $G$.
	If $3\not\in C(u_1)$, then $(u u_3, u u_1, u u_4, u v)\to (\gamma_7, 3, 1, 4)$ gives rise to an acyclic edge $(\Delta + 2)$-coloring for $G$.
	Otherwise, $C(u_1) = C^*_5\cup \{3, 1\}$.
	If $G$ contains no $(2, 3)_{(u_3, u_4)}$-path, then $(u u_2, u u_4)\to (4, 2)$;
	otherwise, let $(u u_2, u u_1)\to (1, 2)$;
	and we are done by Lemma \ref{lemma18}.
\item
	$1\in C(u_4)$.
	If $G$ contains no $(1, i)_{(u_1, u_4)}$-path for every $i\in (C^*_5\setminus \{5\})\setminus C(u_4)$,
	then $C(u_4) = \{4, 1, 3\}$, $C(u_3) = \{2, 6, 7\}$ and
	$(u u_3, u u_4, u v)\to (5, 6, 4)$ gives rise to an acyclic edge $(\Delta + 2)$-coloring for $G$.
	Otherwise, $G$ contains a $(1, i)_{(u_1, u_4)}$-path for some $\alpha_0\in (C^*_5\setminus \{5\})\setminus C(u_4)$ and
	then $G$ contains a $(5, 1)_{(v_1, v_2)}$-path by Proposition~\ref{prop3302}.

	When $3\not\in C(u_1)\cup C(u_4)$, first let $(uu_1, uv)\to (3, 1)$;
	next, if $4\not\in C(u_3)$, then $u u_3\to i\in C^*_5\setminus C(u_3)$;
	or if $4\in C(u_3)$, then $u u_3\to C^*_5\setminus (C(u_3)\cup C(u_4))$;
	either way this gives rise to an acyclic edge $(\Delta + 2)$-coloring for $G$, and we are done.
	In the other case, $C^*_5\cup \{3\}\subseteq C(u_1)\cup C(u_4)$.

\begin{itemize}
\parskip=0pt
\item
	$C(u_4) = \{4, 1, 3\}$ and thus $C^*_5\subseteq C(u_1)$.
    If $4\not\in C(u_1)$ and $G$ contains no $(4, 3)_{(u_1, u_3)}$-path,
    then $(u u_1, u v)\to (4, 1)$ and $uu_4\to C^*_5 \setminus C(u_3)$ give rise to an acyclic edge $(\Delta+2)$-coloring for $G$.
    Otherwise, either (i)\ $4\in C(u_1)$, or (ii)\ $C(u_1) = C^*_5\cup\{3, 1\}$ and $G$ contains a $(3, 4)_{(u_1, u_3)}$-path.

    For (i), if $1\not\in C(u_3)$, then $(u u_1, u u_2)\to (2, 1)$;
    or otherwise, $(u u_1, u u_2, u u_4)\to (2, 4, 7)$.
    For (ii), if $G$ contains no $(3, 2)_{(u u_3, u u_1)}$-path,
    then $(u u_1, u u_2)\to (2, 1)$;
    or otherwise, $(u u_4, u u_2)\to (2, 4)$.
    Either way we are done by Lemma \ref{lemma18}.
\item
	$3\not\in C(u_4)$.
   	When $C(u_3) = \{3, 6, 2, 4\}$, if $G$ contains no $(3, 1)_{(u_2, u_1)}$-path, then $(u_3 u_2, u u_2)\to (1, 6)$ and we are done by Lemma \ref{lemma18};
	or otherwise, $(u u_3, u u_4, u v)\to (7, 3, 4)$ gives rise to an acyclic edge $(\Delta + 2)$-coloring for $G$.
	Hence, it suffices to assume that $4\not\in C(u_3)$.

	If $C(u_4) = \{4, 1, 2\}$ and thus $C(u_1) = C^*_5\cup \{3, 1\}$,
	then $(u u_1, u u_4, u v)\to (4, 6, 1)$ gives rise to an acyclic edge $(\Delta + 2)$-coloring for $G$;
	otherwise, $2\not\in C(u_4)$ and assume w.l.o.g. that $C(u_4) = \{4, 1, \beta\}$, where $\beta\in C^*_5$.
	If $2\not\in C(u_1)$, then $(u u_2, u u_4)\to (4, 2)$ and we are done by Lemma \ref{lemma18};
	otherwise, $C(u_1) = (C^*_5\setminus \{\beta\})\cup \{1, 2, 3\}$.
	Then $(u u_1, u v)\to (4, 1)$ and $u u_4\to C^*_5\setminus C(u_4)$ give rise to an acyclic edge $(\Delta + 2)$-coloring for $G$.
\end{itemize}
\end{itemize}

\paragraph{Case ($A_{5.3}$-3).}\ $c(v v_1) = 1$.

One sees that if $C(u_3) = C^*_5\cup \{3\}$ with $3\in C(u_2)$, then $|B_2|\le 2$.
Hence, $|B_3|\le 2$ or $|B_2|\le 2$.
Together with $|B_4|\le 2$, we have $C^*_5\setminus C(v_2)\ne \emptyset$ by Proposition~\ref{prop3302}(4).

Assume that $C(u_3) = C^*_5\cup \{3\}$ with $c(u_3u_2) = 7^*$.
We discuss the following two subcases.

First, assume that $3\not\in C(u_1)$.
When $C(u_4)\ne \{4, 2, 7\}$, we first $(u u_3, u v)\to (4, 3)$;
next, if $6, 7\in C(u_4)$, then $u u_4\to 5$,
or if $2\not\in C(u_4)$, then $u u_4\to \{6, 7\}\setminus C(u_4)$,
or if $7\not\in C(u_4)$, then $u u_4\to 7$;
this gives rise to an acyclic edge $(\Delta + 2)$-coloring for $G$.
In the other case, $C(u_4) = \{4, 2, 7\}$ and further $5, 6\in C(u_2)$.
If $1\not\in C(u_2)$, then $(u u_3, u u_1, u u_2)\to (2, 3, 1)$;
otherwise, $C(u_2) = C^*_5\cup \{1, 2\}$ and then $(u u_3, u u_1)\to (1, 3)$;
and we are done by Lemma \ref{lemma18}.

Next, assume that $3\in C(u_1)$ and $C(u_1) = C^*_5\cup \{3, 1\}$.
When $C(u_4)\ne \{4, 2, 7\}$ or $3\not\in C(u_2)$, we first $(u u_3, u u_1)\to (1, 4)$;
next, if $6, 7\in C(u_4)$, then $(u u_4, u v)\to (5, 3)$,
or if $2\not\in C(u_4)$, then $u u_4\to \{6, 7\}\setminus C(u_4)$,
or if $7\not\in C(u_4)$, then $u u_4\to 7$,
or if $3\not\in C(u_2)$, then $u u_4\to 3$;
this gives rise to an acyclic edge $(\Delta + 2)$-coloring for $G$.
In the other case, $C(u_4) = \{4, 2, 7\}$ and further $3, 5, 6\in C(u_2)$.
Then $(u u_1, u u_4)\to (4, 1)$, and we are done by Lemma \ref{lemma18}.

We continue with the following proposition:

\begin{proposition}
\label{prop3331}
	$C^*_5\setminus C(v_2)\ne \emptyset$,
    $C^*_5\setminus C(u_3)\ne\emptyset$,
    and if $3\not\in B_1$, then for every $j\in C^*_5\setminus C(u_3)$,
    $G$ contains an $(i, j)_{(u_3, u_i)}$-path for some $i\in \{2, 4\}\cap C(u_3)$.
\end{proposition}

Assume that $3, 2, 4\not\in B_1$,
and thus by Proposition~\ref{prop3303}(2) mult$_{\biguplus_{j\in \{3, 2, 4\}}C(u_j)\setminus \{c(u u_j)\}} (i)\ge 2$ for every $i\in C^*_5$.
It suffices to assume that $\Delta = 5$, $C(u_4) = \{4, a, 6^*\}$ where $a\in \{3, 2\}$, $C(u_3) = \{3, 4, 5, 7\}$,
and $C(u_2) = C^*_5\cup \{2, 4\}$ with $c(u_3 u_2) = 4$.
If $4\not\in C(u_1)$, then let $(u u_1, u u_4)\to (4, 1)$;
otherwise, $C(u_1) = C^*_5\cup \{1, 4\}$, and let $(u u_3, u u_1, u u_4)\to (6, 3, 1)$;
and we are done by Lemma \ref{lemma18}.

Next we assume that $\{3, 2, 4\}\cap B_1\ne\emptyset$.
One sees from Proposition~\ref{prop3302}(3) that if $i\in \{3, 4\}\setminus B_1$, then $G$ contains a $(5, i)_{(v_1, v_2)}$-path.
We discuss the following three cases for possible values of $B_1$.

{Case 3.1.}\ $2\in B_1$ and $C(u_1) = C(v_1) = C^*_5\cup \{1, 2\}$.
We distinguish whether or not $2\in C(u_4)$:

{Case 3.1.1.}\ $2\not\in C(u_4)$.
By Proposition~\ref{prop3302}(1), we assume w.l.o.g. that $C(u_4) = \{4, 3, 6^*\}$, $5, 7\in C(u_3)$ and $\{2, 4\}\cap C(u_3)\ne\emptyset$.
If $G$ contains no $(2, 4)_{(u_1, u_2)}$-path, then $(u u_1, u u_4)\to (4, 1)$ and we are done by Lemma \ref{lemma18};
otherwise, $G$ contains a $(2, 4)_{(u_1, u_2)}$-path and $4\in C(u_2)$.
If $4\not\in C(u_3)$, then $(u u_3, u u_4, u v)\to (4, 7, 3)$ gives rise to an acyclic edge $(\Delta + 2)$-coloring for $G$;
otherwise, $C(u_3) = \{3, 4, 5, 7\}$ with $c(u_3 u_2) = \gamma_7$ and further
$G$ contains a $(3, 5)_{(u_3, u_4)}$-path, a $(3, 7)_{(u_3, u_4)}$-path, and a $(4, 6)_{(u_3, u_4)}$-path.
One sees that $2, 3, 4\in C(u_2)$.
Then $u u_2\to C^*_5\setminus C(u_2)$ and $(u u_4, u v)\to (2, 4)$ gives rise to an acyclic edge $(\Delta + 2)$-coloring for $G$.

{Case 3.1.2.}\ $2\in C(u_4)$.
Proposition~\ref{prop3331} implies that $C^*_5\setminus C(u_3)\ne\emptyset$.
Thus, by Proposition~\ref{prop3303}(1), $C^*_5 \subseteq C(u_3)\cup C(u_4)$.
It follows that $\Delta= 5$ and $C^*_5\setminus C(u_3)\subseteq C(u_4)$.
We assume w.l.o.g. that $C(u_4) = \{4, 2, 6^*\}$.
When $2\not\in C(u_3)$, $C(u_3) = \{3, 4, 5, 7\}$;
if $G$ contains no $(3, 2)_{(u_2, u_4)}$-path, then $(u u_4, u u_3, u v)\to (3, 6, 4)$;
or otherwise, $(u u_3, u u_1)\to (1, 3)$;
this gives rise to an acyclic edge $(\Delta+2)$-coloring for $G$.
In the other case, $2\in C(u_3)$, $C(u_3) = \{3, 2, 5, 7\}$, and $C^*_5\subseteq C(u_2)$.
If $1\not\in C(u_2)$, then let $(u u_1, u u_2, u u_4)\to (4, 1, 7)$ and we are done by Lemma \ref{lemma18};
otherwise, $C(u_2) = C^*_5\cup \{1, 2\}$.
If $G$ contains no $(1, 2)_{(u_3, u_2)}$-path, then let $(u u_3, u u_1)\to (1, 3)$;
or otherwise, let $(u u_1, u u_4)\to (4, 1)$;
we are done by Lemma \ref{lemma18}.

{Case 3.2.}\ $3\in B_1$ and $C(u_1) = C(v_1) = C^*_5\cup \{1, 3\}$.
Proposition~\ref{prop3302}(3) implies that $G$ contains a $(5, 4)_{(v_1, v_2)}$-path and $C^*_5\setminus C(v_2)\ne \emptyset$.

First, assume that $2\not\in C(u_4)$;
we have $C(u_4) = \{4, 3, 6^*\}$, $\Delta = 5$.
By Propositions~\ref{prop3302}(1) and \ref{prop3303}(1), we have $5, 7\in C(u_3)\cap C(u_2)$.
If $G$ contains neither a $(3, 1)_{(u_3, u_4)}$-path nor a $(3, 4)_{(u_3, u_1)}$-path,
then let $(u u_1, u u_4)\to (4, 1)$ and we are done by Lemma \ref{lemma18};
otherwise, $C(u_3) = \{3, 5, 7, a\}$ with $a\in \{1, 4\}$, and $G$ contains a $(3, 1)_{(u_3, u_4)}$-path or a $(3, 4)_{(u_3, u_1)}$-path.
If $6\not\in C(u_2)$, then $(u u_2, u u_4, u v)\to (6, 2, 4)$ gives rise to an acyclic edge $(\Delta + 2)$-coloring for $G$;
otherwise, $C(u_2) = C^*_5\cup \{3, 2\}$ with $c(u_3 u_2)\in \{5, 7\}$.
Let $(u u_1, u u_2)\to (2, 1)$ and we are done by Lemma \ref{lemma18}.

Next, assume that $2\in C(u_4)$.
We discuss the following three possible cases of $C(u_4)$.

{Case 3.2.1.}\ $C(u_4) = \{4, 2, 1\}$.
In this case $G$ contains a $(2, i)_{(u_2, u_4)}$-path for every $i\in C^*_5$.
If $G$ contains no $(3, 2)_{(u_2, u_4)}$-path,
then $(u u_4, u v)\to (3, 4)$ and $u u_3\to C^*_5\setminus C(u_3)$ give rise to an acyclic edge $(\Delta + 2)$-coloring for $G$;
otherwise, $G$ contains a $(3, 2)_{(u_2, u_4)}$-path and $C(u_2) = C^*_5\cup \{3, 2\}$.
We assume w.l.o.g. that $c(u_3 u_2) = 6^*$.
If $4\not\in C(u_3)$, then $(u u_2, u u_4, u v)\to (4, 6, 2)$;
or if $C(u_3) = \{3, 2, 6, 4\}$, then $(u_3 u_2, u u_4, u v)\to (1, 6, 4)$;
this gives rise to an acyclic edge $(\Delta + 2)$-coloring for $G$.

{Case 3.2.2.}\ $C(u_4) = \{4, 2, 6^*\}$.
In this case, $G$ contains a $(2, i)_{(u_2, u_4)}$-path for every $i\in C^*_5\setminus \{6\}$ by Proposition~\ref{prop3302}(1),
and contains a $(5, 2)_{(v_1, v_2)}$-path by Proposition~\ref{prop3302}(2).
If $G$ contains neither a $(1, 2)_{(u_2, u_4)}$-path nor a $(3, 4)_{(u_3, u_1)}$-path,
then let $(u u_4, u u_1)\to (1, 4)$ and we are done by Lemma \ref{lemma18}.
Otherwise, $G$ contains a $(1, 2)_{(u_2, u_4)}$-path or a $(3, 4)_{(u_3, u_1)}$-path.

{Case 3.2.2.1.}\ $4\not\in (u_3)$.
It follows that $1\in (u_2)$.
If $4\not\in C(u_2)$, then $(u u_2, u u_4, u v)\to (4, 7, 2)$;
otherwise, $C(u_2) = \{2, 5, 7, 1, 4\}$, and then $(u u_3, u u_1, u u_2, u u_4)\to (4, 2, 3, 1)$;
this gives rise to an acyclic edge $(\Delta + 2)$-coloring for $G$, and we are done.

{Case 3.2.2.2.}\ $4\in C(u_3)$.
When $c(u_3u_2) = \gamma_7\in C^*_5\setminus \{6\}$, i.e., $C(u_3) = \{3, 2, 4, \gamma_7\}$,
if $3\not\in C(u_2)$, then $(u u_2, u u_3, u v)\to (3, \gamma_5, 2)$;
or if $6\not\in C(u_2)$, then $(u u_2, u u_4, u v)\to (6, \gamma_7, 2)$;
or if $C(u_2) = C^*_5\cup \{3, 2\}$, then $(u_3 u_4, u u_4, u v)$ $\to$ $(1, \gamma_7, 4)$;
this gives rise to an acyclic edge $(\Delta + 2)$-coloring for $G$.
When $c(u_3 u_2)\not\in C^*_5\setminus \{6\}$, there are two possibilities:
\begin{itemize}
\parskip=0pt
\item
	$(C^*_5\setminus \{6\})\subseteq C(u_3)$.
	It follows that $C(u_3) = \{3, 4, 5, 7\}$ with $c(u_3 u_2) = 4$.
	Since $\{1, 3\}\cap C(u_2)\ne\emptyset$, we have $6\not\in C(u_2)$.
	Then $(u u_2, u v)\to (6, 2)$ gives rise to an acyclic edge $(\Delta + 2)$-coloring for $G$.
\item
	There exists an $\eta\in (C^*_5\setminus \{6\})\setminus C(u_3)$.
	If $G$ contains no $(3, 2)_{(u_2, u_4)}$-path, then $(u u_3, u u_4, u v)\to (\eta, 3, 4)$ gives rise to an acyclic edge $(\Delta + 2)$-coloring for $G$;
	otherwise, $G$ contains a $(3, 2)_{(u_2, u_4)}$-path and $3\in C(u_2)$.
	One sees from Proposition~\ref{prop3302}(1) that if $c(u_3 u_2)\in \{1, 4\}$, then $6\in C(u_3)$.
	It follows that $G$ contains a $(3, 4)_{(u_3, u_1)}$-path.

\begin{itemize}
\parskip=0pt
\item
	If $c(u_3u_2)\in \{1, 6\}$, then $(u u_2, u u_4, u v)\to (4, 7, 2)$ gives rise to an acyclic edge $(\Delta + 2)$-coloring for $G$.
\item
	$c(u_3 u_2) = 4$ and $6\in C(u_3)$.
    If $2\not\in C(u_3)$, then let $(u u_3, u u_1, u u_2, u u_4)\to (2, 4, 1, 7)$ and we are done by Lemma \ref{lemma18};
    otherwise, $C(u_3) = \{3, 4, 2, 6\}$, and then $(u_3 u_2, u u_2, u u_4, u v)\to (1, 4, 7, 2)$ gives rise to an acyclic edge $(\Delta + 2)$-coloring for $G$.
\end{itemize}
\end{itemize}

{Case 3.2.3.}\ $C(u_4) = \{4, 2, 3\}$.
In this case, we assume w.l.o.g. that $6^*\not\in C(u_3)$.
Then $u u_3\to 6$ reduces the proof to the above Case 3.2.2.

{Case 3.3.}\ $4\in B_1$.
Since $3, 2\not\in B_1$, we have $C^*_5\subseteq C(u_3)\cup C(u_2)$ by Proposition~\ref{prop3303}(1).
By Proposition~\ref{prop3331}, $|B_3|\le 2$, and thus $G$ contains a $(5, 3)_{(v_1, v_2)}$-path by Proposition~\ref{prop3302}(3).
One sees that if $C^*_5\setminus C(u_2)\ne\emptyset$, then $G$ contains a $(5, 2)_{(v_1, v_2)}$-path.
For every $i\in C^*_5\setminus C(u_3)$,
let $T_2 = \{i\in  C^*_5\setminus C(u_3) \mid G \mbox{ contains a } (2, i)_{(u_3, u_2)}\mbox{-path}\}$ and
$T_4 = \{i\in  C^*_5\setminus C(u_3) \mid G \mbox{ contains a } (4, i)_{(u_3, u_4)}\mbox{-path}\}$.

W.l.o.g., let $6^*\in C(u_2)\setminus C(u_3)$ and $\{\gamma_5, \gamma_7\} = \{5, 7\}$.
If $6\not\in T_2$, $1\not\in C(u_4)$, and $G$ contains neither a $(1, 2)_{(u_2, u_4)}$-path nor a $(1, 6)_{(u_3, u_4)}$-path,
then let $(u u_3, u u_1, u u_4)\to (6, 3, 1)$ and we are done by Lemma \ref{lemma18}.
If $1\not\in C(u_3)$ and $G$ contains none of a $(1, 2)_{(u_3, u_2)}$-path, a $(4, 1)_{(u_3, u_4)}$-path and a $(4, 3)_{(u_1, u_4)}$-path,
then let $(u u_3, u u_1)\to (1, 3)$ and we are done by Lemma \ref{lemma18}.
If $6\not\in C(u_4)$, $4\not\in C(u_2)$, and $G$ contains no $(3, 4)_{(u_3, u_2)}$-path,
then $(u u_2, u u_4, u v)\to (4, 6, 2)$ gives rise to an acyclic edge $(\Delta + 2)$-coloring for $G$.
If $6\not\in T_4$, $3\not\in C(u_2)$, and $G$ contains no $(3, 4)_{(u_2, u_4)}$-path,
then $(u u_3, u u_2, u v)\to (6, 3, 2)$ gives rise to an acyclic edge $(\Delta + 2)$-coloring for $G$.

We continue with the following proposition:

\begin{proposition}
\label{prop3332}
\begin{itemize}
\parskip=0pt
\item[{(1)}]
	(i). If $6\not\in T_2$, then $1\in C(u_4)$, or $G$ contains a $(1, 2)_{(u_2, u_4)}$-path or a $(1, 6)_{(u_3, u_4)}$-path;
	(ii). $\{1, 2\}\cap (C(u_3)\cup C(u_4))\ne\emptyset$.
\item[{(2)}]
	(i). $1\in C(u_3)$, or $G$ contains a $(1, 2)_{(u_3, u_2)}$-path, or a $(4, 1)_{(u_3, u_4)}$-path, or a $(4, 3)_{(u_1, u_4)}$-path;
	(ii). If $1, 2\not\in C(u_3)$, then $\{3, 1\}\cap C(u_4)\ne \emptyset$, and $G$ contains a $(4, 1)_{(u_3, u_4)}$-path or a $(4, 3)_{(u_1, u_4)}$-path.
\item[{(3)}]
	If $6\not\in C(u_4)$, then $4\in C(u_2)$, or $G$ contains a $(3, 4)_{(u_3, u_2)}$-path.
\item[{(4)}]
	If $6\not\in C(u_4)$ or $6\not\in T_4$, then $3\in C(u_2)$, or $G$ contains a $(3, 4)_{(u_2, u_4)}$-path.
\end{itemize}
\end{proposition}

We discuss the two subcases on whether $2\in C(u_3)$ or not.

{Case 3.3.1.}\ $2\not\in C(u_3)$.
Proposition~\ref{prop3302}(1) implies that $(C^*_5\setminus C(u_3))\subseteq T_4\subseteq C(u_4)$ and $4\in C(u_3)$.
One sees from Proposition~\ref{prop3332}(1) that, if there exists $6, \gamma_5\not\in C(u_3)$, then $C(u_4) = \{4, 6, \gamma_5\}$;
and from Proposition~\ref{prop3332}(2) that, if $C(u_3) = \{3, 4, 5, 7\}$, then $C(u_4) = \{4, 6, 1\}$.
Then $(u u_3, u u_4)\to (6, 3)$ reduces the proof to the above Case 3.2.

{Case 3.3.2.}\ $2\in C(u_3)$.

{Case 3.3.2.1.}\ $C(u_3) = \{3, 1, 2, 4\}$.
Proposition~\ref{prop3303}(1) implies that $C(u_2) = \{5, 6, 7, 2, c(u_2u_3)\}$ and $\Delta = 5$.
First let $(u u_2, u v)\to (3, 2)$;
next, if $C(u_4) = \{4, 6, 7\}$, then $(u u_2, u u_4)\to (6, 5)$;
or otherwise, $u u_3\to \{6, 7\}\setminus C(u_4)$;
this gives rise to an acyclic edge $(\Delta + 2)$-coloring for $G$.

{Case 3.3.2.2.}\ Assume that $C(u_3) = \{3, 2, \gamma_5, \gamma_7\}$ with $c(u_3 u_2) = \gamma_7$.
Two possibilities:
\begin{itemize}
\parskip=0pt
\item
	$3\not\in C(u_2)$.
	One sees from Proposition~\ref{prop3332}(4) that $G$ contains a $(3, 4)_{(u_2, u_4)}$-path,
    and from Proposition~\ref{prop3332}(2) that $G$ contains a $(2, 1)_{(u_3, u_2)}$-path.
    It follows that $C(u_2) = \{2, 1, 4, 6, \gamma_7\}$.
	Then $(u_3 u_2, u u_3, u v)\to (3, \gamma_7, 3)$ gives rise to an acyclic edge $(\Delta + 2)$-coloring for $G$.
\item
	$3\in C(u_2)$.
	If $1, \gamma_5\not\in C(u_2)$, then $(u_3 u_2, u u_3, u v)\to (1, \gamma_7, 3)$ gives rise to an acyclic edge $(\Delta + 2)$-coloring for $G$;
	otherwise, $\{1, \gamma_5\}\cap C(u_2)\ne \emptyset$, i.e., (i)\ $C(u_2) = \{3, 1, 2, 6, \gamma_7\}$, or (ii)\ $C(u_2) = C^*_5\cup \{3, 2\}$.
	Proposition~\ref{prop3332}(3) implies that $6\in C(u_4)$.
	For (i), if $2\not\in C(u_4)$, then $(u_3 u_2, u u_3, u v)\to (4, \gamma_7, 3)$;
	or if $C(u_4) = \{4, 2, 6\}$, then $(u u_2, u u_4, u v)\to (4, \gamma_7, 2)$;
	this gives rise to an acyclic edge $(\Delta + 2)$-coloring for $G$.
	For (ii), we have $C(u_4) = \{4, 3, 6\}$ by Proposition~\ref{prop3332}(2),
	then let $(u u_1, u u_2)\to (2, 1)$ and we are done by Lemma \ref{lemma18}.
\end{itemize}

{Case 3.3.2.3.}\ $|C^*_5\cap C(u_3)| = 1$.
We assume w.l.o.g. from Proposition~\ref{prop3303}(1) that $7^*\in C(u_3)$ and $C^*_5\setminus \{7\}\subseteq C(u_2)$.
When $4\not\in C(u_2)$ and $G$ contains no $(3, 4)_{(u_3, u_2)}$-path,
we first $(u u_2, u v)\to (4, 2)$;
next, if $C(u_4) = \{4, 5, 6\}$, then $u u_4\to 7$;
or otherwise there exists a $\beta\in \{5, 6\}\setminus C(u_4)$, and then $u u_4\to \beta$ (additionally, if $\beta = 5$ then $v v_2\to C^*_5\setminus C(v_2)$);
this gives rise to an acyclic edge $(\Delta + 2)$-coloring for $G$.

In the other case, $4\in C(u_2)$ or $G$ contains a $(3, 4)_{(u_3, u_2)}$-path.
Three possibilities:
\begin{itemize}
\parskip=0pt
\item
	$c(u_3 u_2) = 1$ and $C(u_3) = \{3, 2, 7, 1\}$.
    It follows that $C(u_2) = \{1, 2, 4, 5, 6\}$.
    Further, Propositions~\ref{prop3302}(1) and \ref{prop3332}(4) imply that $C(u_4) = \{4, 7, 3\}$.
    Then let $(u_3 u_2, u u_3, u u_2, u u_1)\to (3, 5, 1, 2)$ and we are done by Lemma \ref{lemma18}.
\item
	$c(u_3 u_2) = 4$.
    One sees that if $1, 7\not\in C(u_2)$, then we can let $u_3 u_2\to 1$.
    Hence, it suffices to assume that (i)\ $C(u_2) = \{1, 2, 4, 5, 6\}$, or (ii)\ $C(u_2) = C^*_5\cup \{2, 4\}$.
	For (i), $7\in C(u_4)$ by Proposition~\ref{prop3302}(1), and $G$ contains a $(2, 5)_{(v_1, v_2)}$-path by Proposition~\ref{prop3302}(3);
	then $(u u_2, u v)\to (3, 2)$ and $u u_3\to \{5, 6\}\setminus C(u_4)$ give rise to an acyclic edge $(\Delta + 2)$-coloring for $G$.
    For (ii), $C(u_4) = \{4, 3, 6\}$ by Proposition~\ref{prop3332}(2,4),
    then let $(u u_1, u u_2)\to (2, 1)$ and we are done by Lemma \ref{lemma18}.
\item
	$c(u_3 u_2) = 7$.
    If $C(u_3) = \{3, 1, 2, 7\}$ and thus $C(u_2) = C^*_5\cup \{2, 4\}$,
    then $(u_3 u_2, u u_3, u v)\to (3, 7, 3)$ gives rise to an acyclic edge $(\Delta + 2)$-coloring for $G$;
    otherwise, $C(u_3) = \{3, 2, 4, 7\}$.
    One sees that if $1\not\in C(u_2)$ and $G$ contains no $(1, 4)_{(u_3, u_2)}$-path, then we can let $u_3 u_2\to 1$.
    We have $C(u_2) = C^*_5\cup \{2, 4\}$ and $G$ contains a $(4, 1)_{(u_3, u_2)}$-path.
	Further, Proposition~\ref{prop3332}(2) implies that $G$ contains a $(4, 3)_{(u_1, u_4)}$-path and $3\in C(u_4)$.
    Then $(u u_2, u v)\to (3, 2)$ and $u u_3\to \beta\in \{5, 6\}\setminus C(u_4)$ (additionally, if $\beta = 5$ then $v v_2\to C^*_5\setminus C(v_2)$)
	give rise to an acyclic edge $(\Delta + 2)$-coloring for $G$.
\end{itemize}

\subsubsection{Configuration ($A_{5.4}$): $d(u_3) = d(u_4) = 4$, and $u_3 u_4\in E(G)$}
One sees that $N(u)\cap N(v)\ne \emptyset$, and we let $(v v_1, v v_2)_c = (a, 5)$ where $a\in \{3, 1\}$.

\paragraph{Case ($A_{5.4}$-1).}\ $c(v v_1) = 3$.

In this case, $c(u_3 u_4) = \gamma_5$ with $3\in C(u_4)$ by Lemma~\ref{lemma19}.
Proposition~\ref{prop3302}(1) implies that $\{1, 2\}\cap C(u_4)\ne \emptyset$.
We assume w.l.o.g. that $C(u_4) = \{4, \gamma_5, 3, 1\}$.
Then let $(u_3 u_4, u u_3)\to (2, \gamma_5)$.
If $\gamma_5 = 5$, then since $3\not\in B_5$ we are done by Lemma~\ref{lemma19};
otherwise, $\gamma_5\ne 5$, and then $u v\to \{6, 7\}\setminus \{\gamma_5\}$ gives rise to an acyclic edge $(\Delta + 2)$-coloring for $G$.

\paragraph{Case ($A_{5.4}$-2).}\ $c(v v_1) = 1$ and $C^*_5\subseteq C(u_1)\cap C(v_1)$.

If $G$ contains no $(5, i)_{(v_1, v_2)}$-path for some $i\in \{3, 4\}\setminus C(v_1)$,
then $v v_1\to i$ reduces the proof to the above Case ($A_{5.4}$-1).

Assume that $1\not\in C(v_2)$.
One sees that $\{3, 4\}\setminus C(v_1)\ne \emptyset$ and we assume w.l.o.g. that $3\not\in C(v_1)$.
Note that no new bichromatic cycles will be produced if letting $(v v_1, v v_2)\to (3, 1)$.
By Lemma~\ref{lemma19}, $G$ contains a $(3, i)_{(u_3, v_1)}$-path for every $i\in C^*_5$, $C(u_3) = C^*_5\cup \{3\}$ with $c(u_3 u_4) = \gamma_5$, and $3\in C(u_4)$.
When $2\not\in C(u_4)$, we first let $u u_4\to C^*_5\setminus C(u_4)$ and $uv\to 4$;
next, if $4\not\in B_1$, then we are done;
or if $4\in B_1$, then $(v v_1, v v_2)\to (3, 1)$ gives rise to an acyclic edge $(\Delta + 2)$-coloring for $G$.
In the other case, $C(u_4) = \{4, 2, 3, \gamma_5\}$,
and $(u_3 u_4, u u_3, u v)\to (1, \gamma_5, 3)$ gives rise to an acyclic edge $(\Delta + 2)$-coloring for $G$.
One sees that letting $v v_1\to 2$ gives also an acyclic edge coloring of $H = G - u v$.
We thus have $2\in C(v_1)\cup C(v_2)$, and continue with the following proposition:

\begin{proposition}
\label{prop3341}
	$\{1, 2, 3, 4\}\subseteq (C(v_1)\cup C(v_2))\setminus \{c(v v_1), c(v v_2)\}$,
	$G$ contains a $(5, i)_{(v_1, v_2)}$-path for every $i\in \{3, 4\}\setminus C(v_1)$,
	$|C(v_2)\cap \{1, 2, 3, 4\}|\ge 3$,
	and $C^*_5\setminus C(v_2)\ne\emptyset$.
\end{proposition}

\begin{lemma}
\label{lemma20}
For every $i\in \{3, 4\}\setminus C(u_1)$, $C^*_5\setminus C(u_i)\ne\emptyset$ and, for every $j\in C^*_5\setminus C(u_4)$,
$G$ contains a $(j, k)_{(u_i, u)}$-path for some $k\in (\{2, 3, 4\}\cap C(u_i))\setminus \{i\}$;
or otherwise an acyclic edge $(\Delta + 2)$-coloring for $G$ can be obtained in $O(1)$ time.
\end{lemma}
\begin{proof}
Assume that $4\not\in C(u_1)$ and $C(u_4) = C^*_5\cup \{4\}$ with $c(u_3 u_4) = 6^*$.
It follows that $\Delta = 5$.
If $G$ contains no $(2, i)_{(u_3, u_2)}$-path for some $i\in C^*_5\setminus C(u_3)$,
then $(u u_3, u u_4, u v)\to (i, 3, 4)$ (additionally, if $i = 5$ then $i\to C^*_5\setminus C(v_2)$) gives rise to an acyclic edge $(\Delta + 2)$-coloring for $G$;
otherwise, $5, 7\in C(u_3)$ or $2\in C(u_3)$, and $G$ contains a $(2, i)_{(u_3, u_2)}$-path for every $i\in C^*_5\setminus C(u_3)$.

If $G$ contains neither a $(3, 4)_{(u_3, u_1)}$-path nor a $(2, 4)_{(u_1, u_2)}$-path,
then let $(u u_1, u u_4)\to (4, 1)$ and we are done by Lemma \ref{lemma18};
otherwise, $G$ contains a $(3, 4)_{(u_3, u_1)}$-path or a $(2, 4)_{(u_1, u_2)}$-path, and $\{2, 3\}\cap C(u_1)\ne \emptyset$.
If $C(u_1) = C^*_5\cup \{1, 3\}$ and $4\in C(u_3)$,
then $C(u_3) = \{3, 6, 4, 2\}$ and $(u_3 u_4, u u_4, u v)\to (1, 6, 4)$ gives rise to an acyclic edge $(\Delta + 2)$-coloring for $G$;
otherwise, $C(u_1) = C^*_5\cup \{1, 2\}$ and $G$ contains a $(2, 4)_{(u_2, u_1)}$-path with $4\in C(u_2)$.

If $6\not\in C(u_2)$, then $(u u_2, u u_4, u v)\to (6, 2, 4)$ gives rise to an acyclic edge $(\Delta + 2)$-coloring for $G$;
otherwise, $6\in C(u_2)$.
If $1\not\in C(u_2)$ and $G$ contains no $(1, 3)_{(u_2, u_3)}$-path,
then let $(u u_1, u u_2, u u_4)\to (4, 1, 2)$ and we are done by Lemma \ref{lemma18};
otherwise, $1\in C(u_2)$ or $G$ contains a $(1, 3)_{(u_2, u_3)}$-path.
If $1\not\in C(u_3)$ and $G$ contains neither a $(1, 2)_{(u_2, u_3)}$-path nor a $(2, 3)_{(u_1, u_2)}$-path,
then let $(u u_3, u u_1)\to (1, 3)$ and we are done by Lemma \ref{lemma18};
otherwise, $1\in C(u_3)$ or $G$ contains a $(1, 2)_{(u_2, u_3)}$-path or a $(2, 3)_{(u_1, u_2)}$-path.

It follows that (i)\ $C(u_3) = C^*_5\cup \{3\}$ and $C(u_2) = C\setminus \{5, 7\}$,
or (ii) $C(u_3) = \{3, 2, 6, \gamma_5\}$ and $C(u_2)= \{1, 2, 4, 6, \gamma_7\}$, where $\{\gamma_5, \gamma_7\}= \{5, 7\}$.
For (i), first let $(u u_2, u u_3, u v)\to (7, 2, 3)$;
next, if $G$ contains no $(4, 7)_{(u_2, u_4)}$-path, then we are done;
or otherwise, let $(u u_3, u u_4)\to (4, 2)$ and we are done by Lemma \ref{lemma18}.
For (ii), let $(u u_1, u u_3, u u_4)\to (3, 4, 1)$ and we are done by Lemma \ref{lemma18}.
\end{proof}

One sees that $\{3, 4\}\setminus C(u_1)\ne\emptyset$.
In the sequel we assume w.l.o.g. that $4\not\in C(u_1)$ and let $\{\gamma_5, \gamma_7\}= \{5, 7\}$.
It follows that $C(u_4)\cap \{2, 3\}\ne \emptyset$,
and if $3\not\in C(u_1)$, then $C(u_3)\cap \{2, 4\}\ne \emptyset$.

{Case 2.1.}\ $C(u_1) = C^*_5\cup \{1, 2\}$.
In this case, one sees that $3, 4\not\in C(u_1)$.
If $2\not\in C(u_3)\cup C(u_4)$, then $C(u_3) = \{3, 4, 6, 7\}$, $C(u_4) = \{4, 3, 7, 5\}$,
and $(u_3 u_4, u u_4, u v)\to (1, 7, 4)$ gives rise to an acyclic edge $(\Delta + 2)$-coloring for $G$;
otherwise, by symmetry we assume that $2\in C(u_3)$.

{Case 2.1.1.}\ $c(u_3 u_4) = 2$ and, by Proposition~\ref{prop3303}, $C^*_5\subseteq C(u_3)\cup C(u_4)$.
W.l.o.g., let $C(u_4) = \{2, 4, 5, 7\}$ and $6\in C(u_3)$.
Note that there exists a $\gamma_7\in \{5, 7\}\setminus C(u_3)$.

If $4\not\in C(u_3)$ and $G$ contains no $(4, i)_{(u_3, u_4)}$-path for every $i\in C(u_3)\cap C(u_4)$,
then $(u_3 u_4, u u_4, u v)\to (4, 6, 4)$ gives rise to an acyclic edge $(\Delta+2)$-coloring for $G$;
otherwise, $\Delta = 5$ and, if $4\in C(u_3)$, then $G$ contains a $(4, i)_{(u_3, u_4)}$-path for some $i\in \{5, 7\}$.
If $4\in C(u_3)$, that is, $C(u_3) = \{3, 2, 4, 6\}$, then $(u_3 u_4, u u_4, u v)\to (1, 6, 4)$ gives rise to an acyclic edge $(\Delta + 2)$-coloring for $G$;
otherwise, $\gamma_5\in C(u_3)$ and $C(u_3) = \{3, 2, 6, \gamma_5\}$.

Next, if $3\not\in C(u_2)$, then let $(u u_3, u u_1)\to (1, 3)$;
or if $4\not\in C(u_2)$, then let $(u u_4, u u_1)\to (1, 4)$;
or otherwise we have $C(u_2) = \{2, 3, 4, 6, \gamma_7\}$, and let $(u u_1, u u_2, u u_4)\to (4, 1, 6)$;
we are done by Lemma \ref{lemma18}.

{Case 2.1.2.}\ $2\in (C(u_3)\cup C(u_4))\setminus \{c(u_3 u_4)\}$.
We assume w.l.o.g. that $2\in C(u_3)\setminus \{c(u_3 u_4)\}$.
One sees that $C^*_5\subseteq C(u_3)\cup C(u_4)$ and $\{2, 3\}\cap C(u_4)\ne \emptyset$.
We have $\Delta = 5$ and w.l.o.g. let $C(u_3) = \{3, 2, \gamma_7, 6\}$.

If $C(u_4) = \{4, 3, 5, 7\}$, then $(u_3 u_4, u u_4, u v)\to (1, \gamma_7, 4)$ gives rise to an acyclic edge $(\Delta + 2)$-coloring for $G$;
otherwise, $C(u_4) = \{4, 2, 7, 5\}$, and $6, \gamma_5\in C(u_2)$.
If $1\not\in C(u_2)$, then let $(u u_4, u u_2, u u_1)\to (6, 1, 4)$ and we are done by Lemma \ref{lemma18};
otherwise, $1\in C(u_2)$.

If $G$ contains no $(1, 2)_{(u_3, u_4)}$-path, then let $u_3 u_4\to 1$ and we are done by Proposition~\ref{prop3303}(1);
otherwise, $G$ contains a $(1, 2)_{(u_3, u_4)}$-path.
Next, if $4\not\in C(u_2)$, then let $(u u_1, u u_4)\to (4, 1)$;
or otherwise, let $(u u_1, u u_3)\to (3, 1)$;
we are done by Lemma \ref{lemma18}.

{Case 2.2.}\ $C(u_1) = C^*_5\cup \{1, 3\}$ and $C^*_5\subseteq  C(u_4)\cup C(u_2)$.
If $1\not\in C(u_4)$ and $G$ contains none of $(1, 3)_{(u_3, u_4)}$-path, $(1, 2)_{(u_2, u_4)}$-path and $(3, 4)_{(u_3, u_1)}$-path,
then let $(u u_4, u u_1)\to (1, 4)$ and we are done by Lemma \ref{lemma18};
otherwise, $1\in C(u_4)$ or $G$ contains one of $(1, 3)_{(u_3, u_4)}$-path, $(1, 2)_{(u_2, u_4)}$-path and $(3, 4)_{(u_3, u_1)}$-path.
We distinguish the following two cases on whether $3\in C(u_4)$ or not.

{Case 2.2.1}\ $3\in C(u_4)$, and there are four possibilities of $C(u_4)$.

{Case 2.2.1.1}\ $C(u_4) = \{4, 3, \gamma_5, \gamma_7\}$ with $c(u_3 u_4) = \gamma_5$.
We have $\{1, 4\}\cap C(u_3)\ne\emptyset$, $6\in C(u_3)$ and $\Delta = 5$,
We first let $u_1 u_3\to i\in \{1, 4\}\setminus C(u_3)$ and $(u u_4, u v)\to (\gamma_5, 4)$.
Next, if $G$ contains no $(3, 1)_{(u_1, u_4)}$-path, then we are done;
otherwise, $i = 1$, and $(u u_3, u u_4, u v)\to (\gamma_5, 4, 3)$ gives rise to an acyclic edge $(\Delta + 2)$-coloring for $G$.

{Case 2.2.1.2}\ $C(u_4) = \{4, 3, 1, 6^*\}$, and then $\Delta = 5$.
If $c(u_3 u_4) = 6$, that is, $C(u_3) = \{3, 6, 5, 7\}$, then $(u_3 u_4, u u_4, u v)\to (4, 6, 4)$ gives rise to an acyclic edge $(\Delta+2)$-coloring for $G$;
otherwise, $c(u_3 u_4) = 1$.
It follows that $C(u_3) = \{3, 1, 5, 7\}$ and we assume w.l.o.g. that $c(u_4 y_2) = 3$ with $5, 7\in C(y_2)$.
If $2\not\in C(y_2)$, then let $(u_3 u_4, u u_4, u u_1)\to (2, 1, 4)$;
or if $4\not\in C(y_2)$, then let $(u_3 u_4, u u_4, u u_1)\to (4, 1, 4)$;
or otherwise, $C(y_2) = \{2, 3, 4, 5, 7\}$, then let $\{u_4 y_2, u u_3\}\to 1$, $(u_3 u_4, u u_1)\to 4$ and $u u_4\to 3$;
we are done by Lemma \ref{lemma18}.

{Case 2.2.1.3}\ $C(u_4) = \{4, 3, 2, 6^*\}$.
If $c(u_3 u_4) = 2$, that is, $C(u_3) = \{2, 3, 5, 7\}$, then let $(u u_4, u u_1)\to (1, 4)$ and we are done by Lemma \ref{lemma18};
otherwise, $c(u_3 u_4) = 6$ and we assume w.l.o.g. that $c(u_4 y_2) = 3$.
When $C(u_3)= \{3, 6, 5, 7\}$, since $1, 4\not\in C(u_3)$, we have $1\in C(u_2)$;
if $6\not\in C(u_2)$, then $(u_3 u_4, u u_4, u v)\to (4, 6, 4)$ gives rise to an acyclic edge $(\Delta + 2)$-coloring for $G$;
or if $C(u_2) = C\setminus \{3, 4\}$, then let $(u u_1, u u_2, u u_3)\to (2, 3, 1)$ and we are done by Lemma \ref{lemma18}.

In the other case, $C^*_5\setminus C(u_3)\ne\emptyset$ and we assume w.l.o.g. that $\gamma_7\not\in C(u_3)$.
It follows that $3\in B_1$.
One sees that $G$ contains a $(2, \gamma_7)_{(u_2, u_4)}$-path.
If $6\not\in C(u_2)$, then $(u_3 u_4, u u_4, u v)\to (\gamma_7, 6, 4)$ gives rise to an acyclic edge $(\Delta + 2)$-coloring for $G$;
otherwise, $6\in C(u_2)$.
If $3\not\in C(u_2)$ and $G$ contains no $(3, 4)_{(u_2, u_4)}$-path,
then $(u u_3, u u_2, u v)\to (\gamma_7, 3, 2)$ and additionally, if $\gamma_7 = 5$ then $v v_2\to C^*_5\setminus C(v_2)$,
gives rise to an acyclic edge $(\Delta + 2)$-coloring for $G$;
otherwise, $3\in C(u_2)$ or $G$ contains a $(3, 4)_{(u_2, u_4)}$-path, i.e., $\{3, 4\}\cap C(u_2)\ne\emptyset$.
Thus, two possibilities:
\begin{itemize}
\parskip=0pt
\item
	$C(u_2) = C^*_5\cup \{2, 3\}$ and $G$ contains either a $(3, 1)_{(u_3, u_4)}$-path or a $(3, 4)_{(u_3, u_1)}$-path.
    If $G$ contains neither a $(1, 3)_{(u_3, u_2)}$-path nor a $(2, 3)_{(u_3, u_1)}$-path,
    then let $(u u_1, u u_2)\to (2, 1)$ and we are done by Lemma \ref{lemma18};
    otherwise, $G$ contains a $(1, 3)_{(u_3, u_2)}$-path or a $(2, 3)_{(u_3, u_1)}$-path.
    One sees that $(C(u_3)\setminus \{3, 6\})\subseteq \{1, 2, 4\}$, and that $4\not\in C(u_3)$ or $G$ contains a $(4, 3)_{(u_3, u_1)}$-path.
    Then $(u u_2, u u_4, u v)\to (4, 7, 2)$ gives rise to an acyclic edge $(\Delta + 2)$-coloring for $G$.
\item
	$C(u_2) = C^*_5\cup \{4, 2\}$ and $G$ contains a $(4, 3)_{(u_4, u_2)}$-path.
	If $2\not\in C(u_3)$, then let $(u u_1, u u_2)\to (2, 1)$;
	or if $4\not\in C(u_3)$, then let $(u u_1, u u_2, u u_3, u u_4)\to (2, 3, 4, 1)$;
	we are done by Lemma \ref{lemma18}.
	Otherwise, $C(u_3) = \{3, 6, 2, 4\}$ and $G$ contains a $(5, 4)_{(v_1, v_2)}$-path.
	If there exists a $\beta\in C^*_5\setminus C(y_2)$, then $(u u_2, u u_4, u v)\to (3, \beta, 4)$.
	Next, if $\beta\ne 6$, then $uu_3\to \{5, 7\}\setminus \{\beta\}$;
	or if $\beta = 6$, then $(u u_3, u_3 u_4)\to (5, 7)$;
	these give rise to an acyclic edge $(\Delta + 2)$-coloring for $G$.
	Otherwise, $C(y_2) = C^*_5\cup \{3, 4\}$, and then $(u_4 y_2, u u_4, u u_3, u v)\to (1, 3, 7, 4)$ gives rise to an acyclic edge $(\Delta + 2)$-coloring for $G$.
\end{itemize}

{Case 2.2.1.4}\ $C(u_4) = \{4, 3, 2, 1\}$ with $c(u_3 u_4) = 1$.
One sees that $G$ contains a $(2, i_0)_{(u_4, u_2)}$-path.
Then $u_3 u_4\to i_0\in C^*_5\setminus C(u_3)$ (additionally, if $i_0 = 5$ then $v v_2\to C^*_5\setminus C(v_2)$) reduces the proof to the above Case 2.2.1.3.

{Case 2.2.2}\ $3\not\in C(u_4)$.
Proposition~\ref{prop3302}(1) implies that $2\in C(u_4)$ and $G$ contains a $(2, i)_{(u_2, u_4)}$-path for every $i\in C^*_5\setminus C(u_4)$.
W.l.o.g., let $6^*\not\in C(u_4)$.
Proposition~\ref{prop3302}(2) implies that $2\in C(v_1)$ or $G$ contains a $(5, 2)_{(v_1, v_2)}$-path.

{Case 2.2.2.1}\ $c(u_3 u_4) = 2$ and $(C^*_5\setminus C(u_4))\subseteq C(u_3)$.
That is, $C(u_4) = \{4, 2, \gamma_5, \gamma_7\}$, or $C(u_4) = \{4, 2, 7^*, 1\}$.
If $4\not\in C(u_3)$ and $G$ contains no $(4, i)_{(u_3, u_4)}$-path for every $i\in C(u_3)\cap C(u_4)$,
then $(u_3 u_4, u u_4, u v)\to (4, 6, 4)$ gives rise to an acyclic edge $(\Delta+2)$-coloring for $G$;
otherwise, it suffices to assume that $C(u_4) = \{4, 2, 5, 7\}$, $\Delta= 5$,
and that $4\in C(u_3)$ or $G$ contains a $(4, i)_{(u_3, u_4)}$-path for some $\beta^*\in \{5, 7\}$.
If $C(u_3) = \{3, 6, 2, 4\}$, then $(u_3 u_4, u u_4, u v)\to (1, 6, 4)$ gives rise to an acyclic edge $(\Delta + 2)$-coloring for $G$;
or if $C(u_3) = \{3, 6, 2, \beta^*\}$, then let $(u u_4, u u_1)\to (1, 4)$ and we are done by Lemma \ref{lemma18}.

{Case 2.2.2.2}\ $c(u_4 y_2) = 2$ with $C(u_4) = \{4, 2, 1, 7^*\}$ or $C(u_4) = \{4, 2, \gamma_5, \gamma_7\}$.
If $4\not\in C(u_2)$ and $G$ contains no $(4, 3)_{(u_3, u_2)}$-path,
then $(u u_2, u u_4, u v)\to (4, 6, 2)$ gives rise to an acyclic edge $(\Delta + 2)$-coloring for $G$;
or if $1\not\in C(u_4)$ and $G$ contains neither a $(1, 2)_{(u_2, u_4)}$-path nor a $(3, 4)_{(u_3, u_1)}$-path,
then let $(u u_1, u u_4)\to (4, 1)$ and we are done by Lemma \ref{lemma18};
or if $1\not\in C(u_2)$ and $G$ contains none of $(1, 4)_{(u_2, u_4)}$-path, $(1, 3)_{(u_2, u_3)}$-path, and $(3, 2)_{(u_3, u_1)}$-path,
then let $(u u_2, u u_1)\to (1, 2)$ and we are done by Lemma \ref{lemma18}.

In the other case, we assume that there exists a $\gamma^*\in C^*_5\setminus (C(u_3)\cup C(u_4))$.
Then by Proposition~\ref{prop3302}(1), we have $3\in B_1$ and $C(v_1) = C^*_5\cup \{1, 3\}$,
and by Proposition~\ref{prop3302}(3), $G$ contains a $(5, 4)_{(v_1, v_2)}$-path.
Further, if $G$ contains no $(2, 3)_{(u_4, u_2)}$-path, then $(u u_3, u u_4, u v)\to (\gamma^*, 3, 4)$ gives rise to an acyclic edge $(\Delta + 2)$-coloring for $G$.

We continue with the following proposition, to discuss the three possible values of $C(u_4)$:

\begin{proposition}
\label{prop3342}
\begin{itemize}
\parskip=0pt
\item[{(1)}]
	$4\in C(u_2)$, or $G$ contains a $(4, 3)_{(u_3, u_2)}$-path.
\item[{(2)}]
	$1\in C(u_4)$, or $G$ contains a $(1, 2)_{(u_2, u_4)}$-path or a $(3, 4)_{(u_3, u_1)}$-path.
\item[{(3)}]
	$1\in C(u_2)$, or $G$ contains a $(1, 4)_{(u_2, u_4)}$-path, a $(1, 3)_{(u_2, u_3)}$-path or a $(3, 2)_{(u_3, u_1)}$-path.
\item[{(4)}]
	If there exists a $\gamma^*\in C^*_5\setminus (C(u_3)\cup C(u_4))$,
	then $3\in B_1$, $G$ contains a $(5, 4)_{(v_1, v_2)}$-path, a $(2, 3)_{(u_4, u_2)}$-path and $3\in C(y_2)\cap C(u_2)$.
\end{itemize}
\end{proposition}

{Case 2.2.2.2.1}\ $C(u_4) = \{4, 2, 1, 7^*\}$ with $c(u_3 u_4) = 7$.
Two possibilities:
\begin{itemize}
\parskip=0pt
\item
	$(C(u_3)\setminus \{3, 7\})\subseteq C^*_5\setminus \{7\}$.
    If $7\not\in C(u_2)$, then $(u_3 u_4, u u_4, u v)\to (4, 7, 4)$ gives rise to an acyclic edge $(\Delta + 2)$-coloring for $G$;
    or if $3\not\in C(u_4)$, then $(u_3 u_4, u u_3, u u_4, u v)\to (4, 7, 3, 4)$ gives rise to an acyclic edge $(\Delta + 2)$-coloring for $G$;
    or otherwise, $C(u_2)= C\setminus \{1, 4\}$, and then let $(u u_1, u u_2)\to (2, 1)$ and we are done by Lemma \ref{lemma18}.
\item
	There exists a $\gamma_6\in C^*_5\setminus C(u_3)$.
	Then $G$ contains a $(2, i)_{(u_2, u_4)}$-path for every $C^*_5\setminus \{7\}$, a $(2, 3)_{(u_2, u_4)}$-path, and a $(5, 2)_{(v_1, v_2)}$-path
	by Propositions~\ref{prop3302}(1), \ref{prop3342}(4), and \ref{prop3302}(2), respectively.
	If $7\not\in C(u_2)$, then $(u u_2, u u_4, u v)\to (7, \gamma_6, 2)$ gives rise to an acyclic edge $(\Delta + 2)$-coloring for $G$;
	otherwise, $C(u_2) = C^*_5\cup \{2, 3\}$.
	Proposition~\ref{prop3342}(1,3) implies that $4\in C(u_3)$ and $G$ contains a $(4, 3)_{(u_2, u_3)}$-path,
	and $G$ contains a $(3, 1)_{(u_2, u_3)}$-path or a $(3, 2)_{(u_3, u_1)}$-path.
	We first $(u_3 u_4, u u_3, u u_4, u v)\to (3, 6, 7, 4)$;
	next, if $7\not\in C(y_2)$, then we are done;
	or otherwise, $C(y_2) = C^*_5\cup \{2, 3\}$, and $u_4 y_2\to 4$;
	these give rise to an acyclic edge $(\Delta + 2)$-coloring for $G$.
\end{itemize}

{Case 2.2.2.2.2}\ $C(u_4) = \{4, 2, 1, 7^*\}$ with $c(u_3u_4)= 1$.
First, assume that $(C(u_3)\setminus \{1, 3\})\subseteq C^*_5\setminus \{7\}$.
Proposition~\ref{prop3342}(1) implies that $4\in C(u_2)$.
If $3\not\in C(u_2)$, then $(u u_3, u u_4, u v)$ $\to (7, 3, 4)$ gives rise to an acyclic edge $(\Delta + 2)$-coloring for $G$;
otherwise, let $(u_3 u_4, u u_4, u u_1)\to (4, 1, 4)$ and we are done by Lemma \ref{lemma18}.

Next, assume that there exists $\gamma_6\in (C^*_5\setminus \{7\})\setminus C(u_3)$.
Then, by Proposition~\ref{prop3302}(1), $G$ contains a $(2, i)_{(u_2, u_4)}$-path for every $C^*_5\setminus \{7\}$;
by Proposition~\ref{prop3342}(4), $3\in B_1$ and $G$ contains a $(2, 3)_{(u_2, u_4)}$-path;
and by Proposition~\ref{prop3302}(2), $G$ contains a $(5, 2)_{(v_1, v_2)}$-path.
Two possibilities:
\begin{itemize}
\parskip=0pt
\item
	$4, 7\not\in C(u_2)$ and $G$ contains a $(3, 4)_{(u_2, u_3)}$-path.
	If $G$ contains no $(k, \gamma_6)_{(u_2, u_3)}$-path for some $k\in \{4, 7\}$,
	then $(u u_2, u u_3, u v)\to (k, \gamma_6, 2)$ (additionally, if $k= 4$ then $u u_4\to 3$) gives rise to an acyclic edge $(\Delta + 2)$-coloring for $G$;
	otherwise, $C(u_3) = \{3, 1, 4, 7\}$ and $G$ contains an $(i, \gamma_6)_{(u_2, u_3)}$-path for every $i\in \{4, 7\}$.
	We first let $(u_3 u_4, u u_4, u u_1)\to (\gamma_6, 1, 4)$;
	next, if $1\not\in C(y_2)$, then we are done by Lemma \ref{lemma18};
	otherwise, $C(y_2) = C\setminus \{4, 7\}$, and $u_4 y_2\to 4$ gives rise to an acyclic edge $(\Delta + 2)$-coloring for $G$.
\item
	$1\not\in C(u_2)$.
    If $G$ contains no $(2, 3)_{(u_1, u_3)}$-path, then let $(u u_2, u u_4, u u_1)\to (1, \gamma_6, 2)$ and we are done by Lemma \ref{lemma18};
    otherwise, $G$ contains a $(2, 3)_{(u_1, u_3)}$-path and $2\in C(u_3)$.
    If $4\not\in C(u_3)$ and $G$ contains no $(3, 7)_{(u_3, u_4)}$-path,
    then let $(u_3 u_4, u u_3, u u_4, u u_1)\to (3, \gamma_6, 1, 4)$ and we are done by Lemma \ref{lemma18};
    otherwise, (i)\ $4\in C(u_3)$ or (ii)\ $G$ contains a $(3, 7)_{(u_3, u_4)}$-path.
    \begin{description}
	\parskip=0pt
	\item[{\rm (i)}]
		$C(u_3) = \{3, 1, 2, 4\}$ with $c(u_3 x_2) = 2$.
        If $G$ contains no $(2, i)_{(u_3, u_1)}$-path for some $i\in C^*_5$, then $(u_3 u_4, u u_1, u u_2, u u_3)\to (3, 2, 1, i)$;
        next, if $i = 5$, then $v v_2\to C^*_5\setminus C(v_2)$;
        or if $i = 7$, then $u u_4\to 6$;
        we are done by Lemma \ref{lemma18}.
		Otherwise, $C(x_2) = C^*_5\cup \{2, 3\}$, and let $\{u_3 x_2, u u_2\}\to 1$ and $\{u_3 u_4, u u_1, u u_3, u u_4\}\to (6, 4, 2, 3)$,
		and we are done by Lemma \ref{lemma18}.
	\item[{\rm (ii)}]
		We have $C(u_3) = \{3, 1, 2, 7\}$.
        Proposition~\ref{prop3342}(1) implies that $4\in C(u_2)$ and $C(u_2) = C\setminus \{1, 7\}$.
		If $G$ contains no $(7, 6)_{(u_3, u_4)}$-path, then let $(u_3 u_4, u u_4, u u_1)\to (6, 1, 4)$ and we are done by Lemma \ref{lemma18};
        or otherwise, $(u u_2, u u_4, u v)\to (7, 6, 2)$ gives rise to an acyclic edge $(\Delta + 2)$-coloring for $G$.
	\end{description}
\end{itemize}

{Case 2.2.2.2.3}\ $C(u_4) = \{4, 2, \gamma_5, \gamma_7\}$ with $c(u_3 u_4) = \gamma_7$.
If $1\not\in C(u_3)$ and $G$ contains no $(1, i)_{(u_3, u_4)}$-path for every $i\in \{2, \gamma_5\}$,
then $u_3 u_4\to 1$ (additionally, if $\gamma_5 = 5$ then $v v_2\to C^*_5\setminus C(v_2)$) reduces the proof to the above Cases 2.2.2.2.1 and 2.2.2.2.2.
Thus, $1\in C(u_3)$ or $G$ contains a $(1, i)_{(u_3, u_4)}$-path for some $i\in \{2, \gamma_5\}$ and $\{1, 2, \gamma_5\}\cap C(u_3)\ne \emptyset$.
Two possibilities:
\begin{itemize}
\parskip=0pt
\item
	$6\in C(u_3)$ and thus $C(u_3)\in \{\{3, \gamma_7, 6, 1\}, \{3, 6, \gamma_7, 2\}, \{3, 6, \gamma_7, 5\}\}$.
    We have $4\in C(u_2)$ by Proposition~\ref{prop3342}(1).
    If $C(u_3) = \{3, 6, \gamma_7, 2\}$ and $G$ contains $(1, 2)_{(u_3, u_4)}$-path,
    then let $(u u_4, u u_1)\to (1, 4)$ and we are done by Lemma \ref{lemma18}.
    If $C(u_3) = \{3, \gamma_7, 6, 1\}$, then there exists $i\in \{1, 3, \gamma_7\}\setminus C(u_2)$,
    and if $i\in \{1, \gamma_7\}$ (and $vv_2\to C^*_5\setminus C(v_2)$), then let $(u_3 u_4, u u_4, u u_1)\to (4, i, 4)$ and we are done by Lemma \ref{lemma18};
    or if $i = 3$, then $(u u_3, u u_4, u v)\to (\gamma_5, 3, 4)$ gives rise to an acyclic edge $(\Delta + 2)$-coloring for $G$.
    Otherwise, $C(u_3) = \{3, 6, \gamma_7, \gamma_5\}$.
	It follows from Proposition~\ref{prop3342}(1) that $4\in C(u_2)$, and from Proposition~\ref{prop3342}(3) that $1\in C(u_2)$.
    Next, if $3\not\in C(u_2)$, then let $(u u_1, u u_3, u u_4)\to (4, 1, 3)$ and we are done by Lemma \ref{lemma18};
    or otherwise, $(u u_2, u v)\to (7, 2)$ gives rise to an acyclic edge $(\Delta + 2)$-coloring for $G$.
\item
	$6\not\in C(u_3)$.
	Proposition~\ref{prop3342}(4) implies that $G$ contains a $(2, 3)_{(u_2, u_4)}$-path,
	and Proposition~\ref{prop3302}(2) implies that $G$ contains a $(5, i)_{(v_1, v_2)}$-path for every $i\in \{2, 4\}$.
    If $\gamma_7\not\in C(u_2)$, then $(u u_2, u u_4, u v)\to (\gamma_7, 6, 4)$ gives rise to an acyclic edge $(\Delta + 2)$-coloring for $G$;
    otherwise, $\gamma_7\in C(u_2)$.
    Recall that $4\in C(u_2)$ or $G$ contains a $(3, 4)_{(u_2, u_3)}$-path.
    \begin{description}
	\parskip=0pt
	\item[{\rm (i)}]
		$C(u_2) = C\setminus \{1, 5\}$.
        Proposition~\ref{prop3342}(3) implies that $G$ contains a $(3, 4)_{(u_1, u_3)}$-path and $4\in C(u_3)$.
        If $G$ contains neither a $(3, 1)_{(u_2, u_3)}$-path nor a $(3, 2)_{(u_1, u_3)}$-path,
        then let $(u u_1, u u_2)\to (2, 1)$ and we are done by Lemma \ref{lemma18};
        otherwise, $G$ contains a $(3, 1)_{(u_2, u_3)}$-path or a $(3, 2)_{(u_1, u_3)}$-path.
        If $C(u_3)= \{3, 7, 4, 1\}$, then $(u u_1, u u_2, u u_4, u v)\to (2, 5, 1, 4)$ gives rise to an acyclic edge $(\Delta + 2)$-coloring for $G$;
        otherwise, $C(u_3) = \{3, 7, 4, 2\}$ with $c(u_3 x_2) = 2$ and $1, 3\in C(x_2)$.
        Then there exists an $i\in C^*_5\setminus C(x_2)$.
        We first let $(u u_1, u u_2, u u_3, u v)\to (2, 1, i, 3)$;
        next, if $i = \gamma_5$, then let $u u_4\to 6$;
        or otherwise $i = \gamma_7$, and let $u_3 u_4\to 6$;
        these give rise to an acyclic edge $(\Delta + 2)$-coloring for $G$.
	\item[{\rm (ii)}]
		$4\not\in C(u_2)$.
		Proposition~\ref{prop3342}(1--2) implies that $G$ contains a $(3, 4)_{(u_2, u_3)}$-path, a $(1, 2)_{(u_2, u_4)}$-path,
        and thus $C(u_2) = C\setminus \{4, \gamma_5\}$.
        If $C(u_3) = \{3, 7, 4, 1\}$, then $(u u_2, u v)\to (\gamma_5, 2)$ gives rise to an acyclic edge $(\Delta + 2)$-coloring for $G$;
        otherwise, $C(u_3) = \{3, \gamma_7, 4, \gamma_5\}$ and $G$ contains a $(1, \gamma_5)_{(u_3, u_4)}$-path.
        Let $(u u_1, u u_2, u u_3)\to (2, \gamma_5, 1)$ (additionally, if $\gamma_5 = 5$ then $v v_2\to C^*_5\setminus C(v_2)$)
		and we are done by Lemma \ref{lemma18}.
	\end{description}
\end{itemize}

This finishes the inductive step for the case where $G$ contains the configuration ($A_5$).

\subsection{Configuration ($A_6$)}
In this subsection we prove the inductive step for the case where $G$ contains the configuration ($A_6$):

\begin{description}
\parskip=0pt
\item[$(A_6)$]
	A $6$-vertex $u$ adjacent to $u_1$, $u_2$ and four $3$-vertices $u_3$, $u_4$, $u_5$ and $v$, sorted by their degrees.
    At least one of the configurations ($A_{6.1}$)--($A_{6.2}$) occurs.
\end{description}

The same as before, we first not to distinguish the first five neighbors but refer to them as $x_1, x_2, x_3, x_4, x_5$;
we assume w.l.o.g. that $c(u x_i) = i$ for $1\le i\le 5$.
Let $C^*_6 = C\setminus \{1, 2, 3, 4, 5\} = \{i_6, i_7, \ldots, i_{\Delta + 2}\}$.
Note that $|C^*_6|\ge 3$.
From Proposition~\ref{prop3001}, we have $1\le|C(u)\cap C(v)|\le 2$.

\begin{lemma}
\label{lemma21}
If $|C(u)\cap C(v)| = 2$, then in $O(1)$ time either an acyclic edge $(\Delta + 2)$-coloring for $G$ can be obtained,
or an acyclic edge $(\Delta + 2)$-coloring for $H$ can be obtained such that $|C(u)\cap C(v)| = 1$.
\end{lemma}
\begin{proof}
Assume $(v v_1, v v_2)_c = (a, b)\subseteq C(u)$.

If there exists an $i\in C^*_6\setminus C(v_1)$, then let $v v_1\to i$, resulting in $|C(u)\cap C(v)| = 1$.
Otherwise, $C^*_6\subseteq C(v_1)\cap C(v_2)$.
If for every $l\in C(u)\backslash \{a, b\}$, $G$ contains no $(k, l)_{(u_a, u_l)}$-path for some $k\in C^*_6\setminus C(u_a)$,
then let $uu_a\to k$, also resulting in $|C(u)\cap C(v)| = 1$.
We proceed with the following proposition:
\begin{proposition}
\label{prop3401}
	For each $i\in \{a, b\}$, $C^*_6\subseteq C(u_i)$ or for every $k\in C^*_6\setminus C(u_i)$,
    $G$ contains a $(k, l)_{(u_i, u_l)}$-path for some $l\in C(u)\backslash \{a, b\}$.
\end{proposition}

Proposition~\ref{prop3401} implies that $C(v)\setminus S_3\ne\emptyset$.
One sees that $C^*_6\subseteq C(v_2)\cap C(v_1)$;
we consider the following three cases.

{Case 1.}\ $(v v_1, v v_2)_c = (1, 2)$ for ($A_{6.1}$).
If $G$ contains no $(i, 2)_{(v_1, v_2)}$-path for some $i\in \{3, 4, 5\}\setminus C(v_1)$,
then let $v v_1\to i$ and we are done by Proposition~\ref{prop3401};
otherwise, $G$ contains an $(i, 2)_{(v_1, v_2)}$-path for every $i\in \{3, 4, 5\}\setminus C(v_1)$.
It follows that $C(v_1) = C^*_6\cup \{1, 2, 3\}$ and $C(v_2) = C^*_6\cup \{2, 4, 5\}$.
By Proposition~\ref{prop3401}, $C(u_2)\cap \{3, 4, 5\}\ne\emptyset$.
Note that $(C^*_6\setminus C(u_2))\subseteq B_1$.
Further we have $(C^*_6\setminus C(u_2))\subseteq B_i$ if letting $(v v_1, v v_2)\to (j, 1)$ for $j\in \{4, 5\}$.
It implies that $\Delta = 6$, $C(u_2) = \{2, 3, 6\}$, $C(u_4) =\{4, 7, 8\}$, and $C(u_5) = \{5, 7, 8\}$.
Then $(u u_4, u v)\to (6, 4)$ gives rise to an acyclic edge $(\Delta + 2)$-coloring for $G$.

{\bf Case 2.}\ $(v v_1, v v_2)_c = (1, 3)$ for ($A_{6.2}$).
By Proposition~\ref{prop3401}, $C(u_3)\cap \{2, 4, 5\}\ne \emptyset$ and we assume w.l.o.g. that $7, 8\in B_1\setminus C(u_3)$.
If $G$ contains no $(i, 3)_{(v_1, v_2)}$-path for some $i\in \{4, 5\}\setminus C(v_1)$,
then let $v v_1\to i$ and we are done by Proposition~\ref{prop3401};
otherwise, $\{4, 5\}\subseteq C(v_1)$ or $G$ contains an $(i, 3)_{(v_1, v_2)}$-path for every $i\in \{4, 5\}\setminus C(v_1)$.

{Case 2.1.}\ $C(v_1) = C^*_6\cup \{1, 4, 5\}$.
If there exists an $i\in \{4, 5\}\setminus C(v_2)$, then let $(v v_1, v v_2)\to (3, i)$ and we are done by Proposition~\ref{prop3401};
otherwise, $C(v_2) = C^*_6\cup \{3, 4, 5\}$, and $(v v_1, v v_2, u v)\to (3, 1, 7)$ gives rise to an acyclic edge $(\Delta + 2)$-coloring for $G$.

{Case 2.2.}\ $3\in C(v_1)$ and $G$ contains an $(i, 3)_{(v_1, v_2)}$-path for every $i\in \{4, 5\}\setminus C(v_1)$.
It follows that $\{4, 5\}\setminus C(v_1)\ne \emptyset$ and we assume w.l.o.g. that $5\not\in C(v_1)$.
Two possibilities:
\begin{itemize}
\parskip=0pt
\item
	If $1\not\in C(v_2)$, then we first let $(v v_1, v v_2)\to (5, 1)$;
    next, if there is a $k \in \{7, 8\} \setminus B_5$, then $u v\to k$;
    or otherwise, $(u u_5, u v)\to (6, 7)$;
    these give rise to an acyclic edge $(\Delta + 2)$-coloring for $G$.
\item
	If $1\in C(v_2)$, then we have $C(v_1) = C^*_6\cup \{1, 3, 4\}$ and $C(v_2) = C^*_6\cup \{3, 1, 5\}$.
    Let $v v_1\to 2$ and we are done by a similar discussion as in the last paragraph since $2\not\in C(v_2)$.
\end{itemize}

{Case 3.}\ $(v v_1, v v_2)_c = (1, 2)$ for ($A_{6.2}$).
If $1\not\in C(v_2)$, then $v v_2 \to \{3, 4, 5\}\setminus C(v_2)$ reduces the discussion to the above Case 2;
otherwise, we have $1\in C(v_2)$, $2\in C(v_1)$, and there exists a $j_1\in \{3, 4, 5\}\setminus (C(v_1)\cup C(v_2))$.
Then $v v_2 \to j_1$ reduces the discussion again to the above Case 2.

This finishes the proof of the lemma.
\end{proof}

By Lemma \ref{lemma21}, we assume in the sequel that $|C(u)\cap C(v)| = 1$, and further w.l.o.g. that $(v v_1, v v_2)_c = (1, 6)$.
It follows from Proposition~\ref{prop3001} that $(C\setminus \{1, 2, 3, 4, 5, 6\}) \subseteq C(x_1)\cap C(v_1)$.

Let $i_0\in \{2, 3, 4, 5\}\setminus B_1$.
If there exists a $j\in (C^*_6\cap B_1)\setminus C(x_{i_0})$ such that there is no $(i, j)_{(x_{i_0}, x_i)}$-path for every $i\in \{2, 3, 4, 5\}\setminus \{i_0\}$,
then $(u x_{i_0}, u v)\to (j, i_0)$ gives rise to an acyclic edge $(\Delta + 2)$-coloring for $G$.
Let
\[
S_u = \biguplus\limits_{i\in \{2, 3, 4, 5\}}(C(x_i)\setminus \{c(ux_i)\}).
\]
Assume that $2, 3, 4, 5\not\in B_1$ and mult$_{S_u}(j_0)\le 1$ for some $j_0\in (C^*_6\cap B_1)$.
We assume w.l.o.g. that $j_0\not\in C(x_{y_1})\cup C(x_{y_2})\cup C(x_{y_3})$, where $\{y_1, y_2, y_3, y_4\} = \{2, 3, 4, 5\}$.
If $G$ contains no $(j_0, y_4)_{(x_{y_1}, x_{y_4})}$-path, then $(u x_{y_1}, u v)\to (j_0, y_1)$;
or otherwise, $(u x_{y_2}, u v)\to (j_0, y_2)$;
this gives rises to an acyclic edge $(\Delta + 2)$-coloring for $G$.

We proceed with the following proposition:

\begin{proposition}
\label{prop3402}
\begin{itemize}
\parskip=0pt
\item[{(1)}]
	For each $i_0\in \{2, 3, 4, 5\}\setminus B_1$, we have $(C^*_6\cap B_1)\subseteq C(x_{i_0})$,
	or for every $j\in (C^*_6\cap B_1)\setminus C(u_{x_0})$,
	$G$ contains an $(i, j)_{(x_{i_0}, x_i)}$-path for some $i\in (\{1, 2, 3, 4\}\cap C(x_{i_0})) \setminus\{i_0\}$.
\item[{(2)}]
	If $2, 3, 4, 5\not\in B_1$, then mult$_{S_u}(j)\ge 2$ for every $j\in C^*_6\cap B_1$.
\end{itemize}
\end{proposition}

\paragraph{Case ($A_6$-1).}\ $c(v v_1)\in S_3$.

In this case, we assume w.l.o.g. that $c(u u_i) = i$ for $i\in \{2, 4, 5\}$, and $(u u_1, u u_3)_c = (3, 1)$.
It follows that $\Delta = d(u) = 6$ and $C(u_3) = \{1, 7, 8\}$.
One sees that $2, 3, 4, 5\not\in B_3 = \{7, 8\}$.

{Case 1.1.}\ ($A_{6.1}$) holds, i.e., $d(u_2) = 3$.
Proposition~\ref{prop3402} implies that $C(u_k) = \{k, a_k, k + 4\}$, where $a_k\in C(u)\setminus \{1\}$ for every $k\in \{2, 4, 5\}$,
and, since $6\not\in C(u_4)\cup C(u_5)$, $\{a_4, a_5\} = \{2, 3\}$.
We assume w.l.o.g. that $a_4 = 2$.
Then $(u u_4, u v)\to (8, 4)$ gives rise to an acyclic edge $(\Delta + 2)$-coloring for $G$.

{Case 1.2.}\ ($A_{6.2}$) holds, i.e., $d(u_1) = d(u_2) = 4$ and $u_1 u_2\in E(G)$.
If $C(u_i) = \{i, 6, 7, 8\}$ for some $i\in \{3, 2\}$,
then $(u u_i, u u_3, u v)\to (1, i, 7)$ gives rise to an acyclic edge $(\Delta + 2)$-coloring for $G$;
otherwise, $C^*_6\setminus C(u_i)\ne \emptyset$ for every $i\in \{3, 2\}$.
Proposition~\ref{prop3402} implies that $c(u_1 u_2) = a_{12}\in \{4, 5\}$ and $1\not\in \biguplus\limits_{i\in \{2, 3, 4, 5\}}C(u_i)$.
Then $(u u_1, u u_3, u v)\to (1, 3, 7)$ gives rise to an acyclic edge $(\Delta + 2)$-coloring for $G$.

\paragraph{Case ($A_6$-2).}\ $c(v v_1)\not\in S_3$.

In this case, we assume w.l.o.g. that $c(u u_i) = i$ for $1\le i\le 5$.

{Case 2.1.}\ $6\not\in B_1$.
If there is no $(1, i)_{(v_1, v_2)}$-path for some $i\in C\setminus (C(v_2)\cup \{1\})$,
then $(v v_2, u v)\to (i, 6)$ gives rise to an acyclic edge $(\Delta + 2)$-coloring for $G$;
otherwise, $G$ contains a $(1, i)_{(v_1, v_2)}$-path for every $i\in C(u)\setminus (C(v_2)\cup \{1\})$.
It follows that $\{1\}\cup C^*_6 \subseteq C(v_2)$, $C(u)\subseteq C(v_1)\cup C(v_2)$, and $\{3, 4, 5\}\setminus C(v_2) \ne\emptyset$.
If $6\not\in C(v_1)$, then $v v_1\to 6$ and $v v_2\to \{3, 4, 5\}\setminus C(v_2)$ reduce the discussion to the above Case ($A_6$-1);
otherwise, $6\in C(v_1)$.

It suffices to assume that $C(v_1) = C^*_6\cup \{1, \gamma_2, \gamma_3\}$,
$C(v_2) = C^*_6\cup \{1, \gamma_4, \gamma_5\}$, where $\{\gamma_2, \gamma_3, \gamma_4, \gamma_5\} = \{2, 3, 4, 5\}$,
and $G$ contains a $(1, i)_{(v_1, v_2)}$-path for every $i\in \{\gamma_2, \gamma_3\}$.
One sees that $2, 3, 4, 5\not\in B_1$.
Proposition~\ref{prop3402} implies that mult$_{S_u}(j)\ge 2$ for every $j\in C^*_6\setminus \{6\}$.

For ($A_{6.1}$), we have $1, 6\not\in S_u$, and $(u u_5, u v)\to (6, 5)$ gives rise to an acyclic edge $(\Delta + 2)$-coloring for $G$.
For ($A_{6.2}$), we have $C(u_3) = \{3, b_3, \gamma\}$, $C(u_4) = \{4, b_4, 7\}$, and $C(u_5) = \{5, b_5, 8\}$, where $b_3, b_4, b_5\in C(u)\setminus \{1\}$.
We can obtain an acyclic edge $(\Delta + 2)$-coloring of $G$ as follows:
\begin{itemize}
\parskip=0pt
\item
	If $b_5\in \{3, 4\}$, then let $u u_5\to \{6, 7, 8\}\setminus (C(u_{b_5})\cup \{8\})$ and $u v\to 5$;
\item
	or if $b_4\in \{3, 5\}$, then let $u u_4\to \{6, 7, 8\}\setminus (C(u_{b_4})\cup \{7\})$ and $u v\to 4$;
\item
	or we have $b_4 = b_5 = 2$;
    if there is no $(2, 6)_{(u_2, u_5)}$-path, then let $(u u_5, u v)\to (6, 5)$;
    or otherwise, let $(u u_4, u v)\to (6, 4)$.
\end{itemize}

{Case 2.2.}\ $6\not\in B_1$.
For ($A_{6.2}$), we have $\Delta= 6$ and $C(u_1) = \{1\}\cup \{6, 7, 8\}$ with $c(u_1 u_2) = i_6$, $1\in C(u_2)$.
One sees that $2, 3, 4, 5\not\in B_1$.
From Proposition~\ref{prop3402}, we may further assume that $C(u_2) = \{1, 2, i_6, 3\}$ and $C(u_3) = \{3, i_7, i_8\}$.
Then, $(u u_3, u v)\to (i_6, 3)$ gives rise to an acyclic edge $(\Delta + 2)$-coloring for $G$.

For ($A_{6.1}$), we assume w.l.o.g. that $C(u_1) = C^*_6\cup \{1, 2, 3\}$.
One sees that $4, 5\not\in B_1$.
By Proposition~\ref{prop3402}, we assume w.l.o.g. that $C(u_4) = \{2, 4, \rho_4\}$, $C(u_5) = \{3, 5, \rho_5\}$.
When $\rho_5\in \{1, 4, 2\}$, we can obtain an acyclic edge $(\Delta + 2)$-coloring of $G$ as follows:
\begin{itemize}
\parskip=0pt
\item
	If $\rho_5 = 1$, then let $u u_5\to C^*_6\setminus C(u_3)$ and $u v\to 5$;
\item
	or if $\rho_5 = 4$, then by Proposition~\ref{prop3402}, we have $C(u_4) = \{2, 4, i_6\}$, $C(u_3)\setminus \{3\} = C(u_2)\setminus \{2\} = \{i_7, i_8\}$,
    and we let $(u u_1, u u_5, u v)\to (5, 1, 7)$;
\item
	or if $\rho_5 = 2$, then by Proposition~\ref{prop3402}, we assume w.l.o.g. that $C(u_3) = \{3, i_6, i_7\}$, $i_8\in C(u_2)$,
    and we let $(u u_3, u v)\to (8, 5)$ and $u u_5\to \{i_6, i_7\}\setminus C(u_2)$.
\end{itemize}

If $\rho_4\in \{1, 5, 3\}$, then we are done by a similar discussion as for $\rho_5$.
Hence, we have $C(u_5) = \{5, 3, i_6\}$, $C(u_3) = \{3, i_7, i_8\}$ and $C(u_4) = \{4, 2, j_6\}$, $C(u_2) = \{2, j_7, j_8\}$, where $\{j_6, j_7, j_8\} = \{6, 7, 8\}$.
Then $(u u_1, u u_5, u v)\to (5, 1, 7)$ gives rise to an acyclic edge $(\Delta + 2)$-coloring for $G$.

This finishes the inductive step for the case where $G$ contains the configuration ($A_6$), and completes the proof of Theorem~\ref{thm01}.

\begin{thebibliography}{10}

\bibitem{AMR91}
N.~Alon, C.~Mcdiarmid, and B.~Reed.
\newblock Acyclic coloring of graphs.
\newblock {\em Random Structures \& Algorithms}, 2:277--288, 1991.

\bibitem{ABZ01}
N.~Alon, B.~Sudakov, and A.~Zaks.
\newblock Acyclic edge colorings of graphs.
\newblock {\em Journal of Graph Theory}, 37:157--167, 2001.

\bibitem{AMM12}
L.~D. Andersen, E.~M{\'a}{\v{c}}ajov{\'a}, and J.~Maz{\'a}k.
\newblock Optimal acyclic edge-coloring of cubic graphs.
\newblock {\em Journal of Graph Theory}, 71:353--364, 2012.

\bibitem{BC09}
M.~Basavaraju and L.~S. Chandran.
\newblock Acyclic edge coloring of graphs with maximum degree $4$.
\newblock {\em Journal of Graph Theory}, 61:192--209, 2009.

\bibitem{BLSNHT11}
M.~Basavaraju, L.~S. Chandran, N.~Cohen, F.~Havet, and T.~M{\"u}ller.
\newblock Acyclic edge-coloring of planar graphs.
\newblock {\em SIAM Journal on Discrete Mathematics}, 25:463--478, 2011.

\bibitem{EP12}
L.~Esperet and A.~Parreau.
\newblock Acyclic edge-coloring using entropy compression.
\newblock {\em European Journal of Combinatorics}, 34:1019--1027, 2013.

\bibitem{F78}
J.~Fiamcik.
\newblock The acyclic chromatic class of a graph.
\newblock {\em Mathematica Slovaca}, 28:139--145, 1978.

\bibitem{GKP17}
I.~Giotis, L.~Kirousis, K.~I. Psaromiligkos, and D.~M. Thilikos.
\newblock Acyclic edge coloring through the lov{\'a}sz local lemma.
\newblock {\em Theoretical Computer Science}, 665:40--50, 2017.

\bibitem{HRW11}
J.~Hou, N.~Roussel, and J.~Wu.
\newblock Acyclic chromatic index of planar graphs with triangles.
\newblock {\em Information Processing Letters}, 111:836--840, 2011.

\bibitem{MR98}
M.~Molloy and B.~Reed.
\newblock Further algorithmic aspects of the local lemma.
\newblock In {\em Proceedings of the 30th Annual ACM Symposium on Theory of
  Computing (STOC 1998)}, pages 524--529, 1998.

\bibitem{NPS12}
S.~Ndreca, A.~Procacci, and B.~Scoppola.
\newblock Improved bounds on coloring of graphs.
\newblock {\em European Journal of Combinatorics}, 33:592--609, 2012.

\bibitem{SWW12}
Q.~Shu, W.~Wang, and Y.~Wang.
\newblock Acyclic edge coloring of planar graphs without $5$-cycles.
\newblock {\em Discrete Applied Mathematics}, 160:1211--1223, 2012.

\bibitem{SW11}
Q.~Shu, W.~Wang, and Y.~Wang.
\newblock Acyclic chromatic indices of planar graphs with girth at least $4$.
\newblock {\em Journal of Graph Theory}, 73:386--399, 2013.

\bibitem{SWMW19}
Q.~Shu, Y.~Wang, Y.~Ma, and W.~Wang.
\newblock Acyclic edge coloring of $4$-regular graphs without 3-cycles.
\newblock {\em Bulletin of the Malaysian Mathematical Sciences Society},
  42:285--296, 2019.

\bibitem{Sku04}
S.~Skulrattankulchai.
\newblock Acyclic colorings of subcubic graphs.
\newblock {\em Information Processing Letters}, 92:161--167, 2004.

\bibitem{Vizing64}
V.~G. Vizing.
\newblock On an estimate of the chromatic class of a $p$-graph [in {R}ussian].
\newblock {\em Discret. Analiz. (discontinued)}, 3:25--30, 1964.

\bibitem{WZ14}
T.~Wang and Y.~Zhang.
\newblock Acyclic edge coloring of graphs.
\newblock {\em Discrete Applied Mathematics}, 167:290--303, 2014.

\bibitem{WZ16}
T.~Wang and Y.~Zhang.
\newblock Further result on acyclic chromatic index of planar graphs.
\newblock {\em Discrete Applied Mathematics}, 201:228--247, 2016.

\bibitem{WMSW19}
W.~Wang, Y.~Ma, Q.~Shu, and Y.~Wang.
\newblock Acyclic edge coloring of $4$-regular graphs (ii).
\newblock {\em Bulletin of the Malaysian Mathematical Sciences Society},
  42:2047--2054, 2019.

\bibitem{WSW4}
W.~Wang, Q.~Shu, and Y.~Wang.
\newblock Acyclic edge coloring of planar graphs without $4$-cycles.
\newblock {\em Journal of Combinatorial Optimization}, 25:562--586, 2013.

\bibitem{WSW127}
W.~Wang, Q.~Shu, and Y.~Wang.
\newblock A new upper bound on the acyclic chromatic indices of planar graphs.
\newblock {\em European Journal of Combinatorics}, 34:338--354, 2013.

\bibitem{WSWZ14}
Y.~Wang, Q.~Shu, J.-L. Wu, and W.~Zhang.
\newblock Acyclic edge coloring of planar graphs without a $3$-cycle adjacent
  to a $6$-cycle.
\newblock {\em Journal of Combinatorial Optimization}, 28:692--715, 2014.

\end{thebibliography}
\end{document}